\documentclass[preprint,1p,10pt,authoryear]{elsarticle}
\usepackage{amsthm}
\usepackage{amsfonts}
\usepackage{dsfont}
\usepackage{amsmath}
\usepackage{natbib}
\usepackage[colorlinks,citecolor=blue,urlcolor=blue,filecolor=blue,backref=page]{hyperref}
\usepackage{graphicx}
\usepackage{xargs}
\usepackage{enumerate}
\DeclareMathAlphabet{\pazocal}{OMS}{zplm}{m}{n}
\usepackage{epstopdf}
\usepackage{aliascnt}
\usepackage{tikz}
\usepackage{multirow}
\usepackage{algorithm}
\usepackage{algpseudocode}

\newtheorem{prop}{Proposition}
\newtheorem{coro}{Corollary}
\newtheorem{lemm}{Lemma}
\newtheorem{rem}{Remark}

\newtheorem{exa}{Example}

\newcounter{assumption}
\newenvironment{assumption}[1][]{\refstepcounter{assumption}\par\medskip
   \textbf{A\hspace{-0.15cm}~\theassumption.\;#1} \rmfamily}{\smallskip}

\def\1{\mathds{1}}

\def\A{\texttt{a}}
\newcommandx\ABC[2][1=]{
\ifthenelse{\equal{#1}{}}
{
\ifthenelse{\equal{#2}{}}{(\tparam,U)}{(\tparam^{(#2)},U^{(#2)})}
}
{
\ifthenelse{\equal{#2}{}}{(\tparam^{#2}_{#1},U^{#2}_{#1})}{(\tparam^{#2}_{#1},U^{#2}_{#1})}
}
}

\def\barK{\bar{K}}

\def\ber{\mathrm{Bernoulli}}
\def\borel{\mathcal{B}}

\def\bS{\bar{S}}
\def\bQ{\overline{Q}}

\def\bigO{\mathcal{O}}

\def\bvarrho{\overline{\varrho}}
\def\bDelta{\bar{\Delta}}

\newcommandx\coord[3][1=1, 3=n]{(#2_{#1},\ldots,#2_{#3})}
\newcommandx\coor[3][1=1, 3=n]{#2_{#1},\ldots,#2_{#3}}

\def\cov{\text{cov}}

\def\diag{\mathrm{diag}}

\def\esp{\mathbb{E}}
\def\eg{\textit{e.g.}\,}
\def\eps{\epsilon}
\def\event{\mathcal{E}}

\def\fam{\mathsf{F}}

\def\grad{\nabla}

\newcommand{\KL}{\mathrm{KL}}

\def\ie{\textit{i.e}\;}
\def\iid{\textit{i.i.d.}\,}

\def\hpi{\hat{\pi}}

\newcommandx\MH[2][1=]{
\ifthenelse{\equal{#1}{}}
{
\ifthenelse{\equal{#2}{}}{\param}{\param^{(#2)}}
}
{
\ifthenelse{\equal{#2}{}}{\param^{#2}_{#1}}{\param^{(#2)}_{#1}}
}
}

\def\meas{\mathcal{M}}

\def\norm{\pazocal{N}}
\def\nset{\mathbb{N}}
\def\Ntest{N_{\text{test}}}

\def\proba{\mathbb{P}}

\def\param{\theta}
\def\paramst{\theta^{\ast}}
\def\paramalg{\borel(\paramset)}
\def\paramset{\Theta}
\def\proba{\mathbb{P}}
\newcommand{\pscal} [2]{\left\langle #1 , #2\right\rangle}
\def\plim{\mathrm{plim}}

\def\rmd{\mathrm{d}}
\def\rset{\mathbb{R}}

\def\suffset{\mathsf{S}}

\def\tparam{\tilde{\theta}}
\def\targ{\pi}
\def\T{^{T}}
\def\TV{\mathrm{TV}}
\def\talpha{\tilde{\alpha}}

\def\tparam{\tilde{\param}}

\def\tK{\tilde{K}}
\def\ttarg{\tilde{\pi}}

\def\tf{\tilde{f}}
\def\thetaMLE{\theta^{\ast}}
\def\ta{\tilde{a}}
\def\ttheta{\tilde{\theta}}
\def\tpi{\tilde{\pi}}
\def\tn{|\!|\!|}
\def\tra{\text{tr}}
\def\trho{\tilde{\rho}}

\def\tA{\tilde{A}}

\def\U{U}
\def\Uset{\mathsf{\U}}
\def\Usetst{\mathsf{\U}^{\star}}
\def\Ualg{\mathcal{\U}}
\def\unif{\text{unif}}

\def\Vset{\bar{\Uset}}
\def\var{\text{var}}

\def\wKL{\widehat{\KL}}

\def\X{X}
\def\Xset{\mathsf{\X}}
\def\Xalg{\mathcal{\X}}

\def\Y{Y}
\def\Yset{\mathsf{\Y}}

\def\tZ{\tilde{Z}}

\newtheorem{theorem}{Theorem}
\newtheorem{lemma}[theorem]{Lemma}

\begin{document}

\begin{frontmatter}
\title{Informed Sub-Sampling MCMC: Approximate Bayesian Inference for Large Datasets}

\author[ucd,insight]{Florian Maire\corref{corresp}}
\cortext[corresp]{Corresponding author}
\ead{florian.maire@ucd.ie}

\author[ucd,insight]{Nial Friel}
\author[ensae]{Pierre Alquier}
\address[ucd]{School of Mathematics and Statistics, University College Dublin}
\address[insight]{Insight Centre for Data Analytics, University College Dublin}
\address[ensae]{CREST, ENSAE, Universit\'{e} Paris Saclay}

\begin{abstract}
This paper introduces a framework for speeding up Bayesian inference conducted in presence of large datasets. We design a Markov chain whose transition kernel uses an {unknown} fraction of {fixed size} of 
the available data that is randomly refreshed throughout the algorithm. Inspired by the Approximate Bayesian Computation (ABC) literature, the subsampling process is guided by the fidelity to the observed data, 
as measured by summary statistics. The resulting algorithm, Informed Sub-Sampling MCMC (ISS-MCMC), is a generic and flexible approach which, contrary to existing scalable methodologies, preserves the simplicity of the 
Metropolis-Hastings algorithm. Even though exactness is lost, \ie the chain distribution approximates the posterior, we study and quantify theoretically this bias and show on a diverse set of examples that it 
yields excellent performances when the computational budget is limited. If available and cheap to compute, we show that setting the summary statistics as the maximum likelihood estimator is supported by theoretical 
arguments.
\end{abstract}

\begin{keyword}
{Bayesian inference \sep Big-data \sep Approximate Bayesian Computation  \sep noisy Markov chain Monte Carlo}

Primary: 65C40, 65C60 -- Secondary: 62F15

\end{keyword}

\end{frontmatter}

\section{Introduction}
The development of statistical methodology that scale to large datasets represents a significant research frontier in modern statistics. This paper presents a generic and flexible approach to directly address this challenge when a Bayesian strategy is followed. Given a set of observed data $\coord{Y}[N]$, a specified prior distribution $p$ and a likelihood function $f$, estimating parameters $\param\in\paramset$ of the model proceeds via exploration of the posterior distribution $\targ$ defined on $(\paramset,\borel(\paramset))$ by
\begin{equation}
\label{eq:posterior}
\targ(\rmd\param\,|\,\coor{Y}[N])\propto f(\coor{Y}[N]\,|\,\param)p(\rmd\param)\,.
\end{equation}
Stochastic computation methods such as Monte Carlo methods allow one to estimate characteristics of $\targ$. In Bayesian inference, Markov chain Monte Carlo (MCMC) methods remain the most widely used strategy. Paradoxically, improvements in data acquisition technologies together with increased storage capacities, present a new challenge for these methods. Indeed, the size of the data set $N$ (along with the dimension of each observation) can become so large, that even a routine likelihood evaluation is made prohibitively computationally intensive. As a consequence, MCMC methods such as the Metropolis-Hastings algorithm \citep{metropolis1953equation} cannot be considered for reasonable runtime. This issue has recently generated a lot of research activity, see \cite{bardenet2015markov} for a comprehensive review.

Most of the scalable MCMC methods proposed in the literature are based on approximations of the Metropolis-Hastings (M-H) algorithm. In the sequel, we will refer to as \textit{exact approximations}, algorithms that 
produce samples from the target distribution when the chain is in the stationary regime, as opposed to \textit{approximate} methods that do not. Central to those scalable MCMC approaches is the idea that only the 
calculation of the likelihood of a subset of data would be required to simulate a new state of the Markov chain. Following the development of pseudo-marginal algorithms \citep{andrieu2009pseudo,andrieu2012convergence}, 
a first direction has been to replace the likelihoods in the M-H acceptance ratio by positive unbiased estimators (based on a subset of data). Although appealing since exact, this approach remains (for now) mostly 
theoretical because such estimators are in general not available \citep{jacob2015nonnegative}. Attempts to circumvent the positivity and unbiasedness requirements of the estimator have been studied in 
\cite{quiroz2016exact} and \cite{quiroz2015speeding}, respectively. In both cases, the authors resort to sophisticated control variates, which can be computationally expensive to compute. 

Other authors have proposed to approximate the log-likelihood ratio by subsampling data points (\cite{korattikara2013austerity,bardenet2014towards,bardenet2015markov}), the objective being to mimic the accept/reject decision that would be achieved by the Metropolis-Hastings algorithm. Even though the resulting algorithms are not exact, the \textit{Confidence sampler} proposed in \cite{bardenet2014towards} and refined in \cite{bardenet2015markov} is designed such that the accept/reject decision is, with an arbitrarily high probability, identical to that taken by the Metropolis-Hastings algorithm. The construction of this algorithm, based on concentration inequalities, allows to bound the L1 distance between the stationary distribution of the algorithm and $\pi$. The price to pay is that the number of likelihood evaluations is not fixed but adaptively set by the algorithm at each iteration and, as noted in \cite{bardenet2015markov}, it is of order $\bigO(N)$ when the chain reaches equilibrium. This number can be brought down if an accurate proxy of the log-likelihood ratio, acting as control variates, is available, as demonstrated in \cite{bardenet2015markov}.

More recently, a stream of research has shed light on the use of continuous time Markov processes (Zig-Zag process, Langevin diffusion) to perform Bayesian analysis of tall dataset \citep{bierkens2016zig,pollock2016scalable,fearnhead2016piecewise}. The computational bottleneck for this class of methods is the calculation of the gradient of the log-likelihood and it has been shown that provided that an unbiased estimate of this gradient is used, they remain exact. Here again, the use of control variates to reduce the variance of the estimator is in practice essential to reach the full potential of these methods. However, we note that those approaches represent a significant departure from the M-H algorithm and as such lose its implementational simplicity.

In this paper, we propose Informed Sub-Sampling MCMC (ISS-MCMC), a novel methodology which aims to make the best use of a computational resource available for a given computational run-time, while still preserving the celebrated 
simplicity of the standard M-H sampler. The state space $\Theta$ is extended with an $n$-dimensional vector of unique integers $U_k\subset\{1,\ldots,N\}$ identifying a subset of the data used by the Markov transition kernel at the the $k$-th iteration of the algorithm, where $n\ll N$ is set according to the available computational budget. Central to our approach is the fact that each subset is weighted according to a \textit{similarity measure} with respect to the full set of data through summary statistics, in the spirit of Approximate Bayesian Computation (ABC) \cite[see \eg][]{marin2012approximate}. The subset variable is randomly refreshed at each iteration according to the similarity measure. The Markov chain transition kernel only uses a fraction $n/N$ of the available data which is by construction --and contrary to \cite{maclaurin2014firefly}, \cite{korattikara2013austerity} and \cite{bardenet2014towards}-- held constant throughout the algorithm.
Moreover, unlike most of the papers mentioned before, our method can be applied to virtually any model (involving \iid data or not), as it does not require any assumption on the likelihood function nor on the prior distribution. Our algorithm can be cast as a \textit{noisy} MCMC method since the marginal in $\theta$ of our Markov chain targets an approximation of $\targ$ that we quantify using the framework established in \cite{alquier2014noisy}. In the special case where the data are \iid realizations from an exponential model, we prove that when the summary statistics is set as the sufficient statistics, this yields an optimal approximation, in the sense of minimizing an upper bound of the Kullback-Leibler (KL) divergence between $\targ$ and the marginal target of our method. In the general case, we show that setting the summary statistics as the maximum likelihood estimator allows to bound the approximation error (in L1 distance) of our algorithm. We connect our work to a number of recent papers including \cite{rudolf2015perturbation,huggins2016,dalalyan2017further} that bound approximation error of MCMC algorithms, using the Wasserstein metric.

To summarize, the main contribution of our work is to show that, under verifiable conditions, it is possible to infer $\pi$  through a scalable approximation of the M-H algorithm where the computational budget of each iteration is fixed (through the subset size $n$). To do so, it is necessary to draw the subsets according to a similarity measure with respect to the full data set and not uniformly at random, as previously explored in the literature. We show that setting the similarity measure as the squared L2 distance between the full dataset and subsample maximum likelihood estimators is supported by theoretical arguments.

Section \ref{sec:first_ex} presents a striking real data example which we hope will help the reader to understand the problem we address and motivate the solution we propose, without going into further technical details at this stage. In Section \ref{sec:expo}, we provide theoretical results concerning exponential-family models, which we illustrate through a probit example. This section allows us to justify our motivations supporting the Informed Sub-Sampling general methodology which is rigorously presented in Section \ref{sec:alg}. In Section \ref{sec:conv}, we study the transition kernel of our algorithm and show that it yields a Markov chain targeting, marginally, an approximation of $\targ$. The approximation error is quantified and we provide theoretical justifications for setting up the Informed Sub-Sampling tuning parameters, including the choice of summary statistics. Finally, in Section \ref{sec:sim}, our method is used to estimate parameters of an autoregressive time series and a logistic regression model. It is also illustrated to perform a binary classification task. In the latter example, we compare the performance of our algorithm with the SubLikelihoods approach proposed in \cite{bardenet2014towards}. 

\section{An introductory example}
\label{sec:first_ex}
We showcase the principles of our approach on a first real data example. The problem at hand is to infer some template shapes of handwritten digits from the MNIST database
(\href{http://yann.lecun.com/exdb/mnist/}{http://yann.lecun.com/exdb/mnist/}).

\begin{exa}
\label{ex1}
The data $Y_1,Y_2,\ldots$ are modelled by a deformable template model \citep{allassonniere2007towards}. Each data $Y_i$ is a $15\times15$ pixel image representing an handwritten digit whose conditional distribution given its class $J(i)\in(0,1,\ldots,9)$ is a random deformation of the template shape, parameterized by a $d=256$ dimensional vector $\param_{J(i)}$. Assuming small deformations, the model is similar to a standard regression problem:
\begin{equation}
\label{eq:exa1}
Y_i=\phi(\param_{J(i)})+\sigma^2\eps_i\,,
\end{equation}
where $Y_i$ is regarded as a vector $\rset^{225}$, $\phi:\rset^{256}\to\rset^{225}$ is some deterministic mapping and $\sigma>0$ is the standard deviation of the additive noise $\eps_i\sim\norm(0_{225},\mathrm{Id}_{225})$.
\end{exa}

Given a set of $N$ labeled images $Y_1,Y_2,\ldots,Y_N$ and a prior distribution for $\param=\{\param_1,\ldots,\param_9\}$, one can estimate $\param$ through its posterior distribution $\pi$, for example using the Metropolis-Hastings (M-H) algorithm \citep{metropolis1953equation}. However, since the regression function $\phi$ in \eqref{eq:exa1}  is quite sophisticated, even a single likelihood evaluation is expensive to calculate.  As a result, the M-H efficiency can be questioned as computing the $N$ likelihoods in the M-H ratio dramatically slows down each transition.

At this stage, we do not provide precise details on the Informed Sub-Sampling MCMC method but we simply provide an insight of the rationale of our approach. It designs a Markov chain whose transition kernel targets a scaled version of the posterior distribution of the parameter of interest $\param$ given a random subset of $n$ images ($n\ll N$). More specifically, we inject in the standard M-H transition a decision about \textit{refreshing} the subset of data, which, as a result, will change randomly over time. In this example, we use the knowledge of the observation labels to promote subsets of images in which the proportion of each digit is balanced.

We consider $N=10,000$ images of five digits $1,\ldots,5$, subsets of size $n=100$ and a non-informative Gaussian prior for $\param$, as specified in \cite{allassonniere2007towards}. Figure \ref{fig:template} 
indicates a striking advantage of our method compared to a standard M-H using the same $N=10,000$ images. In this scenario, we allow a fixed computational budget (1 hour) for both methods and compare the estimation 
of the mean estimate of the two Markov chains. Qualitatively, the upper part of Figure \ref{fig:template} compares the estimated template shapes of the five digits at different time steps and shows that ISS-MCMC
allows one to extract template shapes much quicker than the standard M-H, while still reaching an apparent similar graphical quality after one hour. This fact is confirmed quantitatively, in the lower part of Figure \ref{fig:template}, which plots, against time and for both methods, the Euclidean distance between the Markov chain mean estimate and the maximum
likelihood estimate $(\param^{\ast}_1,\ldots,\param^{\ast}_5)$ obtained using a stochastic EM \citep{allassonniere2007towards}. More precisely, we compare the real valued function $\left\{d(t),\;t\in\rset\right\}$ defined as
\begin{equation*}
d(t)=\sum_{j=1}^5\|\param_j^{\ast}-\mu(\param_{j,1:\kappa(t)})\|,\qquad\text{where}
\quad
\left\{
\begin{array}{l}
\forall t\in\rset,\;\kappa(t)=\max_{k\in\nset}\{t\geq \tau_k\}\,,\\
\tau_k\;\text{is\,the\,time\,at\,the\,end\,of\,the\,}k\,\text{-th\,iteration}\,,\\
\forall k\in\nset,\;\mu(\param_{j,1:k})=(1/k)\sum_{\ell=1}^k\param_{j,\ell}\,,
\end{array}
\right.
\end{equation*}
where we have defined for $(j,k)\in\{1,\ldots,5\}\times \nset$, $\theta_{j,k}$ as the $j$-th class parameter obtained after $k$ iterations of the Markov chains. For a vector $x\in\rset^n$, $\|\cdot\|$ will refer to the usual L2 norm on $\rset^n$, unless stated otherwise.

\begin{figure}[h!]
\centering
\begin{tabular}{ccc}
\hspace{-.5cm}time & M-H & \hspace{-.7cm} Informed Sub-Sampling MCMC \\
\hspace{-1.3cm}
\begin{tabular}{c}
\vspace{1cm}3 mins\\
\vspace{1cm}15 mins\\
\vspace{1cm}30 mins\\
\vspace{0.1cm}60 mins
\end{tabular} &\hspace{-.9cm}
\begin{tabular}{c}
\includegraphics[scale=0.41]{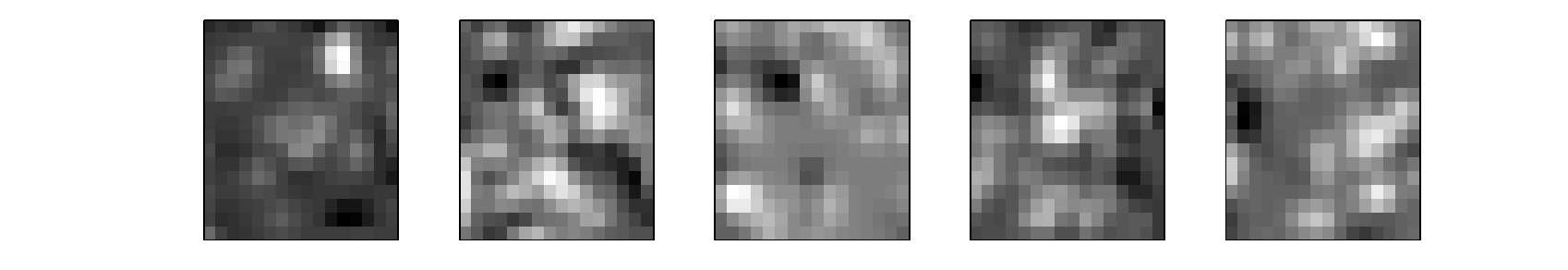}\\
\includegraphics[scale=0.41]{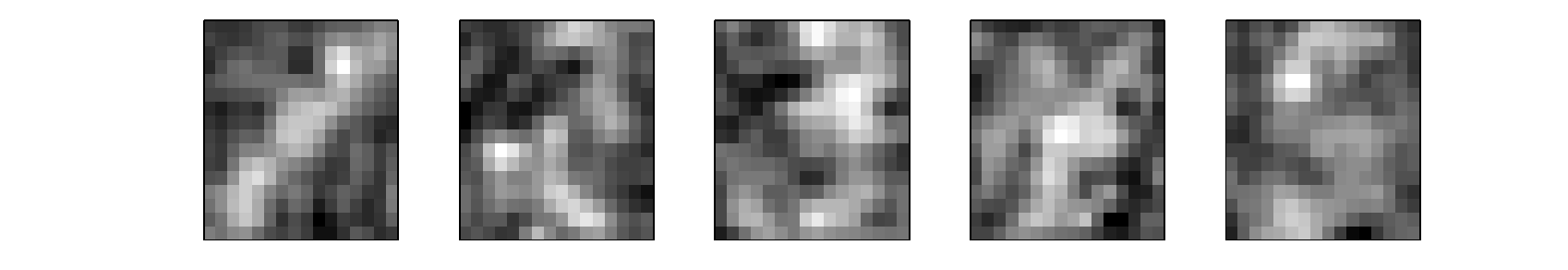}\\
\includegraphics[scale=0.41]{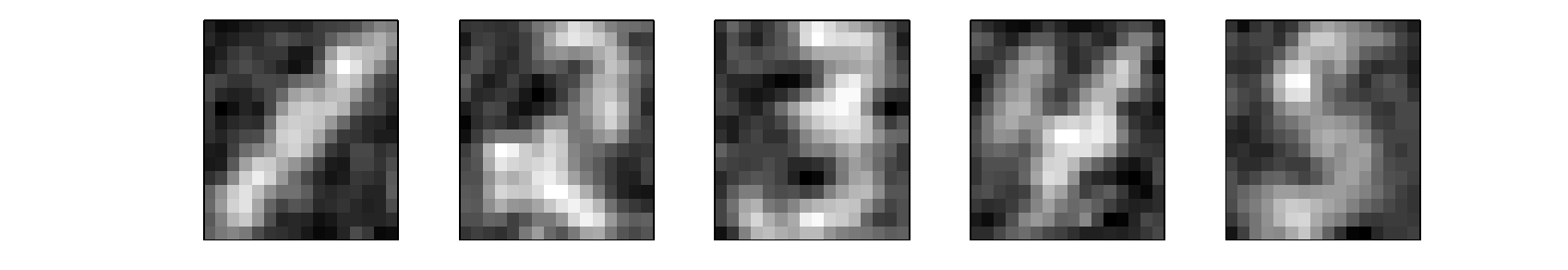}\\
\includegraphics[scale=0.41]{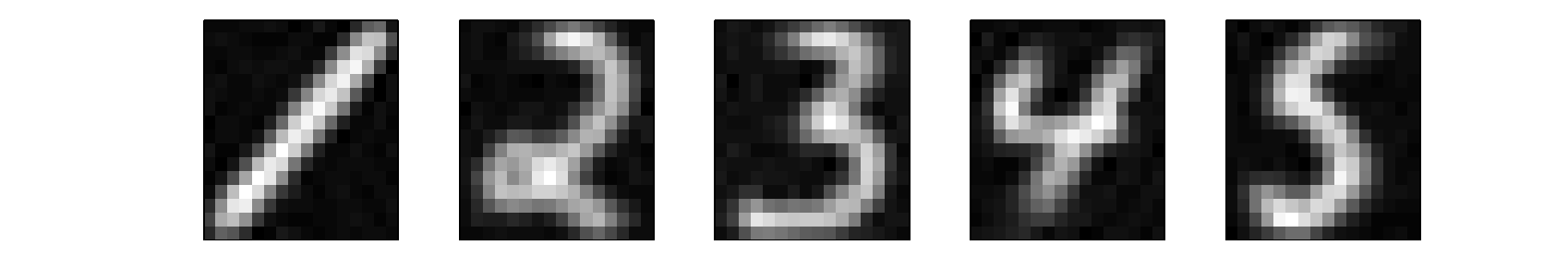}
\end{tabular}&\hspace{-1.5cm}
\begin{tabular}{c}
\includegraphics[scale=0.41]{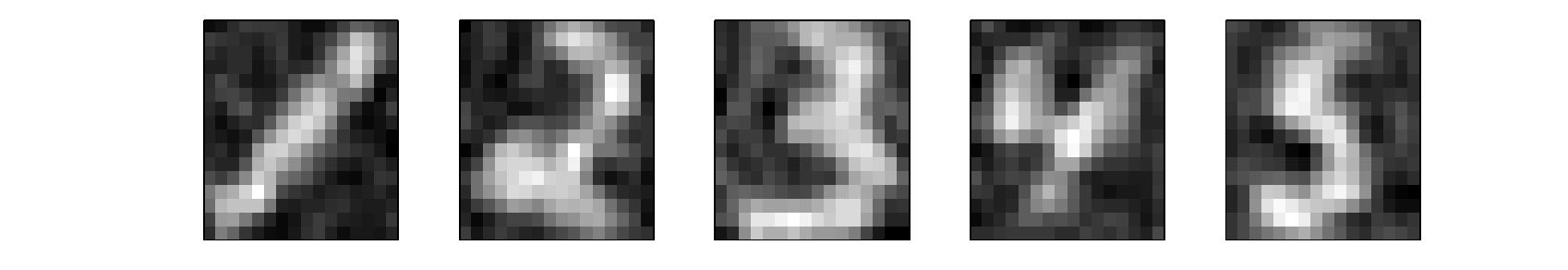}\\
\includegraphics[scale=0.41]{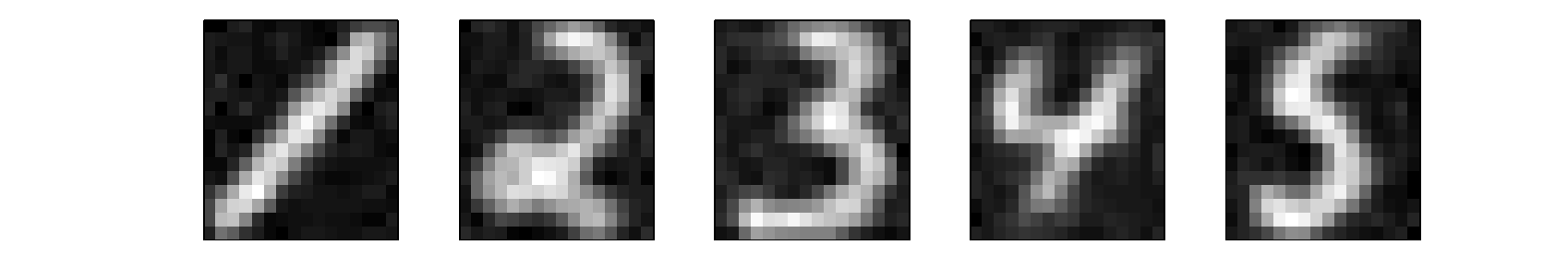}\\
\includegraphics[scale=0.41]{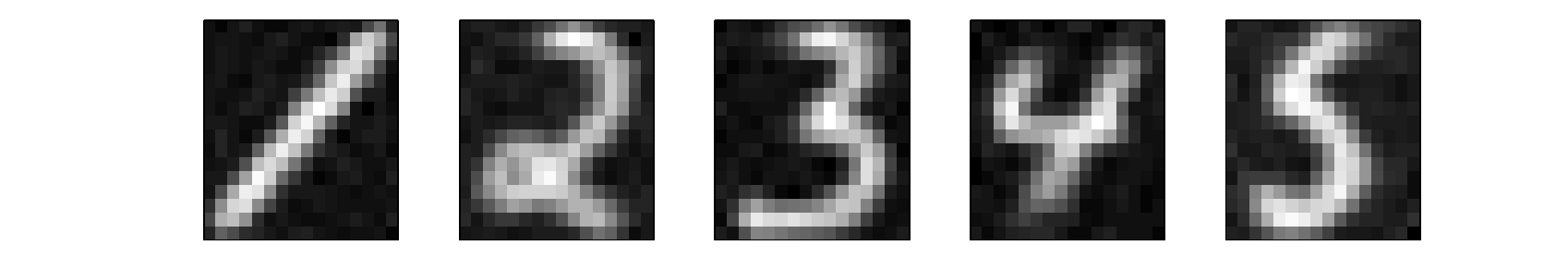}\\
\includegraphics[scale=0.41]{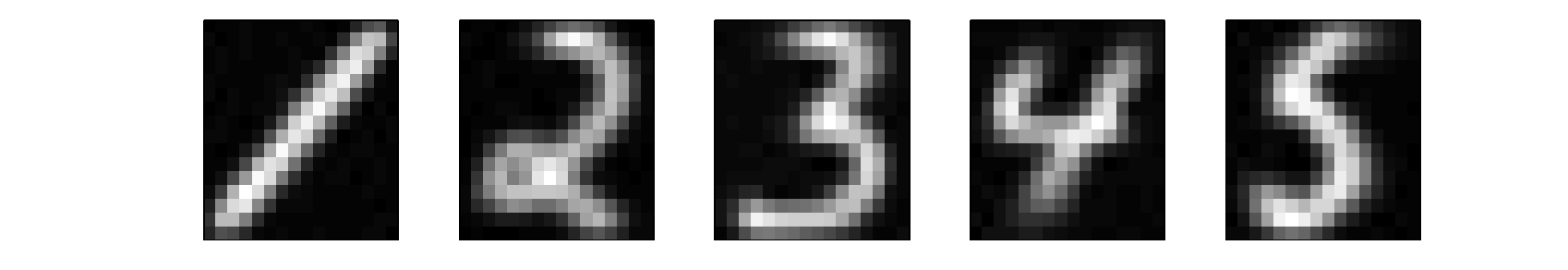}
\end{tabular}
\end{tabular}
\includegraphics[scale=0.5]{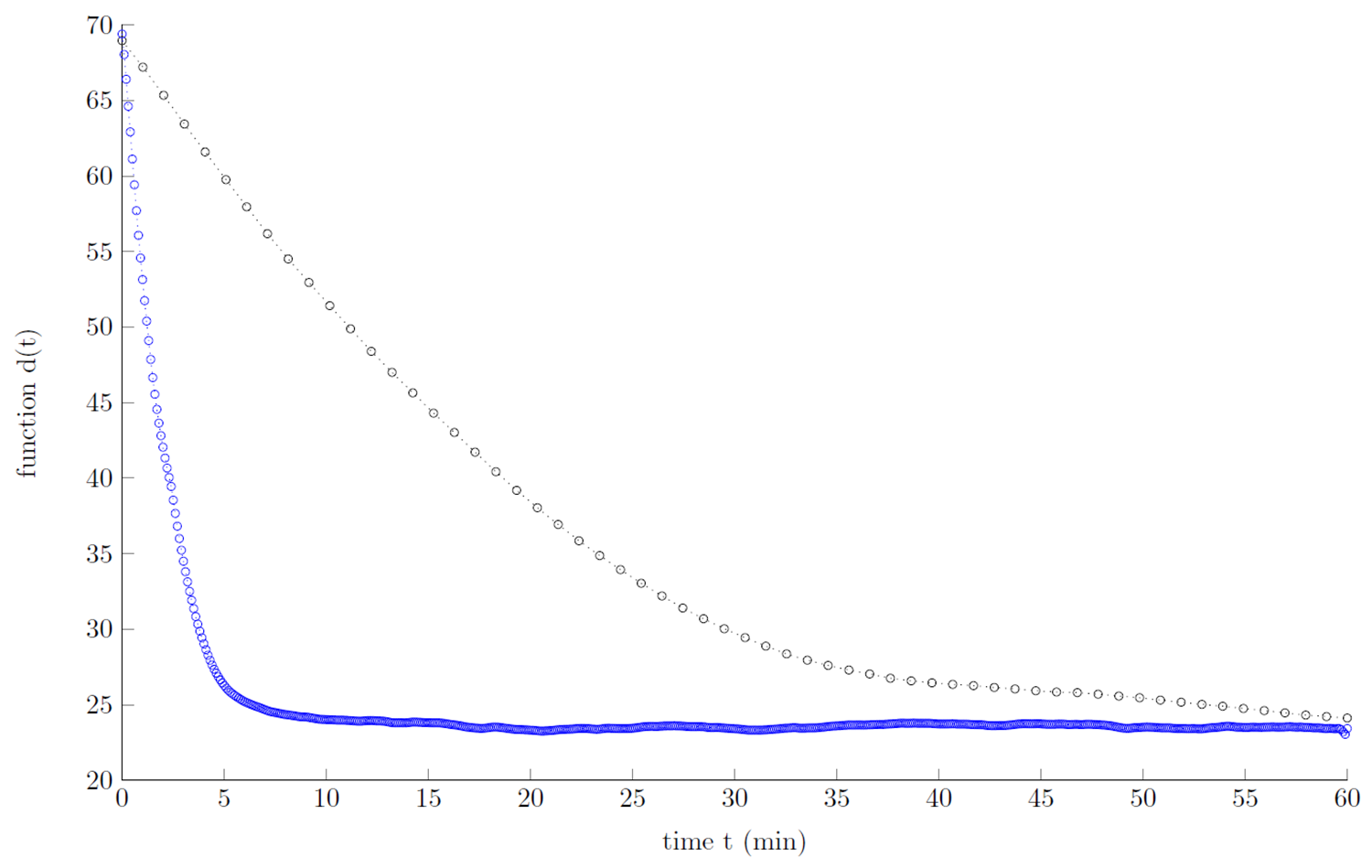}
\caption{(Example \ref{ex1}: Handwritten digits) Efficiency of template estimation through M-H (black) and Informed Sub-Sampling MCMC (blue).}
\label{fig:template}
\end{figure}

One can see that the transient phase of the Informed Sub-Sampling Markov chain is significantly shorter than that of the Metropolis-Hastings chain. More details on Example \ref{ex1} can be found at Section \ref{sec:sim}.  In particular, Figure \ref{fig:template2} shows that the stationary distribution of Informed Sub-Sampling matches reasonably well $\pi$, which is a primary concern in Bayesian inference.

Our algorithm provides very encouraging results for this real data example. We motivate and formalize our method in Sections \ref{sec:expo} and \ref{sec:alg} and provide theoretical arguments supporting it at Section \ref{sec:conv}. 

\section{Approximation of the posterior distribution in exponential models: an optimality result}
\label{sec:expo}
In this section, we consider the case of $N$ independent and identically distributed (\iid) observations from an exponential model. Sampling from the posterior distribution of such models  using the 
Metropolis-Hastings algorithm is effortless since the information conveyed by the $N$ observations is contained in the sufficient statistics vector, which needs to be calculated only once.

The existence of sufficient statistics in this type of models allows us to establish a number of theoretical results that will be used to design and justify our Informed Sub-Sampling methodology that approximately 
samples from posterior distributions in general contexts, \ie non-\iid observations from general likelihood models without sufficient statistics. More precisely, Propositions \ref{prop:KL} and \ref{prop:BvM} put 
forward an optimal approximation of the posterior distribution $\pi$ by a distribution $\ttarg_{n}$ of the parameter of interest given only a subsample of $n$ observations. Finally, Proposition \ref{prop:esp} 
justifies the introduction of a probability distribution on the set of subsamples. This is an essential element of our work as it represents a significant departure from all existing subsampling methodologies 
proposed in the Markov chain Monte Carlo literature, that have assumed uniform distribution on the subsamples.

\subsection{Notation}
Let $\coord{Y}[N]\in \Yset^{N}$ be a set of \iid observed data $(\Yset\subseteq\rset^m,\, m>0)$ and define
\begin{itemize}
\item $Y_{i:j}=(Y_i,\ldots,Y_j)$ if $1\leq i \leq j \leq N$ with the convention that $Y_{i:j}=\{\emptyset\}$, otherwise.
\item $Y_U=\{Y_k,\,k\in U\}$, where $U\subseteq \{1,\ldots,N\}$.
\end{itemize}

In this section, we assume that the likelihood model $f$ belongs to the exponential family and is fully specified by a vector of parameters $\param\in\paramset$, ($\paramset\subseteq \rset^d,\,d>0$), a bounded mapping $g:\paramset\to\suffset $ and a sufficient statistic mapping $S:\Yset\to\suffset$ ($\suffset\subseteq\rset^s,s>0$) such that
$$
f(y\,|\,\param)=\exp\left\{g(\param)\T S(y)\right\}\bigg\slash L(\param)\,,\qquad L(\param)=\int_{Y\in\Yset}\exp\left\{S(y)\T g(\param)\right\}\rmd y\,,
$$
is the density of the likelihood distribution with respect to the Lebesgue measure. The posterior distribution $\targ$ is defined on the measurable space $(\paramset,\paramalg)$ by its density function
\begin{equation}
\label{eq:posterior_full}
\targ(\param\,|\,Y_{1:N})= p(\param) \frac{\exp\left\{\sum_{k=1}^N S(Y_k)\T g(\param)\right\}}{L(\param)^N}\bigg \slash Z(Y_{1:N})\,,
\end{equation}
where
\begin{equation}
Z(Y_{1:N})=\int p(\rmd\param)\frac{\exp\left\{\sum_{k=1}^N S(Y_k)\T g(\param)\right\}}{L(\param)^N}\,.
\end{equation}
$p$ is a prior distribution defined on $(\paramset,\paramalg)$ and with some abuse of notation, $p$ denotes also the probability density function on $\paramset$.

For all $n\leq N$, we define $\Uset_n$ as the set of possible combinations of $n$ different integer numbers less than or equal to $N$ and $\Ualg_n$ as the powerset of $\Uset_n$. 
In the sequel, we set $n$ as a constant and wish to compare the posterior distribution $\pi$ \eqref{eq:posterior_full} with any distribution from the family $\fam_n=\{\ttarg_n(U),\,U\in\Uset_n\}$, where for all $U\in\Uset_n$, we have defined $\ttarg_n(U)$ as the distribution on $(\paramset,\paramalg)$ with probability density function
\begin{equation}
\label{eq:subposterior}
\ttarg_n(\param\,|\,Y_{U})\propto p(\param) f(Y_U\,|\,\param)^{N/n}\,.
\end{equation}

\subsection{Optimal subsets for the Kullback-Leibler divergence between $\targ$ and $\ttarg_n$}
Recall that for two measures $\targ$ and $\ttarg$ defined on the same measurable space $(\paramset,\paramalg)$, the Kullback-Leibler (KL) divergence between $\targ$ and $\ttarg$ is defined as:
\begin{equation}
\label{eq:KL}
\KL(\targ,\ttarg)=\esp_{\targ}\left\{\log\frac{\targ(\param)}{\ttarg(\param)}\right\}\,.
\end{equation}
Although not a proper distance between probability measures defined on the same space, $\KL(\targ,\ttarg)$ is used as a similarity criterion between $\targ$ and $\ttarg$. It can be interpreted in information theory as a measure of the information lost when $\ttarg$ is used to approximate $\targ$, which is our primary concern here. We now state the main result of this section.

\begin{prop}
\label{prop:KL}
For any subset $U\in\Uset_n$, define the vector of difference of sufficient statistics between the whole dataset and the subset $Y_U$ as
\begin{equation}
\label{eq:Delta}
\Delta_n(U)=\sum_{k=1}^NS(Y_k)-(N/n)\sum_{k\in\U}S(Y_k)\,.
\end{equation}
Then, the following inequality holds:
\begin{equation}
\label{eq:KL_inequality}
\KL\left\{\targ,\ttarg_n(U)\right\}\leq B(Y,U)\,,
\end{equation}
where
\begin{equation}
\label{eq:KL_bound}
B(Y,U)=\log\esp_\pi\exp\left\{\|\esp_\pi(g(\param))-g(\param)\|\,\|\Delta_n(U)\|\right\}
\end{equation}
and $\|\cdot\|$ is the L2 norm.
\end{prop}

The proof is detailed in \ref{app1} and follows from straightforward algebra and applying Cauchy-Schwartz inequality. Note that by definition of $B$, we remark that for any two subsets $(U_1,U_2)\in\Uset_n^2$,
$$
\|\Delta_n(U_1)\|\leq \|\Delta_n(U_2)\|\;\Rightarrow \;B(Y,U_1)\leq B(Y,U_2)\,.
$$

The following corollary is an immediate consequence of Proposition \ref{prop:KL}.

\begin{coro}
\label{coro1:KL}
Define the set:
\begin{equation}
\label{eq:opti}
\Usetst_n:=\left\{U\in\Uset_n,\quad \frac{1}{N}\sum_{k=1}^N S(Y_k)=\frac{1}{n}\sum_{k\in U}S(Y_k)\right\}\,.
\end{equation}
If $\Usetst_n$ is non-empty, then for any $U\in\Usetst_n$, then $\pi(\theta\,|\,Y)=\ttarg_n(\theta\,|\,Y_U)$, $\pi$-almost everywhere.
\end{coro}

A stronger result can be obtained under the assumption that a Bernstein-von Mises Theorem \cite{van2000asymptotic,le2012asymptotic} holds for the concentration of $\pi$ to its Normal approximation:
\begin{equation}
\label{eqBvM}
\hpi(\,\cdot\,|\,Y_{1:N}):=\mathcal{N}\left(\thetaMLE(Y_{1:N}),I^{-1}\left(\theta_0\right)/N\right)\,,
\end{equation}
where $\norm$ denotes the Normal distribution, $\thetaMLE(Y_{1:N})=\arg\max_{\theta\in\Theta}f(Y_{1:N}\,|\,\theta)$, $\theta_0\in\Theta$ is some parameter and $I(\theta)$ is the Fisher information matrix given $Y_{1:N}$ at $\theta$.

\begin{prop}
\label{prop:BvM}
Let $(U_1,U_2)\in\Uset_n^2$. Assume that for all $i\in\{1,\ldots,d\}$, $|\Delta_n(U_1)^{(i)}|\leq |\Delta_n(U_2)^{(i)}|$, where $|\Delta_n(U_1)^{(i)}|$ refers to the $i$-th element of $\Delta_n(U_1)$ \eqref{eq:Delta}. Then $\wKL_n(U_1)\leq \wKL_n(U_2)$, where $\wKL_n(U)$ is the Kullback-Leibler divergence between the asymptotic approximation of the posterior $\hpi$ \eqref{eqBvM}  and $\ttarg_n(U)$ \eqref{eq:subposterior}.
\end{prop}

The proof is detailed in \ref{app2}. Note that the asymptotic approximation is for $N\to\infty$ and for a fixed $n$ and is thus relevant to the context of our analysis.

\subsection{Weighting the subsamples}
Consider the distribution $\nu_{n,\eps}$ on the discrete space $\Uset_n$ defined for all $\epsilon\geq 0$ by:
\begin{equation}
\label{eq:nu}
\nu_{n,\eps}(U)\propto \exp\left\{-\epsilon\|\Delta_n(U)\|^2\right\}\,,\qquad \text{for all}\;U\in\Uset_n\,.
\end{equation}
$\nu_{n,\eps}$ assigns a weight to any subset according to their representativeness with respect to the full dataset. When $\eps=0$, $\nu_{n,\eps}$ is uniform on $\Uset_n$ while when $\eps\to\infty$, $\nu_{n,\eps}$ is uniform on the set of subset(s) that minimize(s) $U\mapsto\|\Delta_n(U)\|$.  Proposition \ref{prop:BvM} suggests that for exponential models, the optimal inference based on subsamples of size $n$ is obtained by picking the subposterior $\pi_n(U)$ \eqref{eq:subposterior} using the distribution $U\sim\nu_{n,\eps}$ with $\eps\to \infty$.

We now state Proposition \ref{prop:esp}. This result is important even though somewhat obscure at this stage. Indeed, we will show that it is a necessary condition for the method we introduce in Section \ref{sec:alg} to converge. In fact, moving away to general models (\ie non \iid and non exponential) amounts to relax the sufficient statistics existence assumption as well as the  $\eps\to\infty$ condition. This will be achieved by constructing a class of summary statistics for the model at hand for which a similar result to Proposition \ref{prop:esp} holds.

\begin{prop}
\label{prop:esp}
For any $\theta\in\Theta$ and $\eps>0$, there exists $M<\infty$ such that:
\begin{equation}
\label{upper}
\esp_{n,\eps}\left\{\frac{f(Y\,|\,\theta)}{f(Y_U\,|\,\theta)^{N/n}}\right\}<M\,,
\end{equation}
where $\esp_{n,\eps}$ is the expectation under $\nu_{n,\eps}$, as defined in \eqref{eq:nu}.
\end{prop}
The proof is postponed to \ref{sec:proof:prop3}. Note that Proposition \ref{prop:esp} essentially holds because $\log \nu_{n,\eps}$ is quadratic in $\|\Delta_n(U)\|$. Other weighting schemes for the subsets (e.g. uniform weights or weights $\propto \exp\{-\eps\|\Delta_n(U)\|\}$) would not necessarily allow to bound $\esp_{n,\eps}\{f(Y\,|\,\theta)\slash f(Y_U\,|\,\theta)^{N/n}\}$.

\subsection{Illustration with a probit model: effect of choice of sub-sample}
\label{sec:expo:subsec}
We consider a pedagogical example, based on a probit model, to illustrate the results from the previous subsections.

\begin{exa}
\label{ex2}
A probit model is used in regression problems in which a binary variable $Y_k\in\{0,1\}$ is observed through the following sequence of independent random experiments, defined for all $k\in\{1,\ldots,N\}$ as:
\begin{enumerate}[(i)]
\item Draw $X_k\sim \norm(\paramst,\gamma^2)$
\item Set $Y_k$ as follows
\begin{equation}
\label{eq:probit_model}
Y_k=
\left\{
\begin{array}{ll}
1, & \text{if}\; X_k>0,\\
0, & \text{otherwise}.
\end{array}
\right.
\end{equation}
\end{enumerate}
\end{exa}

Observing a large number of realizations $Y_1,\ldots,Y_N$, we aim to estimate the posterior distribution of $\param$. If $\gamma$ is unknown, the model is not identifiable and for simplicity we considered it as known here. The likelihood function can be expressed as
\begin{equation}
\label{eq:probit}
f(Y_k\,|\,\param)=\alpha(\param)^{Y_k}(1-\alpha(\param))^{(1-Y_k)}=\left(1-\alpha(\param)\right)\left(\frac{\alpha(\param)}{1-\alpha(\param)}\right)^{Y_k},
\end{equation}
where $\alpha(\param)=\int_{0}^{\infty}(2\pi\gamma^2)^{-1/2}\exp\{-(1/2\gamma^2)(t-\param)^2\}\rmd t$ and clearly belongs to the exponential family. The pdf of the posterior distribution $\pi$ and any distribution $\ttarg_n(U)\in\fam_n$ writes respectively as
$$
\targ(\param\,|\,Y_{1:N})\propto p(\param)\left(1-\alpha(\param)\right)^{N}\left(\frac{\alpha(\param)}{1-\alpha(\param)}\right)^{\sum_{k=1}^N Y_k}\,,
$$
$$
\ttarg_n(\param\,|\,Y_U)\propto p(\param)\left(1-\alpha(\param)\right)^{N}\left(\frac{\alpha(\param)}{1-\alpha(\param)}\right)^{(N/n)\sum_{k\in U} Y_k}\,,
$$
where $p$ is a prior density on $\theta$. Again, in this example, the posterior density is easy to evaluate pointwise, even when $N$ is extremely large, as it only requires to sum over all the binary variables $Y_1,\ldots,Y_N$. As a consequence, samples from $\pi$ can routinely be obtained by a standard M-H algorithm and similarly for any distribution $\ttarg_n({U})\in\fam_n$.

\begin{figure}
\centering
\includegraphics[scale=0.65]{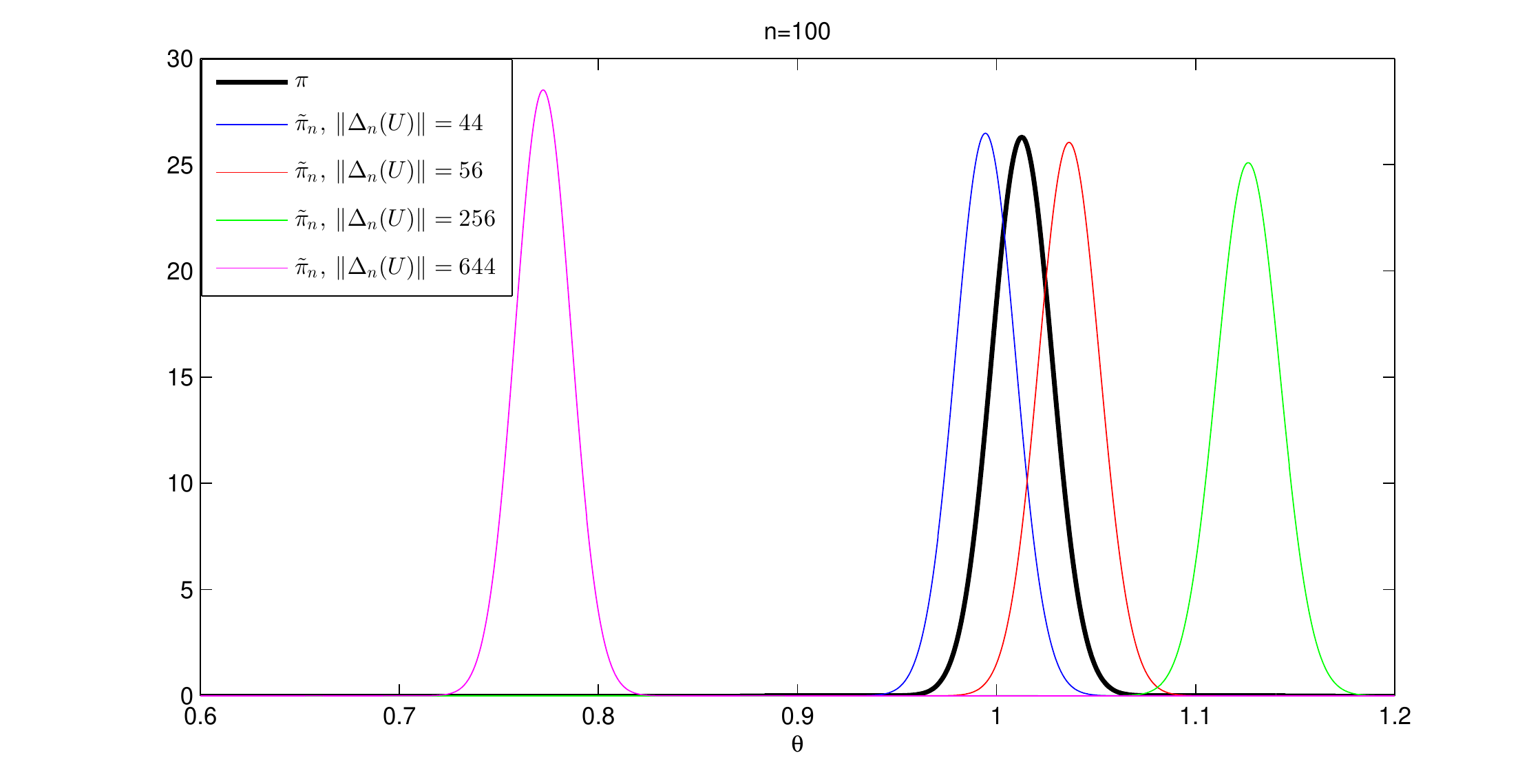}
\includegraphics[scale=0.65]{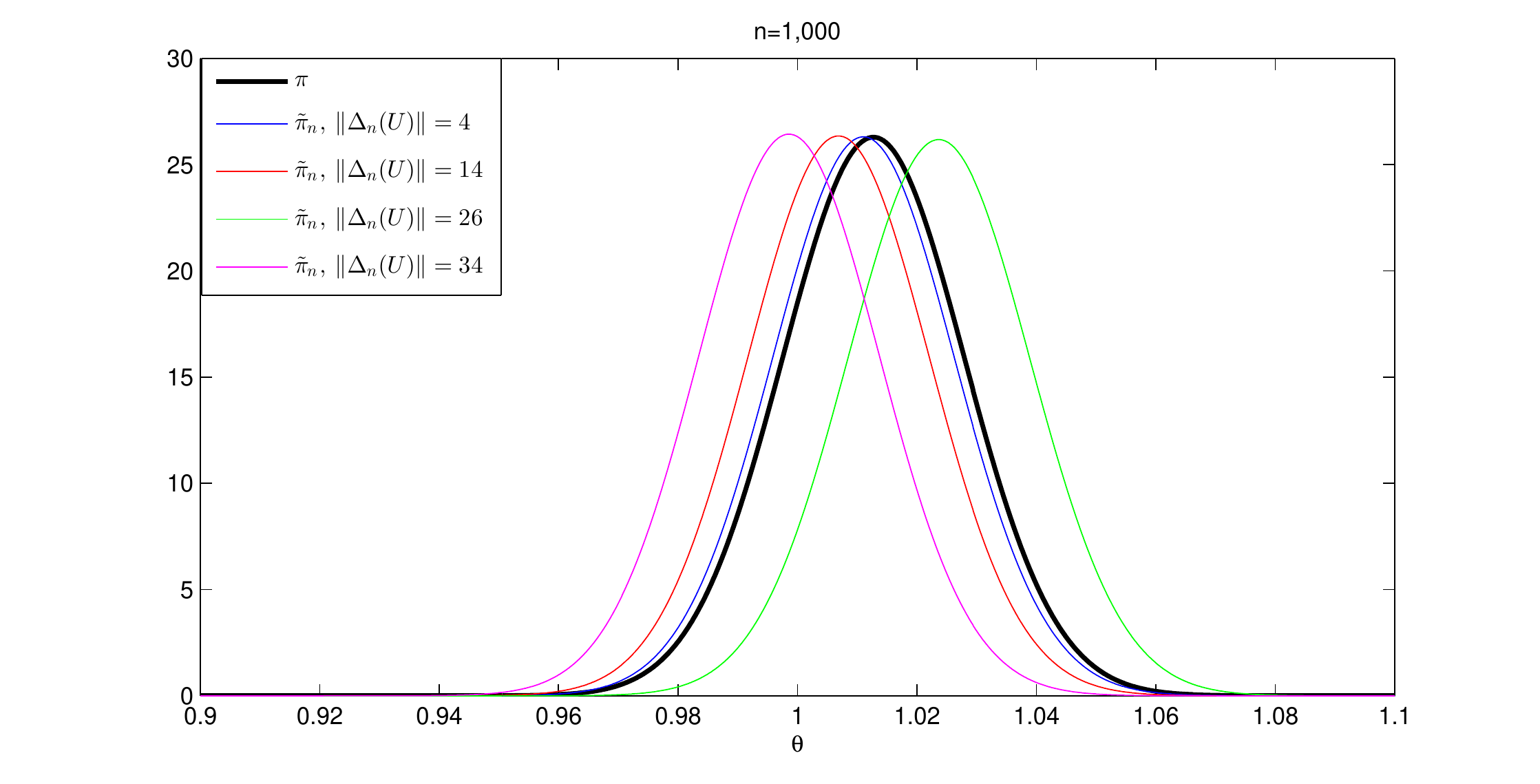}
\caption{(Example \ref{ex2}: Probit model) Influence of the parameter $U\in\Uset_{n}$ on the sub-posterior distribution $\ttarg_n(U)$ and comparison with $\targ$ for subsets of size $n=100$ (top) and $n=1,000$ (bottom).}
\label{fig:probit}
\end{figure}

\begin{table}
\centering
\begin{tabular}{|c|c|c|c|}
\hline
$n$ & $\|\Delta_n(U)\|$ & $\KL\left\{\targ,\ttarg_n(U)\right\}$ & $B(Y,U)$\\
\hline
$1,000$ & $3$ & $0.004$ & $0.04$\\
$1,000$ &$14$ & $0.11$ & $0.18$\\
$1,000$ &$23$ & $0.19$ & $0.29$\\
$100$ &$33$ & $0.41$ & $0.54$\\
\hline
\end{tabular}

\vspace{.3cm}
\caption{(Example \ref{ex2}: Probit model) Comparison of the KL divergence between $\pi$ and the optimal $\ttarg_n\in\fam_n$ ($\|\Delta_n(U)\|=3$) and other distributions in $\fam_n$.}
\label{tab:KL}
\end{table}

We simulated $N=10,000$ simulated data $Y_1,\ldots,Y_N$ from \eqref{eq:probit_model}, with true parameter $\theta^\ast=1$. We used the prior distribution $p=\norm(0,10)$. In this probit model, $S$ is simply the 
identity function, implying that $\|\Delta_n(U)\|$ gives the absolute value of the difference between the scaled proportion of 1 and 0's between the full dataset and the subset $Y_U$. Figure \ref{fig:probit} 
reports the density functions of $\pi$ and several other distributions $\ttarg_n(U)\in\fam_n$, for $n=100$ and $n=1,000$, with different values for the quantity $\|\Delta_n(U)\|$ \eqref{eq:Delta}. This plot, as 
well as the quantitative result of Table \ref{tab:KL} are consistent with the statement of Corollary \ref{prop:KL}: when learning from a subsample of $n$ data, one should work with a subset $U$ featuring a perfect 
match with the full dataset, \ie $\|\Delta_n(U)\|=0$, or as small as possible to achieve an \textit{optimal} approximation of $\targ$. Finally, Figure \ref{fig:upper} illustrates Proposition \ref{prop:esp}: assigning 
the distribution $\nu_{n,\eps}$ \eqref{eq:nu} to the subsamples allows one to control the expectation of the likelihood ratio $f(Y\,|\,\theta)\slash f(Y_U\,|\,\theta)^{N/n}$ around 1.

\begin{figure}
\centering
\includegraphics[scale=0.65]{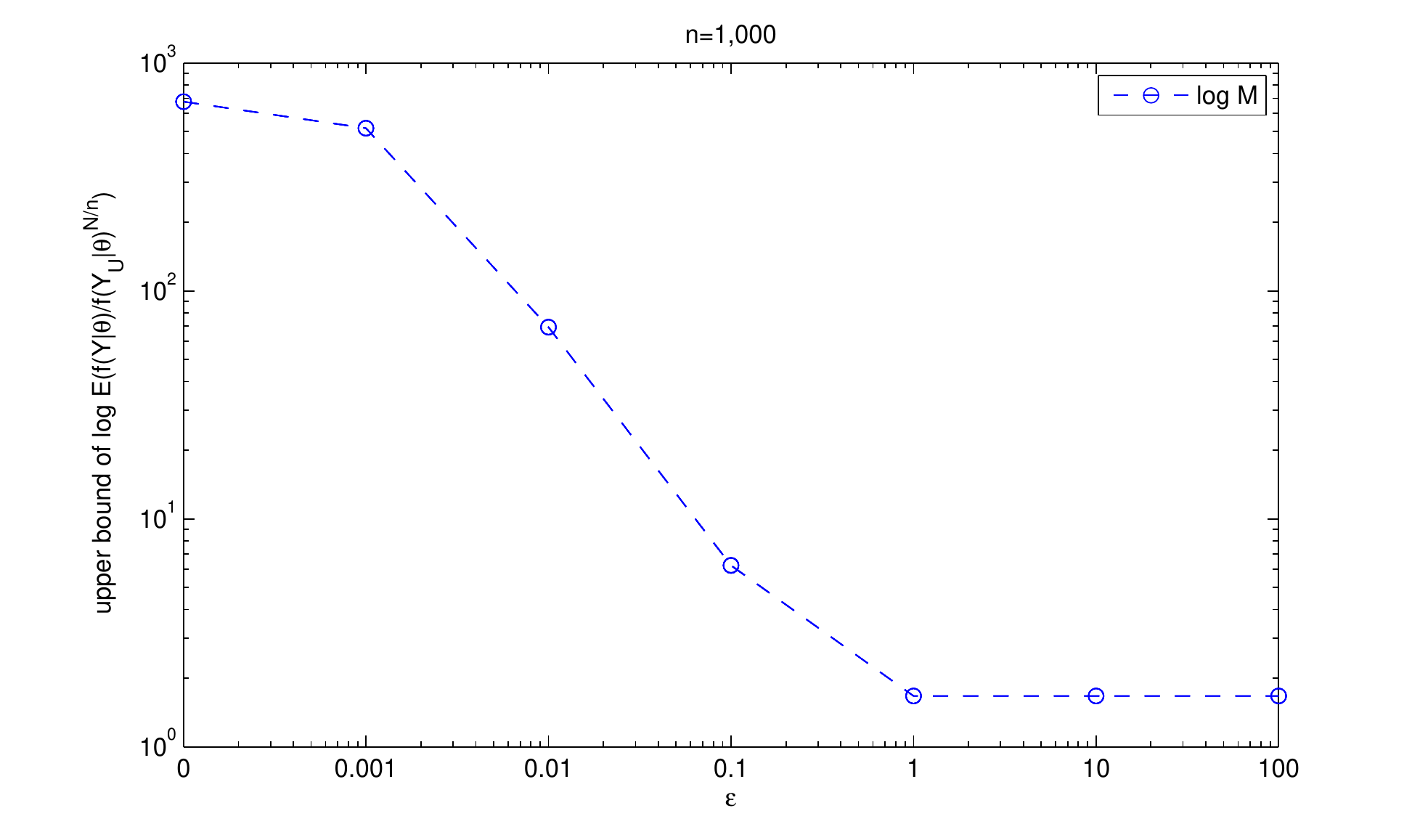}
\caption{(Example \ref{ex2}: Probit model) Influence of the parameter $\eps$ of the distribution $\nu_{n,\eps}$ on the upper bound $M$ of $\esp\{f(Y\,|\,\theta)\slash f(Y_U\,|\,\theta)^{N/n}\}$ for $\theta\in(0,1.5)$, for $n=1,000$. When $\epsilon=0$, $\nu_{n,\eps}$ is uniform (i.e. an identical weight is assigned to all the subsamples) and as a consequence $M\equiv\infty$. Conversely, when $\epsilon\gg 0$, the mass of $\nu_{n,\eps}$ spreads over the best subsamples $Y_U,\,U\in\Uset_n$ (i.e. those minimizing $\Delta_n(U)$) and the bound $M$ is smaller than $e^2$. Indeed, by assigning a weight $\nu_{n,\eps}(U)\propto \exp\{-\eps\Delta_n(U)^2\}$ those subsamples $Y_U,\,U\in\Uset_n$ that have a large $\Delta_n(U)$ will yield a negligible contribution to the expectation, hence preventing from divergence.}
\label{fig:upper}
\end{figure} 

\section{Informed Sub-Sampling MCMC}
\label{sec:alg}
In this section, we do not assume any particular correlation pattern for the sequence of observations, nor any specific likelihood model and simply write the posterior distribution $\targ$ as
\begin{equation}
\targ(\rmd\param\,|\,{Y}_{1:N})\propto p(\rmd\param) f(Y_{1:N}\,|\,\param)\,.
\end{equation}

The Informed Sub-Sampling MCMC (ISS-MCMC) methodology that we describe now can be regarded as an extension of the approximation detailed in the previous section to non-exponential family models with possibly dependent observations.

\subsection{Motivation of our approach}
Central to our approach is the idea that all subsamples $Y_U$ $(U\in\Uset_n)$ are not equally valuable for inferring $\pi$. Here, we do not assume the existence of a sufficient statistic mapping for the models under consideration. Thus, in order to discriminate between different subsamples, we introduce an \textit{artificial} summary statistic mapping $S:\Yset_n\to\suffset$ ($n\leq N$), where $\suffset\subseteq \rset^s$. The choice of the summary statistics $S$ is problem specific and is meant to be the counterpart of the sufficient statistic mapping for general models (hence sharing, slightly abusively, the same notation). Since the question of specifying summary statistics also arises in Approximate Bayesian Computation (ABC), one can take advantage of the abundant ABC literature on this topic to find some examples of summary statistics for usual likelihood models \cite[see \eg][]{nunes2010optimal,csillery2010approximate,marin2012approximate,fearnhead2012constructing}. More details on validation of summary statistics are discussed in Section \ref{sec:alg:sumstat}.

Because the statistics used to assess the representativeness of a subsample $Y_U$ w.r.t. the full dataset $Y$ are only \textit{summary} and not \textit{sufficient}, the results of Section \ref{sec:expo} are no longer valid. In particular, should an optimal subset $U^{\ast}$ minimising a distance between $S(Y_U)$ and $S(Y)$ exist, inferring $\pi$ through the approximation $\ttarg_n(U^\ast)$ is in no sense optimal. In fact, as shown in several examples of Section \ref{sec:sim}, this approximation is usually poor. In such a setting, it is reasonable to consider extending the set of subsamples of interest to a pool of \textit{good} subsamples. This naturally suggests using the distribution $\nu_{n,\eps}$ \eqref{eq:nu} to discriminate between the subsamples, replacing sufficient by summary statistics and relaxing the assumption $\eps\to\infty$, in order to account for a collection of good subsamples. Before proceeding to the presentation of our algorithm, we define the following quantities related to a subset $U\in\Uset_n$:
\begin{equation}
\label{eq:suff_stat}
\bar{S}(Y_U)=S(Y_U)/n\,,\qquad \bar{\Delta}_n(U)=S(Y)/N-S(Y_u)/n\,.
\end{equation}

\subsection{Informed Sub-Sampling MCMC: the methodology}

Informed Sub-Sampling MCMC is a scalable adaptation of the Metropolis-Hastings algorithm \citep{metropolis1953equation}, designed for situations when $N$ is prohibitively large to perform inference on the posterior $\pi$ in a reasonable time frame.  ISS-MCMC relies on a Markov chain whose transition kernel has a bounded computational complexity, which can be controlled through the parameter $n$. We first recall how the Metropolis-Hastings algorithm produces a $\pi$-reversible Markov chain $\{\theta_i,\,i\in\nset\}$, for any distribution $\pi$ known up to a normalizing constant. The index $i$ is used hereafter to refer to the Markov chain iteration.

\subsubsection{Metropolis-Hastings}
Let $Q$ be a transition kernel on $(\Theta,\borel(\Theta))$ and assume that the Metropolis-Hastings Markov chain is at state $\theta_i$. A transition $\theta_i\to \theta_{i+1}$ consists in the two following step:
\begin{enumerate}[(a)]
\item propose a new parameter $\theta\sim Q(\theta_i,\cdot)$
\item set the next state of the Markov chain as $\theta_{i+1}=\theta$ with probability
\begin{equation}
\label{eq:MHratio}
\alpha(\theta_i,\theta)=1\wedge a(\theta_i,\theta)\,,\qquad
a(\theta_i,\theta)=\frac{\pi(\theta\,|\,Y)Q(\theta,\theta_i)}{\pi(\theta_i\,|\,Y)Q(\theta_i,\theta)}
\end{equation}
and as $\theta_{i+1}=\theta_i$ with probability $1-\alpha(\theta_i,\theta)$.
\end{enumerate}
Algorithm \ref{alg:mh} details how to simulate a Metropolis-Hastings Markov chain $\{\theta_i,\,i\in\nset\}$.

\begin{algorithm}
\caption{Metropolis-Hastings algorithm}
\label{alg:mh}
\begin{algorithmic}[1]
\State {\bf{Input: initial state $\param_0$ and posterior evaluation $\pi(\theta_0\,|\,Y)$}}
\For{$i=1,2,\ldots$}
\State propose a new parameter $\param\sim Q(\param_{i-1};\,\cdot\,)$ and draw $I\sim \text{unif}(0,1)$
\State compute $\targ(\theta\,|\,Y)$ and $a=a(\theta_{i-1},\theta)$ defined in \eqref{eq:MHratio}
\If{$I\leq a$}
\State  set $\param_{i}=\param$
\Else
\State set  $\param_{i}=\param_{i-1}$
\EndIf
\EndFor
\State {\bf{return: the Markov chain $\{\param_{i},\,i\in\nset\}$}}
\end{algorithmic}
\end{algorithm}

\subsubsection{Informed Sub-Sampling MCMC}

To avoid any confusion, we denote by $\{\ttheta_i,\,i\in\nset\}$ the sequence of parameters generated by the Informed Sub-Sampling Markov chain, by contrast to the Markov chain $\{\theta_i,\,i\in\nset\}$ produced by the Metropolis-Hastings algorithm (Alg. \ref{alg:mh}). The pool of \textit{good} subsamples used in the Informed Sub-Sampling inference is treated as a sequence of missing data $U_1,U_2,\ldots$ and is thus simulated by our algorithm. More precisely, ISS-MCMC produces a Markov chain $\{(\ttheta_i,U_i),\,i\in\nset\}$ on the extended space $\Theta\times \Uset_n$. Inspired by the analysis of Section \ref{sec:expo}, the sequence of subsamples $\{U_i,\,i\in\nset\}$ is randomly updated in a way that favours those subsets whose summary statistics vector is close to that of the full dataset. Let $R$ be a symmetric transition kernel on $(\Uset_n,\mathcal{U}_n)$, a transition $(\ttheta_i,U_i)\to (\ttheta_{i+1},U_{i+1})$ consists in the two following steps:

\begin{enumerate}[i-]
\item \begin{enumerate}[(a)]
\item propose a new subset variable $U\sim R(U_i,\,\cdot\,)$
\item set $U_{i+1}=U$ with probability
\begin{equation}
\label{eq:accsubset}
\beta(U_i,U)=1\wedge b(U_i,U)\,,\quad b(U_i,U)=\exp\left\{\eps\left(\|\Delta_n(U_i)\|^2-\|\Delta_n(U)\|^2\right)\right\}
\end{equation}
and $U_{i+1}=U_i$ with probability $1-\beta(U_i,U)$. $\Delta_n$ is defined at Eq. \eqref{eq:Delta}.
\end{enumerate}

\item \begin{enumerate}[(a)]
\item propose a new parameter $\ttheta\sim Q(\ttheta_i,\cdot)$
\item set $\ttheta_{i+1}=\ttheta$ with probability
\begin{equation}
\label{eq:InfomedSubratio}
\talpha(\ttheta_i,\ttheta)=1\wedge\ta(\ttheta_i,\theta\,|\,U_{i+1})\,,\quad \ta(\ttheta_i,\ttheta\,|\,U_{i+1})=\frac{\ttarg_n(\ttheta\,|\,Y_{U_{i+1}})Q(\ttheta,\ttheta_i)}{\ttarg_n(\ttheta_i\,|\,Y_{U_{i+1}})Q(\ttheta_i,\ttheta)}
\end{equation}
and as $\ttheta_{i+1}=\ttheta_i$ with probability $1-\talpha(\ttheta_i,\ttheta\,|\,U_{i+1})$.
\end{enumerate}
\end{enumerate}
Algorithm \ref{alg:infsub} details how to simulate an Informed Sub-Sampling Markov chain. Note that at step 11, if $U_i=U_{i-1}$, the quantity
$\ttarg_n(\theta_{i-1}\,|\,U_{i})$ has already been calculated at the previous iteration.
\begin{algorithm}
\caption{Informed Sub-Sampling MCMC algorithm}
\label{alg:infsub}
\begin{algorithmic}[1]
\State {\bf{Input: initial state $(\ttheta_0,U_0)$ and summary statistics $S_0=\bar{S}(Y_{U_0})$, $S^{\ast}=\bS(Y)$}}
\For{$i=1,2,\ldots$}
\State propose a new subset $U\sim R(U_{i-1},\cdot)$ and draw $J\sim\text{unif}(0,1)$,
\State compute $S=\bar{S}(Y_{U})$ and $b=b(U_{i-1},U)$ defined in \eqref{eq:accsubset}
\If {$J\leq b$}
\State set $U_{i}=U$ and $S_{i}=S$
\Else
\State set $U_{i}=U_{i-1}$ and $S_{i}=S_{i-1}$
\EndIf
\State propose a new parameter $\tparam\sim Q(\tparam_{i-1};\,\cdot\,)$ and draw $I\sim \text{unif}(0,1)$
\State compute $\ttarg_{n}(\ttheta_{i-1}\,|\,Y_{U_{i}})$, $\ttarg_{n}(\ttheta\,|\,Y_{U_{i}})$ and $\ta=\ta(\ttheta_{i-1},\ttheta\,|\,U_{i})$ defined in \eqref{eq:InfomedSubratio}
\If {$I\leq \ta$}
\State set $\tparam_{i}=\tparam$
\Else
\State set $\tparam_{i}=\tparam_{i-1}$
\EndIf
\EndFor
\State {\bf{return: the Markov chain $\{(\tparam_{i},U_{i}),\,i\in\nset\}$}}
\end{algorithmic}
\end{algorithm}

\subsection{Connection with noisy ABC}
Approximate Bayesian Computation (ABC) is a class of statistical methods, initiated in \cite{pritchard1999population}, that allows one to infer $\pi$ in situations where the likelihood $f$ is intractable but 
forward simulation of pseudo data $Z\sim f(\,\cdot\,|\,\theta)$ is doable. More precisely, the algorithm consisting of (i) $\vartheta\sim p$, (ii) $Z\sim f(\,\cdot\,|\,\vartheta)$ and 
(iii) set $\theta=\vartheta$ only if $\{Z=Y\}$, does produce a sample $\theta$ whose distribution is $\pi(\,\cdot\,|\,Y)$. Regarding the situation $N\to\infty$ as a source of intractability, one could attempt to borrow from ABC to sample from $\pi$. However, since $N\to\infty$, sampling from the likelihood model is impossible and a natural idea is to replace step (ii) by drawing subsamples $Y_U$ ($U\in\Uset_n$), leading to what we refer as Informed Sub-Sampling, as opposed to Informed Sub-Sampling MCMC described in the previous Subsection. Obviously, the event $\{Y_U=Y\}$ is impossible except in the trivial situation where $N=n$. Overcoming situations where $\{Y=Z\}$ is impossible or very unlikely has already been addressed in the ABC literature (see \cite{fearnhead2012constructing} and \cite{wilkinson2013approximate}), leading to approximate ABC algorithms. In particular, step (iii) is replaced by a step that sets $\theta=\vartheta$ with probability $\propto\exp\{-\eps\|S(Z)-S(Y)\|^2\}$ where $S$ is a vector of summary statistics and $\eps>0$ a tolerance parameter. We build on this analogy to propose a noisy Informed Sub-Sampling algorithm, see Table \ref{tab:abc} for more details.

The Noisy ABC algorithm replaces direct inference of $\pi$ by the following surrogate distribution
\begin{equation}
\label{eq:abc1}
\hpi_{\text{ABC}}(\rmd\theta\,|\,Y):\propto p(\rmd \theta)\hat{f}_{\text{ABC}}(Y\,|\,\theta)=p(\rmd\theta)\int f(\rmd Z\,|\,\theta)\exp\{-\eps\|S(Z)-S(Y)\|^2\}\,,
\end{equation}
where the exact likelihood is replaced by $\hat{f}_{\text{ABC}}$. Similarly, the approximation of $\pi$ stemming from Informed Sub-Sampling is:
\begin{multline}
\label{eq:abc2}
\hpi_{n}(\rmd\theta\,|\,Y):\propto p(\rmd\theta)\hat{f}(Y\,|\,\theta)=\\
p(\rmd\theta)\sum_{U\in\Uset_n} f^{(N/n)}(Y_U\,|\,\theta)\exp\{-\eps\|(N/n)S(Y_U)-S(Y)\|^2\}\,.
\end{multline}

This analogy shows that there is a connection between the ABC and the Informed Sub-Sampling in the way both approximate $\pi$, see \eqref{eq:abc1} and \eqref{eq:abc2}. However, since sampling from $\nu_{n,\eps}$ and $\pi_n(U)$ are not feasible, this approach cannot be considered, hence motivating the use of Markov chains instead, \ie Informed Sub-Sampling MCMC. Moreover, quantifying the approximation of $\pi$ by $\hpi_n$ \eqref{eq:abc2} is technically challenging while resorting to the Informed Sub-Samping Markov chain allows to use the Noisy MCMC framework developed in \cite{alquier2014noisy} to quantify this approximation. This is the purpose of the following Section.

\begin{table}
\centering
\begin{tabular}{c|c|c|c|c}

step  & \multicolumn{2}{|c|}{ABC} &  \multicolumn{2}{c}{Informed Sub-Sampling}\\
 \hline
(i)&\multicolumn{2}{c|}{$\vartheta\sim p$}& \multicolumn{2}{c}{-}\\
\hline
(ii)&\multicolumn{2}{c|}{$Z\sim f(\,\cdot\,|\,\vartheta)$} &  \multicolumn{2}{c}{$Z=Y_U$, $U\sim \unif(\Uset)$}\\
\hline
\multirow{4}{*}{(iii)} & exact & noisy & exact & noisy\\
 \cline{2-5}

& if $Z=Y$,  & with proba. $\propto$ & if $Z=Y$ & with proba. $\propto$\\
& set $\theta=\vartheta$ & $e^{-\eps\|S(Y)-S(Z)\|^2}$ & draw  $\theta\sim\pi(U)$  & $e^{-\eps\|S(Y)-(N/n)S(Y_U)\|^2}$\\
&&set $\theta=\vartheta$&  &draw $\theta\sim\pi(U)$\\
\end{tabular}
\caption{Comparison between ABC and Informed Sub-Sampling, an adaptation of ABC designed for situations where $N\gg 1$ and likelihood simulation is not possible. The exact algorithms provide samples from $\pi$ while the noisy algorithms sample from approximation of $\pi$ given in \eqref{eq:abc1} and \eqref{eq:abc2}.\label{tab:abc}}
\end{table} 

\section{Theoretical Analysis of Informed Sub-Sampling MCMC}
\label{sec:conv}
By construction, ISS-MCMC samples a Markov chain on an extended state space $\{(\ttheta_i,U_i),\,i\in\nset\}$ but the only useful outcome of the algorithm for inferring $\pi$ is the marginal chain $\{\ttheta_i,\,i\in\nset\}$. In this section, we study the distribution of the marginal chain and denote by $\ttarg_i$ the distribution of the random variable $\ttheta_i$. Note that $\{\ttheta_i,\,i\in\nset\}$ is identical to the Metropolis-Hastings chain $\{\theta_i,\,i\in\nset\}$, up to replacing $\alpha$ by $\talpha$ in the accept/reject step. This change, from which the computational gain of our method originates, has important consequences on the stability of the Markov chain and, in particular, implies that $\pi$ is not the stationary distribution of $\{\ttheta_i,\,i\in\nset\}$. Interest lies in quantifying the distance between $\ttarg_i$ and $\pi$. In this paper, our results are expressed in total variation distance but the recent works of \citet{rudolf2015perturbation} and \citet{johndrow2017error} suggest that carrying out the analysis using the Wasserstein metric may lead to more accurate bounds when $\Theta$ is not compact. We first recall the definition of the total variation distance which, for two distributions with density function $\pi$ and $\ttarg_i$ respectively w.r.t. the same common dominating measure, denoted $\rmd\theta$, can be expressed as
$$
\|\pi-\ttarg_i\|=(1/2)\int_\Theta|\pi(\theta)-\ttarg_i(\theta)|\rmd \theta\,.
$$

\subsection{Assumptions}
Let $K$ denote the \textit{exact} MH transition kernel, with proposal $Q$, described in Algorithm \ref{alg:mh}. $Q$ is fixed and set as a random walk kernel that achieves a reasonable acceptance rate. By construction, $K$ is $\pi$-reversible and thus $\pi$-invariant. Moreover, $K$ is assumed to be ergodic \ie $\|K(x,\,\cdot\,)-\pi\|\to 0$ at a geometric rate and the convergence is either simple (Assumption \textbf{A}.\ref{assu:2}) or uniform (Assumption \textbf{A}.\ref{assu:3}).

\begin{assumption}[Geometric ergodicity]
\label{assu:2}
There exists a constant $\varrho\in(0,1)$ and a function $C:\Theta\to\rset^+$ such that for all $(\theta_0,i)\in\Theta\times\nset$
\begin{equation}
\|\pi-K^i(\theta_0,\cdot)\|\leq C(\theta_0)\varrho^i\,.
\end{equation}
\end{assumption}

\begin{assumption}[Uniform ergodicity]
\label{assu:3}
There exists two constants $C<\infty$ and $\varrho\in(0,1)$ such that for all $i\in\nset$
\begin{equation}
\label{eq:assu:3}
\sup_{\theta_0\in\Theta}\|\pi-K^i(\theta_0,\,\cdot\,)\|\leq C\varrho^i\,.
\end{equation}
\end{assumption}

As observed in Remark \ref{rem1} (\ref{app4_bis}), the ISS-MCMC marginal Markov chain $\{\ttheta_i,\,i\in\nset\}$ is time inhomogeneous. Indeed, conditionally on $\ttheta_i$, the probability of the transition $\ttheta_i\to \ttheta_{i+1}$ depends on the iteration index $i$. This complicates the analysis of ISS-MCMC as most results on perturbation of Markov chains are established for time homogeneous Markov chains. For simplicity, we present in this section an analysis of a slight variation of ISS-MCMC that assumes independence between the different subsets $\{U_i,\,i\in\nset\}$ (Assumption \textbf{A}.\ref{assu:1}).
\begin{assumption}[IID subsets]
\label{assu:1}
The subsets $U_1,U_2,\ldots$ are independent and identically distributed under $\nu_{n,\eps}$.
\end{assumption}

In practice, Assumption \textbf{A.}\ref{assu:1} is satisfied when steps (3)-(9) of Algorithm \ref{alg:infsub} are repeated a large number of times to simulate $U_{i+1}$ given $U_i$. Under \textbf{A.}\ref{assu:1}, $\{\ttheta_i,\,i\in\nset\}$ is a time homogeneous Markov chain whose transition kernel $\tK_{n,\eps}$ is
\begin{equation}
\label{eq:8}
\forall\,(\ttheta,A)\in\Theta\times\vartheta,\qquad\tK_{n,\eps}(\ttheta,A)=\sum_{u\in\Uset_n}K(\ttheta,A\,|\,u)\nu_{n,\eps}(u)\,,
\end{equation}
where for all $\theta\in\Theta$, $K(\theta,\,\cdot\,|\,u)$ is the \textit{exact} MH transition kernel conditionally on some $\theta\in\Theta$ with proposal $Q$ that targets $\tpi_n(\,\cdot\,|\,Y_u)$ \eqref{eq:subposterior}.

The results stated in this section hold under Assumption \textbf{A.}\ref{assu:1}. We nevertheless note that this assumption might be relaxed. In particular, we show how the analysis carried out under uniform ergodicity assumption can be extended even if \textbf{A.}\ref{assu:1} does not hold, see \ref{app4_bis}. In the geometric ergodic case, a similar extension may be doable, see \eg \cite[Theorem 8]{douc2004quantitative}, but this is out of the scope of this paper. In general, the perturbation bounds with time inhomogeneous kernels are more obscure to interpret. Note that the numerical illustrations of ISS-MCMC presented at Section \ref{sec:sim} were performed without satisfying Assumption \textbf{A.}\ref{assu:1}, \ie implementing Algorithm \ref{alg:infsub}, and lead to satisfactory results.

Finally, we consider the following assumption for the summary statistics mapping. This assumption is motivated at two levels. First, it is necessary to have some assumptions on the summary statistics to derive theoretical results for ISS-MCMC in absence of sufficient statistics. Second, it offers a way to validate empirically the choice of summary statistics for a given model, see Section \ref{sec:sim}.
\begin{assumption}[Summary Statistics]
\label{assu:4}
There exists a constant $\gamma_n<\infty$, such that for all $(\theta,U)\in\Theta\times \Uset_n$
\begin{equation}
\label{eq:assu4}
\left|\log f(Y\,|\,\theta)-(N/n)\log f(Y_U\,|\,\theta)\right|\leq \gamma_n N \|\bS(Y)-\bS(Y_U)\|\,.
\end{equation}
\end{assumption}

Assumption \textbf{A.}\ref{assu:4} imposes a condition simultaneously on the model $f$ and the summary statistics $S$. In particular, it assumes that for any $\theta\in\Theta$, the variation of the scaled likelihood of the subsamples $Y_U$ $(U\in\Uset_n)$ around $f(Y\,|\,\theta)$ is controlled by the distance between the full dataset $Y$ and the subsample $Y_U$, as measured through their summary statistics. This is a strong assumption which is unlikely to hold if $\Theta$ is not a compact set. It implies that even in absence of sufficient statistics, a result similar to Proposition \ref{prop:esp} exists. One can also note that when $n\to N$, the constant $\gamma_n$ in Eq. \eqref{eq:assu4} goes to zero.

\subsection{$K$ is geometrically ergodic}

Our main result is that for a sufficiently large size of subsample $n$, ISS-MCMC admits a stationary distribution. This follows from an adaptation of the work of \citet{medina2016stability} to the context of ISS-MCMC.

\begin{prop}
\label{prop:geo}
Assume that assumptions \textbf{A.}\ref{assu:2}, \textbf{A.}\ref{assu:1}  and \textbf{A.}\ref{assu:4} hold, then there exists an $n_0\leq N$ such that for all $n>n_0$, $\tK_{n,\eps}$ is also geometrically ergodic for all $\epsilon>0$.
\end{prop}

The proof is outlined to \ref{proof:geo}.

A direct consequence of Proposition \ref{prop:geo}, see for instance \cite[Theorem 16.0.1]{meyn2009markov}, is that for $n$ sufficiently large, $\tK_{n,\eps}$ admits a stationary distribution and that this stationary distribution converges to $\pi$ as $n\to N$. In most cases, it is difficult to obtain a rate of convergence under the assumption that $K$ is geometrically ergodic. We nevertheless note that this rate is related to rate of convergence of $\gamma_n$ to $0$ as hinted by \citet[Theorem 4.1]{medina2016stability}.

\subsection{$K$ is uniformly ergodic}

In addition to admitting a stationary distribution for a large enough $n$, we now show that under the assumption of uniform ergodicity it is possible to quantify the rate of convergence. Our main result follows from an adaptation of the work of \citet{alquier2014noisy} to the context of ISS-MCMC.
\begin{prop}
\label{prop:bound}
Define
\begin{equation}
\label{eq:d2}
A_n:=\esp\left\{\sup_{\theta\in\Theta}\frac{1}{\phi_U(\theta)}\right\}=\sum_{U\in\Uset_n}\nu_{n,\eps}(U)\sup_{\theta\in\Theta}\frac{f(Y\,|\,\theta)}{f(Y_U\,|\,\theta)^{N/n}}\,,
\end{equation}
where for all $(\theta,U)\in(\Theta\times \Uset_n)$, we have set $\phi_U(\theta):=f(Y_U\,|\,\theta)^{N/n}\slash f(Y\,|\,\theta)$ and
\begin{multline}
\label{eq:B_n}
B_n(\theta,U):=\esp\{a(\theta,\theta')|\phi_U(\theta)-\phi_U(\theta')|\}=\int Q(\theta,\rmd\theta')a(\theta,\theta')|\phi_U(\theta)-\phi_U(\theta')|\,.
\end{multline}
Assume that assumptions \textbf{A.}\ref{assu:3}, \textbf{A.}\ref{assu:1} and \textbf{A.}\ref{assu:4} hold, then there exists a constant $\kappa<\infty$ such that for all $i\in\nset$
\begin{equation}
\label{eq:bound2}
\|K^i(\theta_0,\,\cdot)-\tK_{n,\eps}^i(\theta_0,\,\cdot)\|\leq\kappa A_n
\sup_{(\theta,U)\in\Theta\times\Uset_n}B_n(\theta,U)\,,
\end{equation}
and
\begin{equation}
\label{eq:bound2_bis}
\lim_{i\to\infty}\sup_{\theta\in\Theta}\|\pi-\tK_{n,\eps}^i(\theta,\,\cdot\,)\|=\kappa A_n
\sup_{(\theta,U)\in\Theta\times\Uset_n}B_n(\theta,U)\,.
\end{equation}
Moreover, for a large enough subset size $n$, the marginal Markov chain produced by ISS-MCMC admits an invariant distribution $\tpi_n$ that satisfies
\begin{equation}
\label{eq:bound3}
\|\pi-\tpi_n\|\leq\kappa A_n
\sup_{(\theta,U)\in\Theta\times\Uset_n}B_n(\theta,U)\,.
\end{equation}
\end{prop}
The proof of Proposition \ref{prop:bound} is postponed to \ref{app4}. Note that an extension of this result to the case where \textbf{A.}\ref{assu:1} does not hold is presented at \ref{app4_bis}.

Since for any two measures $(\mu,\mu')$, $\|\mu-\mu'\|\leq 1$, the upper bounds of Proposition \ref{prop:bound} are only informative if there are smaller than 1. Those bounds are a product of two expectations. We now show how those two expectations are controlled respectively through the choice of proposal kernel and the choice of summary statistics.

\subsubsection{Choice of the proposal kernel}
\label{sec:alg:prop}
Assuming a Gaussian random walk proposal with covariance matrix $\Sigma\T\Sigma$, $B_n$ can be expressed as $B_n(\theta)=\sup_{U\in\Uset_n}D_1(U,\theta)$ where $D_1$ is defined as
\begin{equation}
\label{eq:boundB}
D_1(U,\theta):=\int \Phi_{d}(\rmd \zeta)\frac{\pi(\theta+\Sigma\zeta)}{\pi(\theta)}|\phi_U(\theta)-\phi_U(\theta+\Sigma\zeta)|\,,
\end{equation}
where $\Phi_d$ is the standard Gaussian distribution in dimension $d=\text{dim}(\Theta)$. When $N\gg1$, the Bernstein-von Mises theorem states that, under conditions on the likelihood function, the posterior distribution can be approximated by a Gaussian with mean set as the maximum likelihood estimator $\thetaMLE$ and covariance $I(\theta_0)^{-1}/N$ where $I$ is the Fisher information matrix and $\theta_0\in\Theta$ some parameter. Since ISS-MCMC aims at sampling from an approximation of $\pi$, setting $\Sigma=(1/\sqrt{N}) M$ where $M\T M$ is an approximation of $I(\theta_0)^{-1}$ is a reasonable choice. Proposition \ref{prop:bound2} shows that $D_1$ can, in this scenario, be arbitrarily brought down close to 0.

\begin{prop}
\label{prop:bound2}
Under the assumption that the proposal kernel $Q$ is a Gaussian Random Walk with covariance matrix $\Sigma=(1/\sqrt{N}) M$, we have
\begin{multline}
D_1(U,\theta)\leq \frac{\|\grad_\theta\phi_U(\theta)\|}{\sqrt{N}}\left\{\sqrt{\frac{2}{\pi}}\|M\|_1+
\frac{\|M\|_2^2\|\grad_\theta\log\pi(\theta)\|}{\sqrt{N}}\right\}\\
+\frac{d}{2N}\tn M\T\grad_\theta^2\phi_U(\theta)M\tn+\esp\{R(\|M\zeta\|/\sqrt{N})\}\,,
\end{multline}
where $R(x)=_{x\to 0}o(x)$ and for any square matrix $M$ of dimension $\rset^d$, we have set $\|M\|_1:=\sum_{1\leq i,j\leq d}|M_{i,j}|$, $\|M\|_2:=\{\sum_{1\leq i,j\leq d}M_{i,j}^2\}^{1/2}$ and $\tn \,\cdot\,\tn$ is the operator norm.
\end{prop}
The proof is postponed to \ref{app5}. Under regularity assumptions on the likelihood model, the gradient of $\log\pi$ and $\phi_U$ and the Hessian of $\phi_U$ are bounded and the upper bound of $D_1(U,\theta)$ can be brought down arbitrarily to $0$, uniformly in $(U,\theta)$, through $M$ when $N\gg 1$.

\subsubsection{Choice of the summary statistics}
\label{sec:alg:sumstat}
In Eq. \eqref{eq:d2}, the likelihood of each subsample is raised at the power $N/n$ (\ie typically several orders of magnitude) and therefore subsamples unlikely under $f(\,\cdot\,|\,\theta)$ will contribute to make $A_n$ very large. Ideally the choice of $S$ would guarantee that subsamples $Y_U$ having a very small likelihood $f(Y_U\,|\,\theta)$ are assigned to a weight $\nu_{n,\eps}(U)\approx 0$ to limit their contribution. In other words, $S$ should be specified in a way that prevents $f(Y_U\,|\,\theta)$ to go to 0 at a rate faster than $\nu_{n,\eps}(U)$. This is ensured if Assumption \textbf{A.}\ref{assu:4} holds. Indeed, in such a case
\begin{multline*}
A_n=\sum_{U\in\mathsf{E}_n(\theta)}\nu_{n,\eps}(U)\sup_{\theta\in\Theta}\frac{f(Y\,|\,\theta)}{f(Y_U\,|\,\theta)^{N/n}}+\sum_{U\in\Uset_n\backslash \mathsf{E}_n(\theta)}\nu_{n,\eps}(U)\sup_{\theta\in\Theta}\frac{f(Y\,|\,\theta)}{f(Y_U\,|\,\theta)^{N/n}}\\
\leq \nu_{n,\eps}\left(\mathsf{E}_n(\theta)\right)+\sum_{U\in\Uset_{n}\backslash \mathsf{E}_n(\theta)}\exp\{-\eps\|\Delta_n(U)\|^2+\gamma_n\|\Delta_n(U)\|-\log Z_n(\epsilon)\}\,,
\end{multline*}
where we have defined $\mathsf{E}_n(\theta):=\{U\in\Uset_n,\,\sup_{\theta\in\Theta}f(Y\,|\,\theta)/f(Y_U\,|\,\theta)^{N/n}<1\}$ and $Z_n(\epsilon)=\sum_{U\in\Uset_n}\exp\{-\eps\|\Delta_n(U)\|^2\}$. Clearly, if $\eps$ has the same order of magnitude as $\gamma_n$, each term of the sum remains bounded when $\|\Delta_n(U)\|\to\infty$. Conversely, setting $\eps=0$ is equivalent to choosing $\nu_{n,\eps}$ as the uniform distribution on $\Uset_n$ and may not allow to bound $A_n$, see Figure \ref{fig:upper} related to the probit example.

Potential summary statistics can be empirically validated by checking that they satisfy Assumption \textbf{A.}\ref{assu:4}. This validation can be performed graphically, by repeating the following operations for a number of parameters $\theta_k\sim_{\text{i.i.d}} p$:
 \begin{enumerate}[(i)]
\item draw subsets $U_1,U_2,\ldots$ uniformly at random in $\Uset_n$,
\item plot the points with coordinates
$$
(x_{k,i},y_{k,i})=(\|\Delta_n(U_i)\|,\log f(Y\,|\,\theta_k)-(N/n)\log f(Y_{U_i}\,|\,\theta_k))\,.
$$
 \end{enumerate}
The statistics are validated if there exists $\gamma_n<\infty$ such that the points $(x_{k,i},y_{k,i})$ satisfy $|y_{k,i}\slash x_{k,i}|\leq \gamma_n$, as illustrated at Figure \ref{fig:val}.

In situations where the maximum likelihood estimator $\thetaMLE(Y_{1:n})$ is easy and quick to evaluate numerically, we recommend setting $\bS(Y_{1:n})=\thetaMLE(Y_{1:n})$. In the case of independent observations of a well-specified model, setting the summary statistics as the maximum likelihood estimate is justified by the following Proposition which implies that Assumption \textbf{A.}\ref{assu:4} holds, asymptotically, up to a constant.
\begin{prop}
\label{propMLE}
We assume that the whole dataset comprises $N=\rho n$ independent observations and there exists some $\theta_0\in\Theta$ such that $Y_i\sim f(\,\cdot\,|\,\theta_0)$. Let $\thetaMLE$ be the MLE of $Y_1,\ldots,Y_{N}$ and $\thetaMLE_U$ be the MLE of the subsample $Y_U$ ($U\in\Uset_n$).
Then, there exists a constant $\beta$, a metric $\|\cdot\|_{\theta_0}$ on $\Theta$ and a non-decreasing subsequence $\{\sigma_n\}_{n\in\nset}$, ($\sigma_n\in\nset$) such that for all $U\subset\{1,2,\ldots,\rho\sigma_n\}$ with $|U|=\sigma_n$, we have for $p$-almost all $\theta$ in an neighborhood of $\theta_0$:
\begin{equation}
\label{eq0}
\log f(Y_{1:\rho \sigma_n}\,|\,\theta)-\rho\log f(Y_U\,|\,\theta)\leq
H_n(Y,\theta)+\beta+\frac{\rho n}{2}\|\thetaMLE_U-\thetaMLE\|_{\theta_0}\,,
\end{equation}
where
$$
\underset{n\to \infty}{\plim}\quad H_n(Y,\theta)\overset{\proba_{\theta_{0}}}{=} 0\,.
$$
\end{prop}
The proof is detailed in \ref{proof:propMLE} and follows from a careful application of a Bernstein-von Mises theorem. Note that an extension of Proposition \ref{propMLE} to cases where the observations are not independent may exist provided that a Bernstein-von Mises theorem holds for the model at hand, which is the case for dependent observations if the likelihood model satisfies local asymptotic normality conditions \cite{le1953some,le2012asymptotic}.

Note that since Assumption \textbf{A.}\ref{assu:4} is mostly used in Propositions \ref{prop:geo} and \ref{prop:bound} to guarantee that the log-likelihood ratio between the likelihood and the scaled likelihood of a subsample is bounded, the constant $\beta$ in Proposition \ref{propMLE} is not a major concern. In addition, it is straightforward to see that this constant vanishes when the subset size grows faster than the full dataset, \ie $\rho\downarrow 1$, once in the asymptotic regime of Eq. \eqref{eq0}.

We remark that Proposition \ref{propMLE} is in line with the results regarding optimal summary statistics for ABC established in \cite{fearnhead2012constructing}. The authors show that the quadratic error loss between the ABC estimate based on $\hpi_{\text{ABC}}$ (Eq. \ref{eq:abc1}) and the true parameter is minimized when setting the summary statistics as the posterior mean, a choice which asymptotically coincides with the maximum likelihood estimator.

Finally, we note that, similarly to any approximate MCMC method, ISS-MCMC does not guarantee Law of Large number for $\pi$-integrable functionals. However, assume that the MH chain $K$ is geometrically ergodic, it is straightforward to establish that for a large enough $n$,
\begin{equation}
\lim_{i\to\infty}\left|\frac{1}{i}\sum_{j=1}^i f(\ttheta_j)-\pi f\right|\leq 2\|f\|\|\pi-\tpi_{n,\eps}\|\,\qquad \text{a.s}\,,
\end{equation}
where $\pi f:=\int f\rmd\pi$ and $\|f\|=\sup_{\theta\in\Theta} |f(\theta)|$. If in addition, $K$ is uniformly ergodic Proposition \ref{prop:bound} can help to bound the asymptotic error that typically arises in MCMC estimation of $\pi f$. 

\section{Illustrations}
\label{sec:sim}
We evaluate the efficiency of ISS-MCMC on three different applications: inferring a time series observed at $N=10^6$ contiguous time steps, a logistic regression with $N=10^6$ observations and a Gaussian binary classification problem based on $N=10^7$ data.

\subsection{Implementation details of Informed Sub-Sampling MCMC}
\label{sec:6_1}
Before illustrating the ISS-MCMC algorithm on the different examples, we address a few technical implementation details.

\begin{itemize}
\item On the subset size $n$: this parameter is essentially related to the computational budget available to the user. In the following examples we have used $n\propto N^{1/2}$ which achieves a substantial computational gain at a price of a negligible asymptotic bias.
\item On the sufficient statistics $S$: to reduce the bias resulting from the Metropolis-Hastings approximation, $S$ should be constructed so that Assumption \textbf{A.}\ref{assu:4} holds.  If the maximum likelihood estimator $\theta^\ast(Y)$ is quick to compute then Proposition \ref{propMLE} suggests that setting $S(Y)=\theta^\ast(Y)$ will theoretically satisfy \textbf{A.}\ref{assu:4}. Other sufficient statistics mapping can be used, typically those arising in the Approximate Bayesian Computation literature. In any case, we recommend checking \textbf{A.}\ref{assu:4} graphically (see Section \ref{sec:conv}).
\item On the bandwidth parameter $\eps$: the theory shows that when $\eps \approx \gamma_n $, the asymptotic bias is controlled ($\gamma_n$ is the constant in \textbf{A.}\ref{assu:4}). In practice, this may prove to be too large and could potentially cause the algorithm to get stuck on a very small number of subsets. To avoid such a situation, we suggest monitoring the refresh rate of subsamples that should occur with probability of at least 1\%.
\item On the initial subset $U_0$: in theory, one would run a preliminary Markov chain $\{U_1^{(0)},\ldots,U_L^{(0)}\}$ (for some $L>0$) targeting $\nu_{n,\eps}$, and set $U_0=U_L^{(0)}$ in order for the results of Section \ref{sec:alg:sumstat} to hold. In practice, a more efficient approach is to use a simulating annealing Metropolis-Hastings algorithm, see \cite{geyer1995annealing}. It introduces a sequence of tempered distributions $\nu_k:=\nu_{n,\eps_k}$, such that $\eps_k=t_k \eps$ $(k\in\{1,\ldots,L\})$ where $t_1=0$ and $t_L=1$. The transition kernel of the $k$-th iteration of the preliminary Markov chain is designed to be $\nu_k$ invariant. This technique facilitates sampling from a proxy of $\nu_{n,\eps}$ in a relative short time period as the successive tempered distributions help identifying those subsamples belonging to the high probability sets of $\nu_{n,\eps}$.
\end{itemize}

\subsection{Inference of an AR(2) model}

\begin{exa}
\label{ex:ar}
An autoregressive time series of order 2 AR(2) $\{Y_k,\,k\leq N\}$ is defined recursively by:
\begin{equation}
\label{eq:ar}
\left\{
\begin{array}{l}
(Y_0,Y_1)\sim\mu:= \mathcal{N}_2(\mathbf{0}_2,\theta_3^2\,\text{Id}_2)\\
\\
Y_{n}=\theta_1 Y_{n-1}+\theta_2 Y_{n-2} + Z_{n}\,,\quad Z_n\sim \mathcal{N}(0,\theta_3^2)\,,\qquad\forall\,n\geq 2,
\end{array}
\right.
\end{equation}
where $\param\in \Theta\subset \mathbb{R}^3$. The likelihood of an observed time series for this model writes
\begin{equation}
\label{eq:arma_lkhd}
f(Y_{0:N}\,|\,\param)=\mu(Y_{0:1})\prod_{k=2}^{N}g(Y_k\,|\,Y_{k-1},Y_{k-2},\theta)\,,
\end{equation}
such that for all $k\geq 1$,
\begin{equation}
\label{eq:arma_lkhd2}
g(Y_k\,|\,Y_{0:k-1},\param)=\Phi_1(Y_k;\,\theta_1 Y_{k-1}+\theta_2 Y_{k-2} ,\theta_3^2)
\end{equation}
where $x\to\Phi_1(x;\,m,v)$ is the pdf of the univariate Gaussian distribution with mean $m$ and variance $v$.
\end{exa}

This model has been used in \citet[Section 7.2]{chib1995understanding} to showcase the Metropolis-Hastings (M-H) algorithm. We follow the same setup here and in particular we use the same true parameter $\paramst=(1,-.5,1)$, same prior distribution and proposal kernel $Q$ ; see \citet{chib1995understanding} for more details. We sampled a time series $\{Y_k,\,k\leq N\}$ according to \eqref{eq:ar}, with $N=10^6$ under $\paramst$. Of course, in such a setup, M-H is prohibitively slow to be used in practice to sample from $\pi$ as it involves evaluating the likelihood of the whole time series at each iteration. We nevertheless use M-H to obtain a ground truth of $\pi$.

For simplicity, we restrict the set of subsamples to $n$ contiguous observations:
$$
\left\{Y_{0:n-1},Y_{1:n},\ldots,Y_{N-n+1:N}\right\}\,.
$$
This induces a set of subset $\Vset_n\subset\Uset_n$ defined such that a subset $U\in\Vset_n$ is identified with its starting index, \ie for all $i\leq |\Vset_n|$, $\Vset_n\ni U_i:=\{i,i+1,\ldots,i+n-1\}$. Indeed, using such subsamples yields a tractable likelihood \eqref{eq:arma_lkhd} as otherwise, missing variables need to be integrated out, hence loosing the simplicity of our approach.

With some abuse of notation, the proposal kernel $R$ can be written as a transition kernel on the alphabet $\{0,\ldots,N-n+1\}$. It is defined in this example as:
\begin{equation}
\label{eq:def_R}
R(i;j)=\1_{i\neq j}\left\{\omega \frac{\exp{\left(-\lambda|j-i|\right)}}{\sum_{j\leq |\bar{\Uset}_n|,\, j\neq i}\exp{\left(-\lambda|j-i|\right)}}+(1-\omega)
\frac{1}{|\bar{\Uset}_n|-1}\right\}\,.
\end{equation}
The rationale is to propose a new subset through a mixture of two distributions: the first gives higher weight to local moves and the latter allows jumps to remote sections of the time series. In this example, we have used $\omega=0.9$ and $\lambda=0.1$. We study the efficiency of ISS-MCMC in function of $n$, $\eps$ and $S$.

For any subsample $Y_U$, $U\in\Vset_n$, we have set the summary statistics $\bS(Y_U)$ to the solution of the AR(2) Yule-Walker equations for the dataset $Y_U$. As shown in Figure \ref{fig:val}, this choice of summary statistics satisfies (graphically) \textbf{A.}\ref{assu:4} with $\gamma\approx 5\,10^6$. We therefore set $\epsilon=5.0\,10^6$ to make sure that $A_n$ \eqref{eq:d2} is bounded. Theoretically, Proposition \ref{prop:bound} guarantee that the bias is controlled. This is illustrated graphically in Figure \ref{fig:ar1} where $\pi$ is compared to $\tpi_i$ ($i=50,000$). We also report the distribution of the Informed Sub-Sampling chain when the $\epsilon=0$, \ie when the subsampling is actually uninformed and all subsamples have the same weight. In the latter case, $A_n$ is not bounded which explains why the bias on $\lim_{i\to\infty}\|\pi-\tpi_i\|$ is not controlled. Figures \ref{fig:ar1} and \ref{fig:ar2} illustrate the distribution of $\{\tilde{\theta}_i,\,i\in\nset\}$ for some runs of ISS-MCMC with $\epsilon=0$ and $\epsilon=5.0\,10^6$. Finally, Figure \ref{fig:cpu3} gives a hint at the computational efficiency of ISS-MCMC. Metropolis-Hastings was compared to ISS-MCMC with $n\in\{1,000;5,000;10,000\}$ and $\epsilon\in\{0;1;5.0\,10^6\}$. The performance indicator is defined as the average of the marginals Total Variation distance, \ie
$$
\TV(t)=\frac{1}{d}\sum_{j=1}^d\|\pi^{(j)}-\tpi_t^{(j)}\|\,,
$$
where $\pi^{(j)}$ and $\tpi_t^{(j)}$ are respectively the true $j$-th marginal and the $j$-th marginal of the chain distribution after a runtime of $t$ seconds. The true marginals were estimated from a long Metropolis-Hastings chain, at stationarity. $\tpi_t^{(j)}$ was estimated using 500 independent chains starting from the prior. The Matlab function \texttt{ksdensity} in default settings was applied to estimate $\|\pi^{(j)}-\tpi_t^{(j)}\|_\TV$ from the chains samples, hence the variability. On the one hand, when $\epsilon=0$, there is no informed search for subsamples which makes the algorithm much faster than the other setups but yields a significant bias (larger than $0.5$). On the other hand, setting $\epsilon>0$ adds to the computational burden but allows to reduce the bias. In fact, for $n=5,000$ and a computational budget of $t=1,000$ seconds, the bias of ISS-MCMC is similar to that of Metropolis Hastings but converges 100 times as fast. Finally, note that as expected by the theory, when the unrepresentative subsamples are not penalized enough (\eg by setting $\epsilon=1$), ISS-MCMC yields a significant bias which hardly improves on uninformed subsampling when $n=10,000$. Following Assumption \textbf{A.}\ref{assu:4}, we see that setting $\epsilon=5.0\,10^6$ significantly reduces the bias. Note that setting $\epsilon>5.0\,10^6$ could potentially reduce further the bias but may fail the algorithm: indeed, when $\epsilon$ is too large the chain $\{U_k,\,k\in\nset\}$ gets easily stuck on a set of the best subsamples (for this choice of summary statistics) and may considerably slow down the convergence of the algorithm.

Other choices of summary statistics can be considered. Since $Y$ is modelled as an autoregressive time series, an option would be to set the summary statistics as the empirical autocorrelation function. Figure \ref{fig:badSS} shows that it is not a recommended choice. The left panel suggests that Assumption \textbf{A.}\ref{assu:4} does not hold for this type of summary statistics: good subsets yield a large value for $\log f(Y\,|\,\theta)-(N/n)\log f(Y_U\,|\,\theta)$ and conversely for bad subsets, hence generating a bias (see Eq. \eqref{eq:bound2_bis}). As a consequence, the right panel which shows a clear mismatch between $\pi^{(3)}$ and the Informed Sub-Sampling third marginal is not surprising.

\begin{figure}
\centering
\includegraphics[scale=0.8]{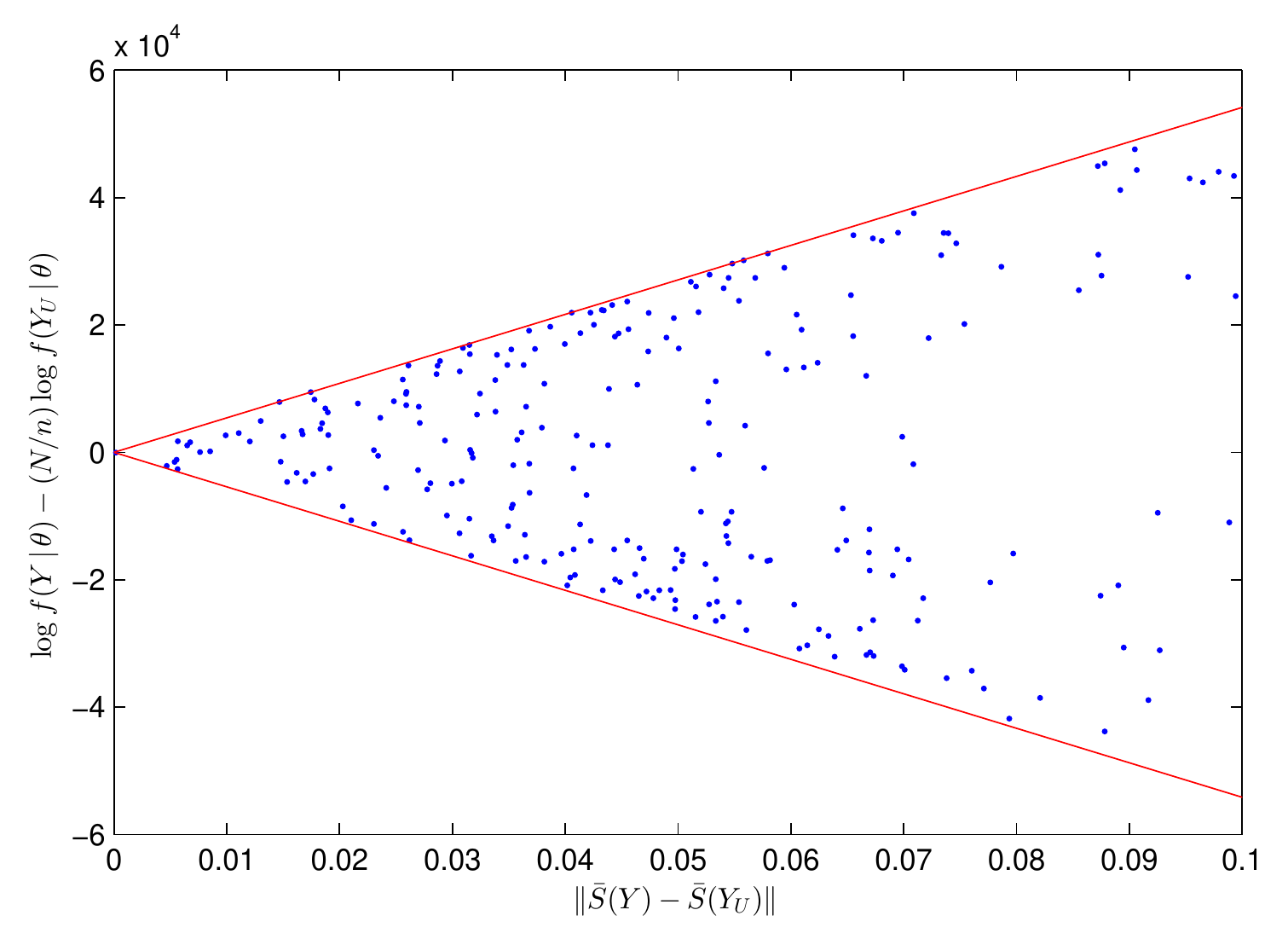}
\caption{(Example \ref{ex:ar}: Autoregressive time series) Validation of summary statistics set as the solution of Yule-Walker equation, with $n=1,000$. This choice of sufficient statistics satisfies Assumption \textbf{A.}\ref{assu:4} with $\gamma_n\approx 5\,10^{6}$. Each dot corresponds to a point $(\log f(Y\,|\,\theta)-(N/n)\log f(Y_U\,|\,\theta),\|\bar{\Delta}_n(U)\|)$ where $\theta$ and $U$ were respectively drawn from the prior $p$ and uniformly at random in $\bar{\Uset}_n$. The red lines allow to estimate the parameter $\gamma_n$ of \textbf{A.}\ref{assu:4}. \label{fig:val}}
\end{figure}

\begin{figure}
\centering

\includegraphics[scale=.8]{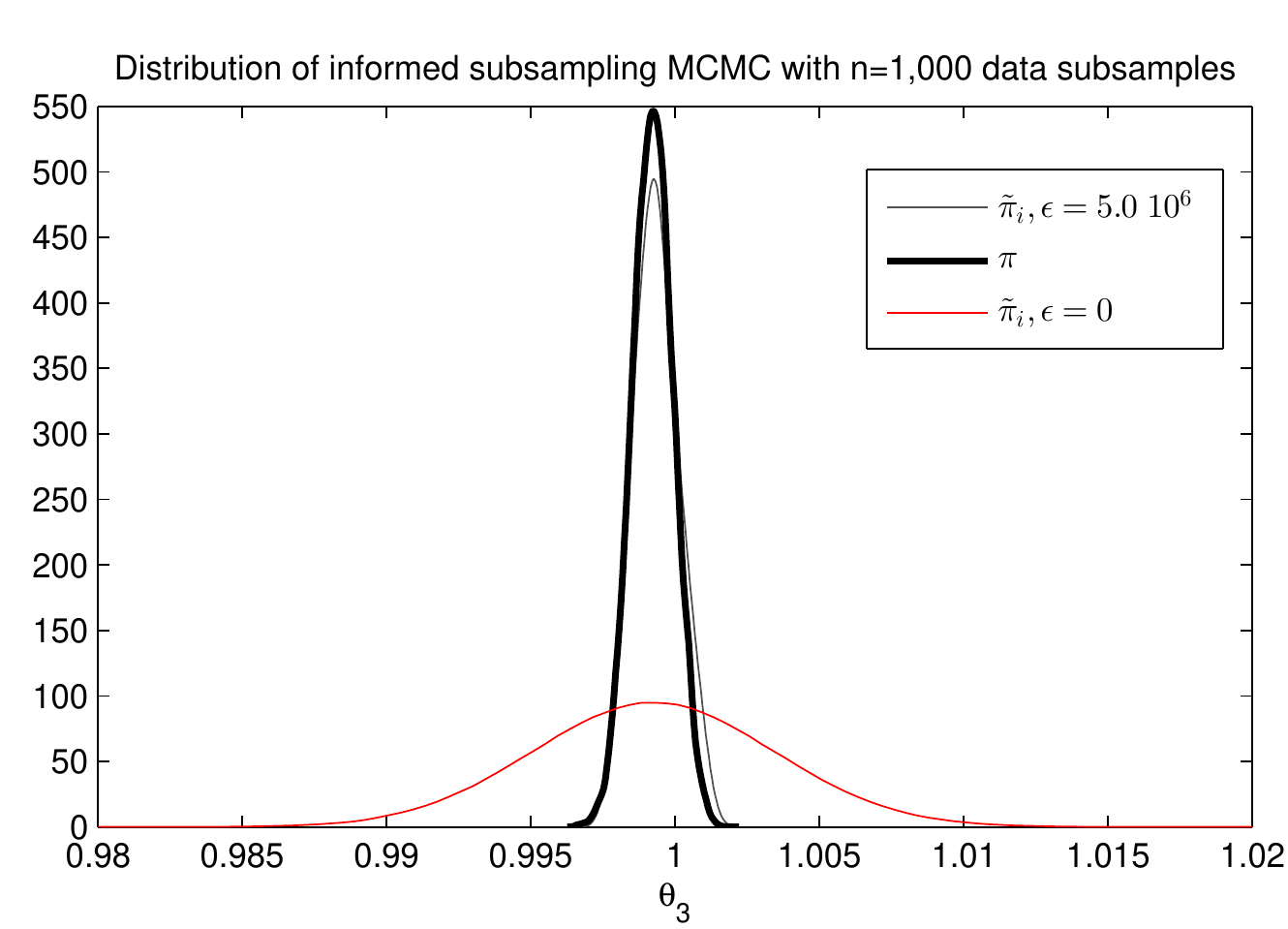}
\caption{(Example \ref{ex:ar}: Autoregressive time series)
Inference of the noise parameter with ISS-MCMC, using subsets comprising of $n=1,000$ contiguous time steps of a $N=10^6$ time-series. The plot represents the distributions $\tpi_i$  ($i=50,000$) of the Informed Sub-Sampling Markov chain for two different values of $\epsilon\in\{0, 5.0\;10^{6}\}$. These distributions where obtained from the replication of $1,000$ independent chains.\label{fig:ar1}}
\end{figure}

\begin{figure}
\centering
\includegraphics[scale=0.8]{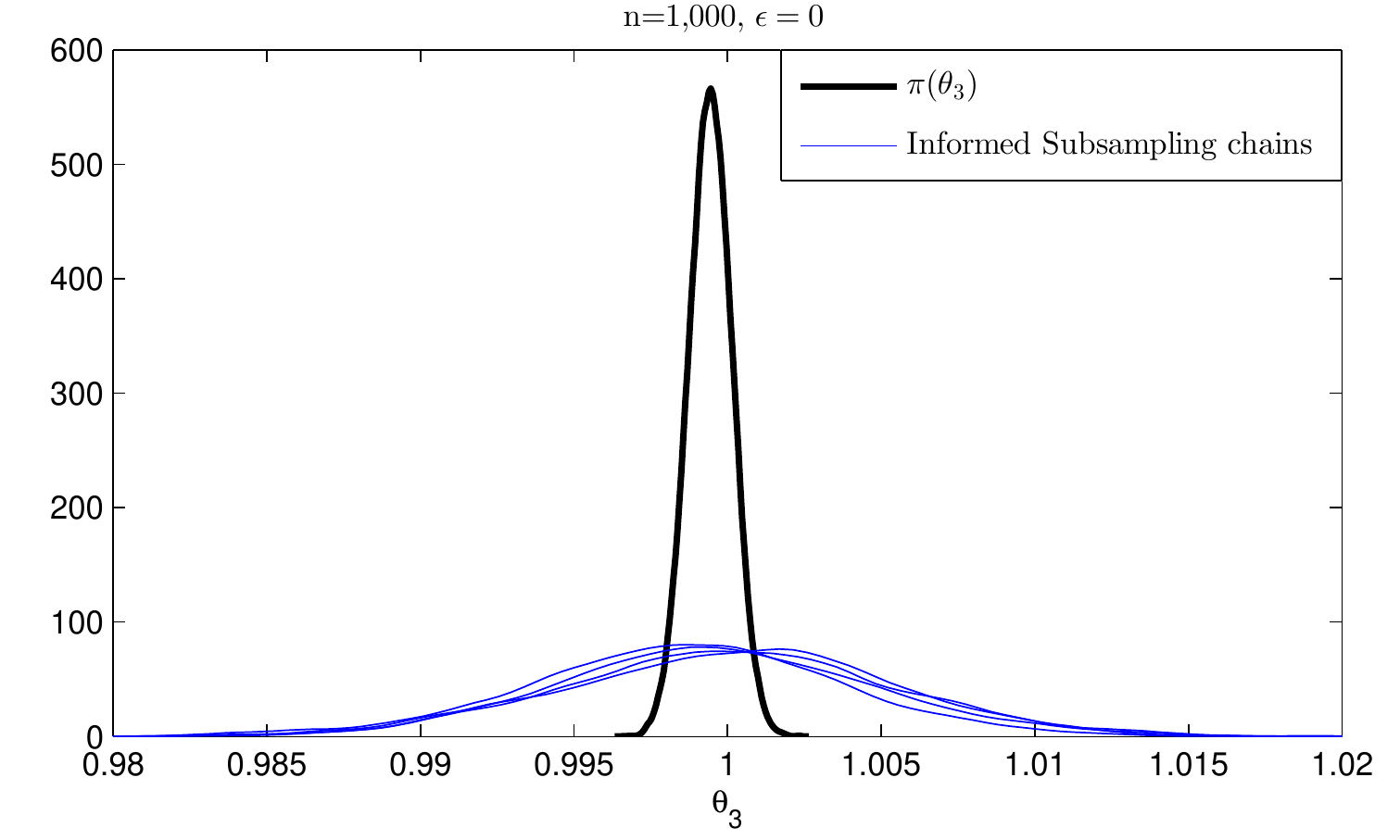}
\caption{(Example \ref{ex:ar}: Autoregressive time series) Marginal distribution of $\theta_3$ and distribution of $\{\ttheta_i,\,i\in\nset\}$ for four independent Informed Sub-Sampling Markov chains with $\eps=0$ and $n=1,000$.\label{fig:ar2}}
\end{figure}

\begin{figure}
\centering
\includegraphics[scale=0.8]{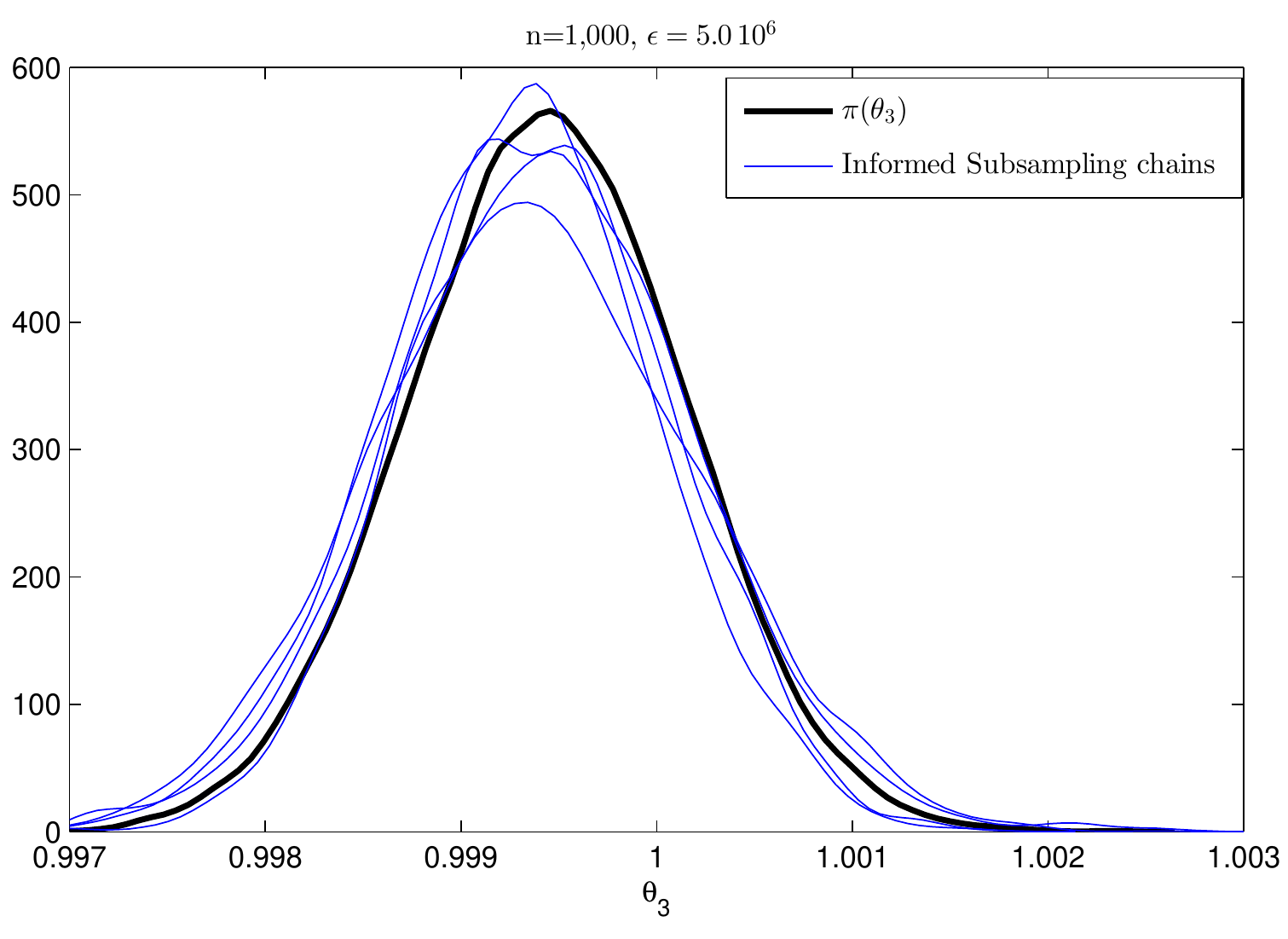}
\caption{(Example \ref{ex:ar}: Autoregressive time series) Marginal distribution of $\theta_3$ and distribution of $\{\ttheta_i,\,i\in\nset\}$ for four independent Informed Sub-Sampling Markov chains with $\eps=5.0\,10^{6}$ and $n=1,000$.}
\end{figure}

\begin{figure}
\centering
\includegraphics[scale=.52]{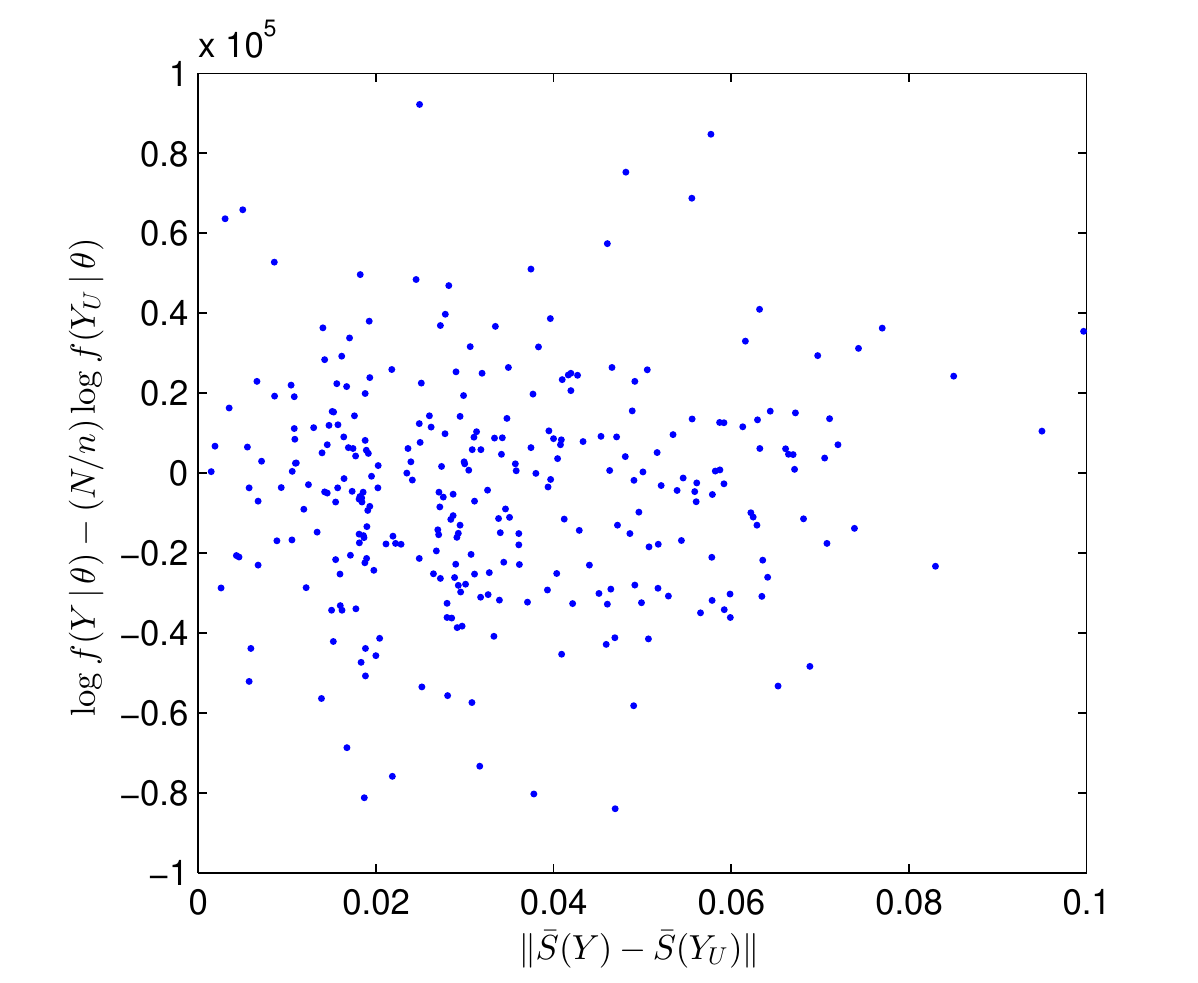}
\includegraphics[scale=.52]{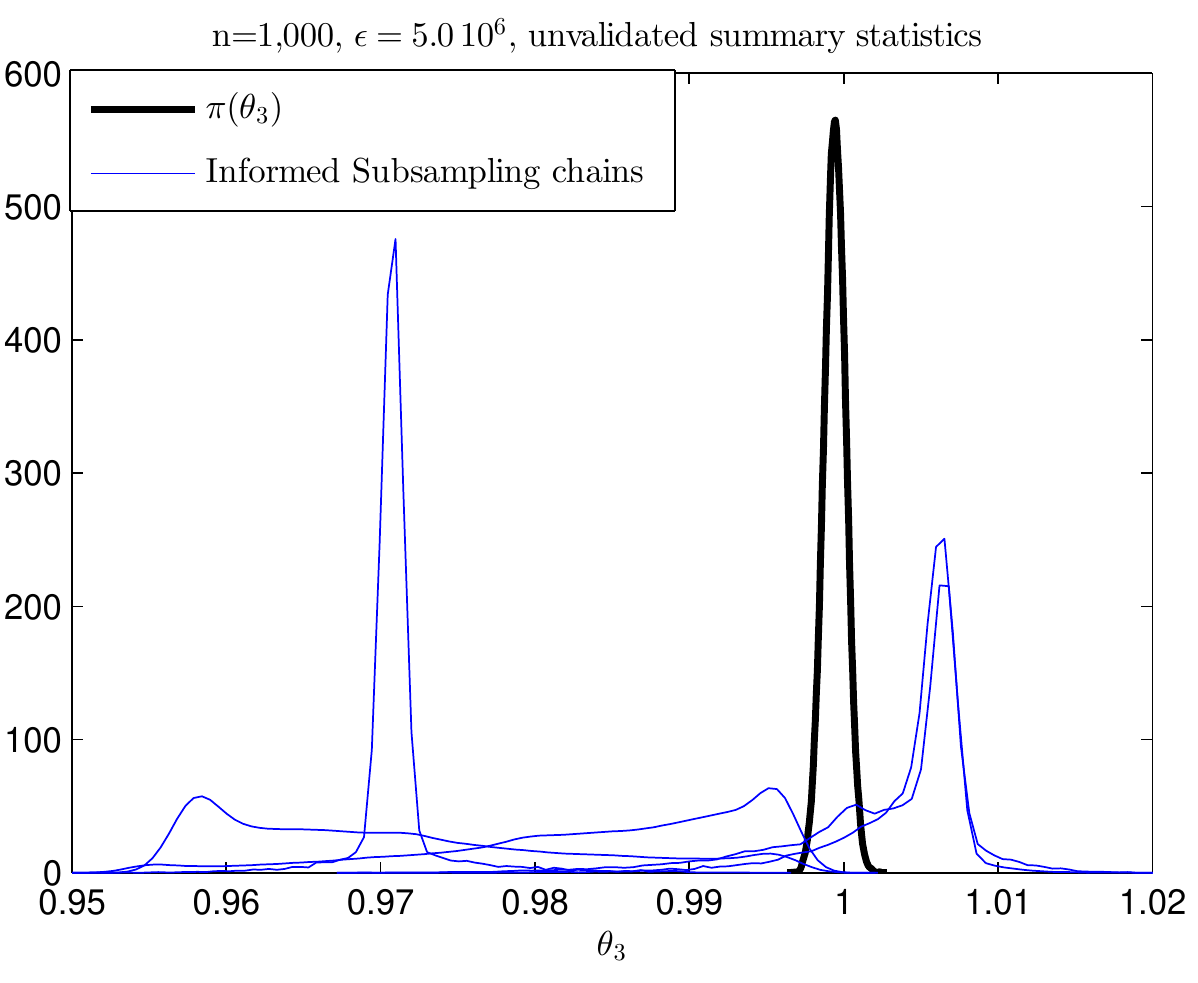}
\caption{(Example \ref{ex:ar}: Autoregressive time series) In this case, the summary statistics were defined as the first 5 empirical autocorrelation coefficients. The left panel shows that this is not a recommended choice and the right panel illustrates the distribution of $\{\ttheta_i,\,i\in\nset\}$ for four independent Informed Sub-Sampling Markov chains ($\eps=5.0\,10^{6}$,  $n=1,000$ and this choice of summary statistics), yielding an obvious mismatch.\label{fig:badSS}}
\end{figure}

\begin{figure}
\centering
\includegraphics[scale=.5]{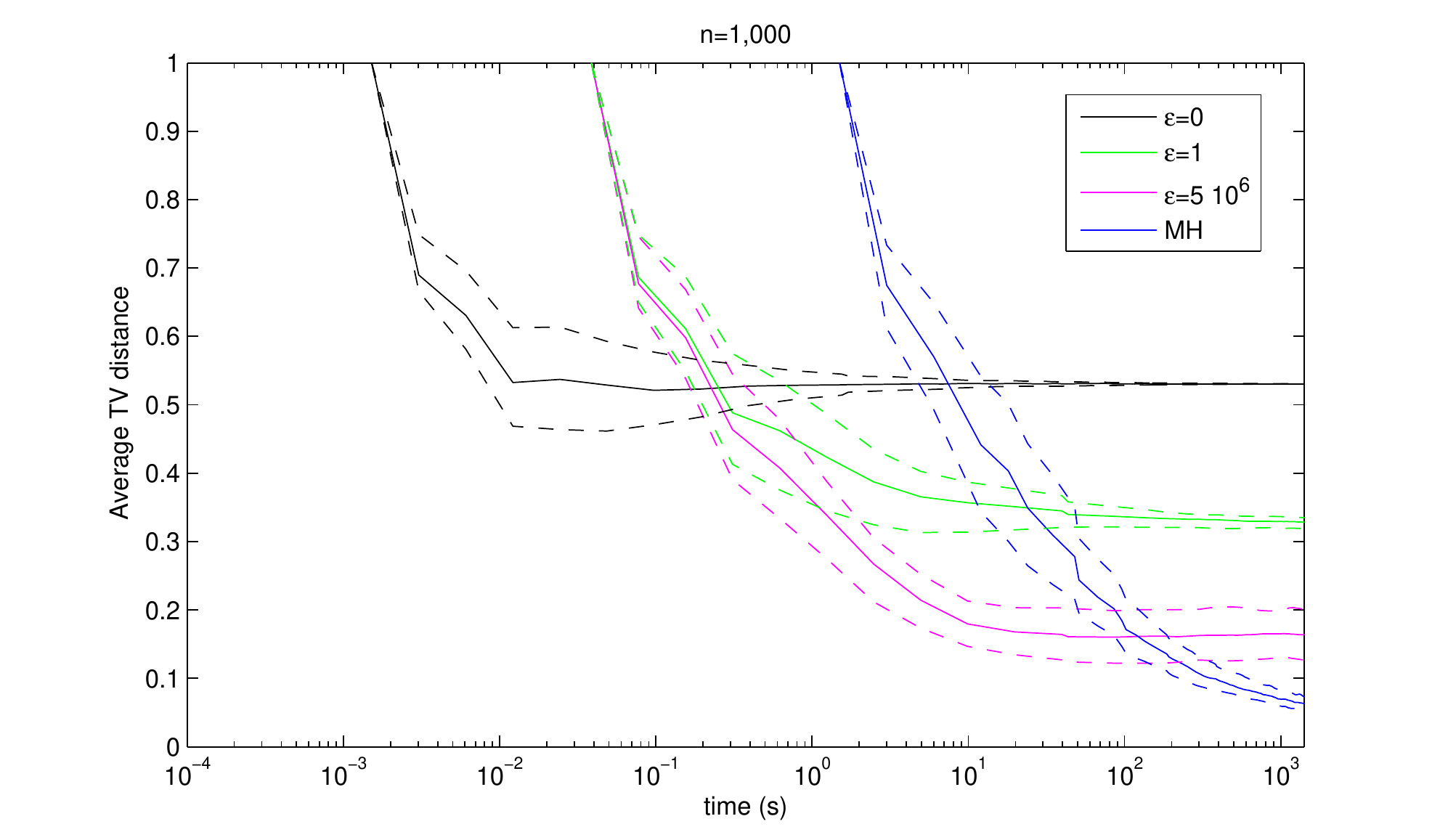}
\includegraphics[scale=.5]{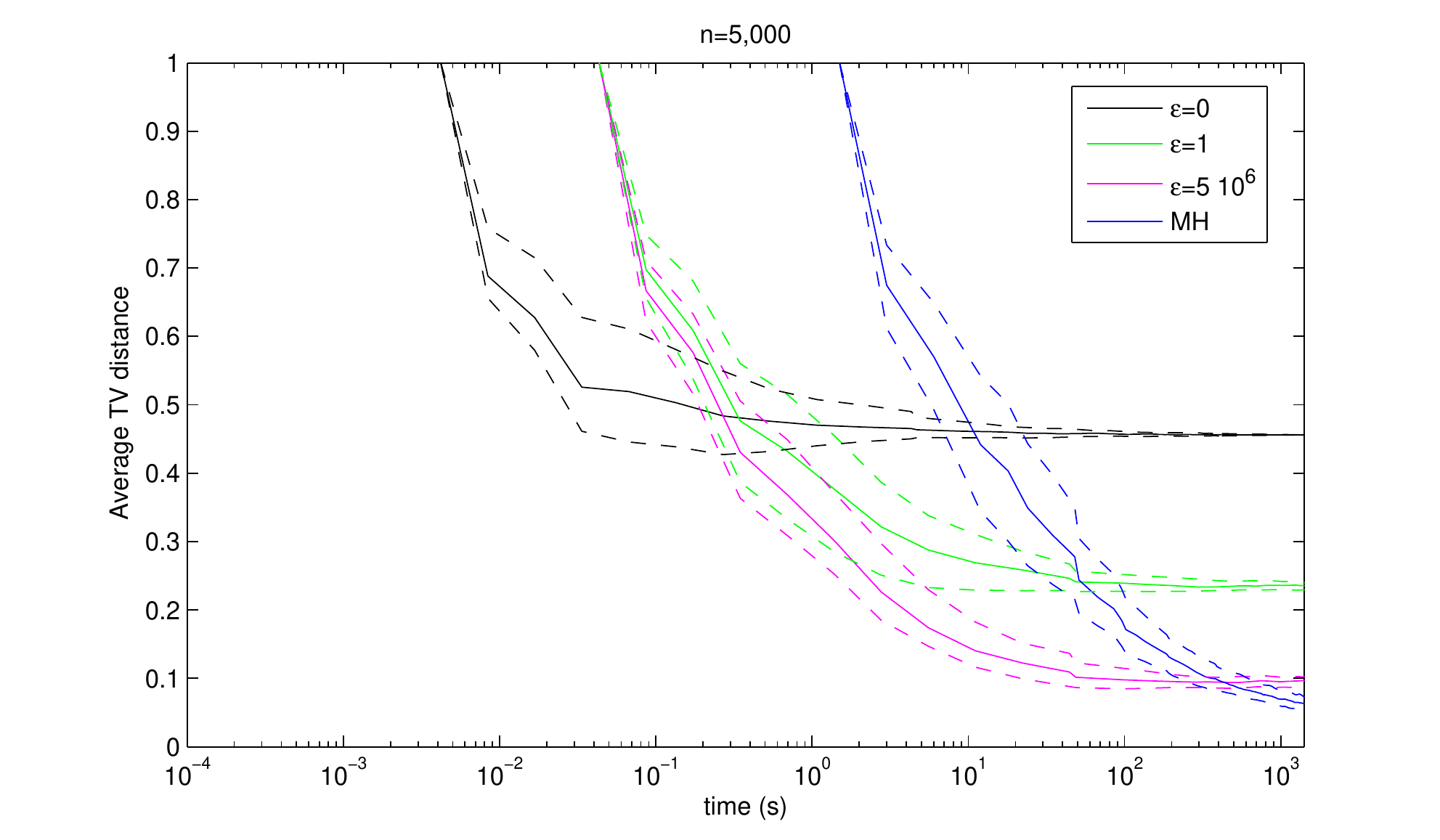}

\includegraphics[scale=.5]{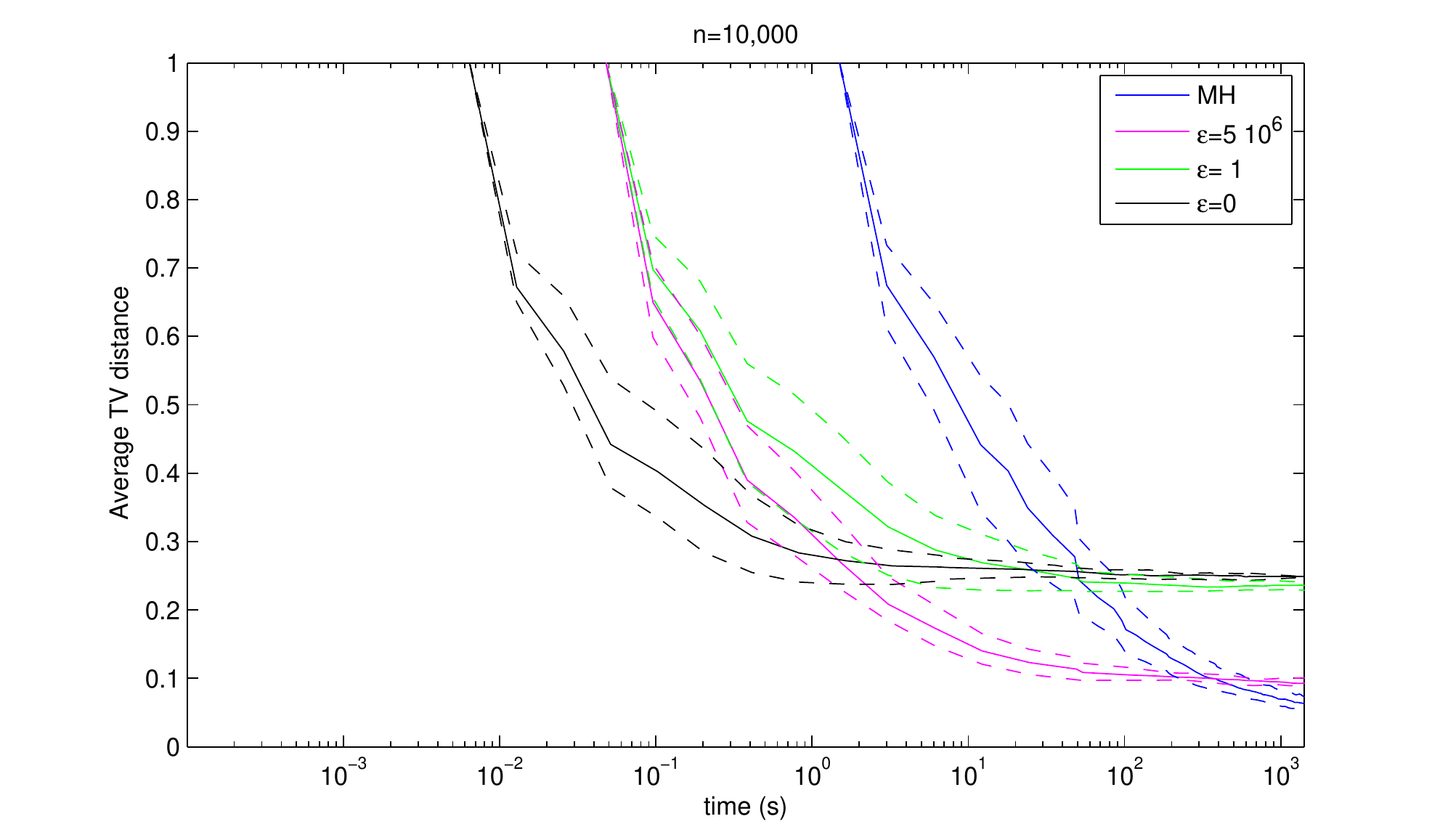}

\caption{(Example \ref{ex:ar}: Autoregressive time series) Average Total variation distance over the three marginals between $\pi$ and $\tilde{\pi}_t$ in function of the simulation time $t$. The dashed lines represent the first and third quartiles. Scenario $n=1,000$ (top), $n=5,000$ (middle) and $n=10,000$ (bottom) with three different $\epsilon$. Note that in all the three plots, the MH curves are identical and are just reported for comparison purpose.\label{fig:cpu3}}
\end{figure}

\subsection{Logistic Regression Example}
\begin{exa}
\label{ex:logreg}
A $d$-dimensional logistic regression model is parameterized by a vector $\theta=(\theta_1,\ldots,\theta_d)\in\Theta\subset\rset^d$. Observations are realizations of the following model:
\begin{itemize}
\item simulate covariates $X_i=(X_{i,1},\ldots,X_{i,d})\sim\norm(0,(1/d)^2)$
\item simulate $Y_i$ given $\theta$ and $X_i$ as
\begin{equation}
Y_i=
\left\{
\begin{array}{ll}
1&\text{w.p.}\;1\slash\left(1+e^{-\theta X\T}\right)\,,\\
0&\text{otherwise}\,.
\end{array}
\right.
\end{equation}
We have simulated $N=10^6$ observations $Y_1,Y_2,\ldots$ under the true parameter $\theta^{\ast}=(1, 2, -1)$ $(d=3)$.
\end{itemize}

\end{exa}

The summary statistics were set as the maximum likelihood estimator returned by the Matlab routine \texttt{glmfit} and were graphically validated, as in Figure \ref{fig:val}. The tolerance parameter was consequently set $\epsilon=5.0\,10^{6}$. We study the influence of $n$ on the Informed Sub-Sampling chain marginal distributions in Figure \ref{fig:logistic}. We note that as soon as $n\geq 5,000$, the bias vanishes and that when random subsampling is used (\ie $\eps=0$), the bias is much larger. Of course, Figure \ref{fig:logistic} only gives information about the marginal distributions. To complement the study, we consider estimating the probability $\pi(D)$ where $D$ is the domain defined as:
$$
D\subset\Theta=\left\{\theta_1\in(0.98, 1.00)\,;\,\theta_2\in (1.98, 2.01)\,;\,\theta_3\in(-0.98, -0.95)\right\}\,,
$$
in order to check that the joint distribution $\pi$ is reasonably inferred. Numerical integration using a long Metropolis-Hastings algorithm, gave the ground truth $\pi(D)=0.1$. The top panel of Figure \ref{fig:var} illustrates the Monte Carlo estimation of $\pi(D)$ based on $i=10,000$ iterations of ISS-MCMC implemented with $n\in\{1,000\,;\,5,000\,;\,10,000\}$ and compares it to Metropolis-Hastings. As expected ISS-MCMC has a negligible bias and the variance of the estimator decreases when $n$ increases. Indeed, when $n$ increases, the Informed Sub-Sampling process is less likely to pick irrelevant subsets, which in turns lower the variability of the chain. The Monte-Carlo estimation based on ISS-MCMC with $n=10,000$ and Metropolis-Hastings are very similar. However, when we normalize the experiment by the CPU time, Metropolis-Hastings is clearly outperformed by ISS-MCMC. The lower panel of Figure \ref{fig:var} assumes that only $t=500$ seconds of computation are available. All the chains are started from $\theta^\ast$. Table \ref{tab:var} reports the quantitative details of this experiment. In such a situation, one should clearly opt for the Informed Sub-Sampling approach as the Metropolis-Hastings algorithm only achieves 50 iterations for this amount of computation and as such fails to reach stationarity.

We also compare ISS-MCMC with two other algorithms that approximate the Metro\allowbreak polis-Hastings algorithm by using subset of data, drawn, unlike ISS-MCMC, uniformly at random. More precisely, we have implemented the Stochastic Gradient Langevin Dynamic (SGLD) from \cite{welling2011bayesian} and the Subsampled likelihoods MH (SubLhd1) algorithm from \cite{bardenet2014towards} along with an improved version of this algorithm that makes use of control variates, referred to as the Improved Confidence sampler in \cite{bardenet2015markov} but abbreviated here as SubLhd2 for simplicity. Those algorithms have been implemented in their default version, following the parameterization prescribed in their original article. All those methods are inexact and we are interested in comparing the bias/variance tradeoff per CPU time unit. Results in terms of convergence in distribution and Monte Carlo estimation are reported respectively in Figure \ref{fig:logistic} and Table \ref{tab:var}. For this model, SGLD and SubLhd1 show a larger bias than ISS-MCMC and SubLhd2: they need larger subset size $n$ to achieve a similar precision than ISS-MCMC or SubLhd2, see Figure \ref{fig:logistic}. SubLhd2 seems to outperform ISS-MCMC when $n$ is low in terms of distribution bias but the two methods perform equally good when $n\geq 5,000$. Quantitatively, the Monte Carlo estimation of $\pi(D)$ appears better with SubLhd2 than any other method for any subset size, as indicated by the RMSE reported at Table \ref{tab:var}. However, looking at the comparative boxplot representing the distribution of the Monte Carlo estimator of $\pi(D)$ in time normalized experiments (Figure \ref{fig:var_app}), one can see that when $n$ is larger than $5,000$, estimators from ISS-MCMC and SubLhd2 are quite similar confirming the qualitative impression of Figure \ref{fig:logistic} and perhaps moderating the RMSE-based assessment made at Table \ref{tab:var}.

\begin{figure}
\centering
\includegraphics[scale=.6]{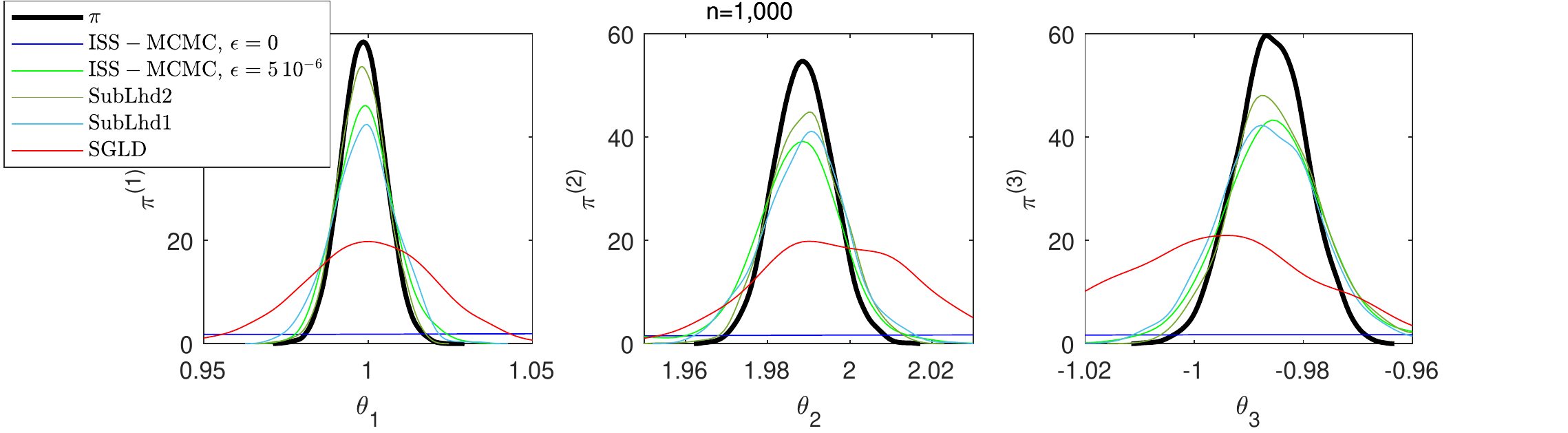}

\includegraphics[scale=.6]{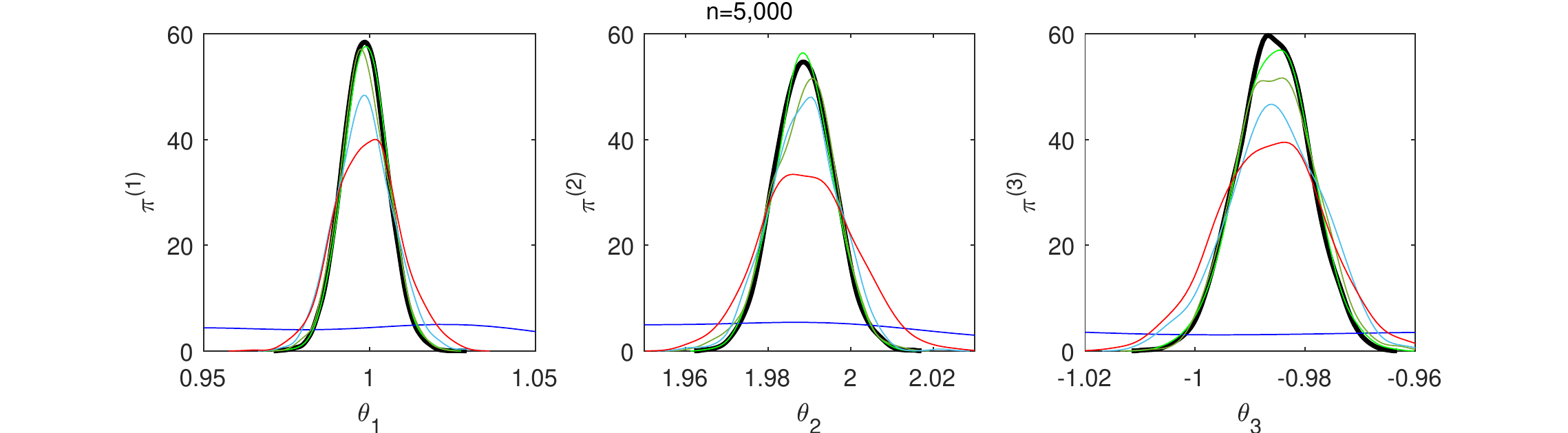}

\includegraphics[scale=.6]{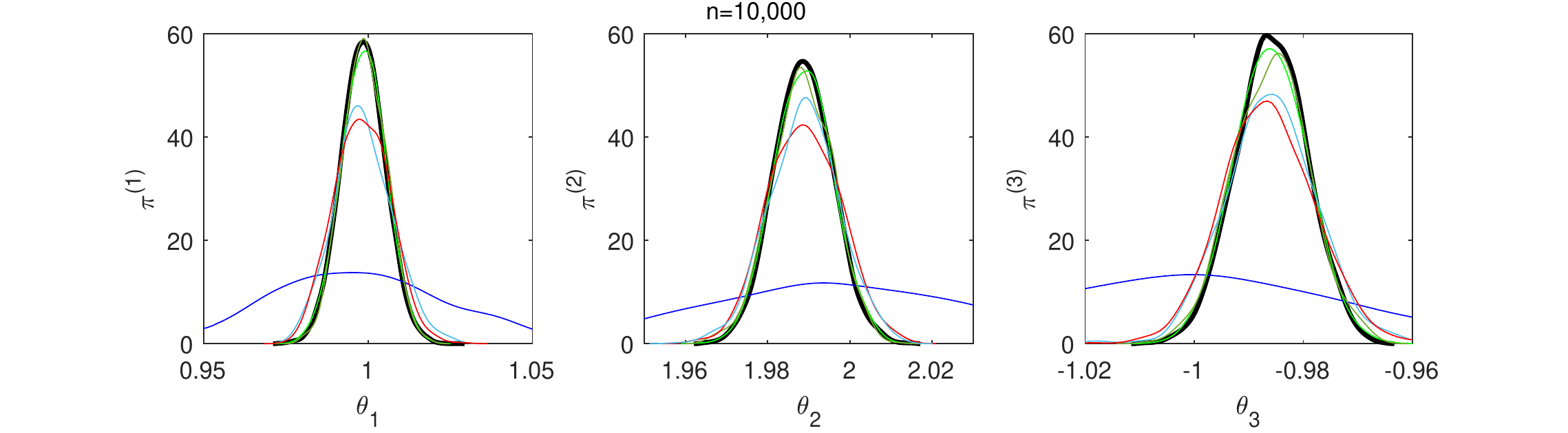}
\caption{(Example \ref{ex:logreg}: Logistic regression) Stationary marginal distributions of several algorithms approximating Metropolis-Hastings using subsamples: ISS-MCMC with $\epsilon=\{0\,;\,5.0\,10^4\}$, the MH Sublikelihood algorithm \cite{bardenet2014towards} (SubLhd1) (and its improved version denoted SubLhd2, see \cite{bardenet2015markov}) and the Stochastic Gradient Langevin Dynamic (SGLD) \cite{welling2011bayesian}. The plots represent the marginal distribution $\tpi_i$, (after $i=1,000$ iterations) and different subset sizes $n\in\{1,000\,;\,5,000\,;\, 10,000\}$. The true marginal $\pi$ is in black. $\tpi_i$ was estimated by simulating 1,000 iid copies of the Markov chain generated by the five algorithms. \label{fig:logistic}}
\end{figure}

\begin{figure}
\centering

\includegraphics[scale=.6]{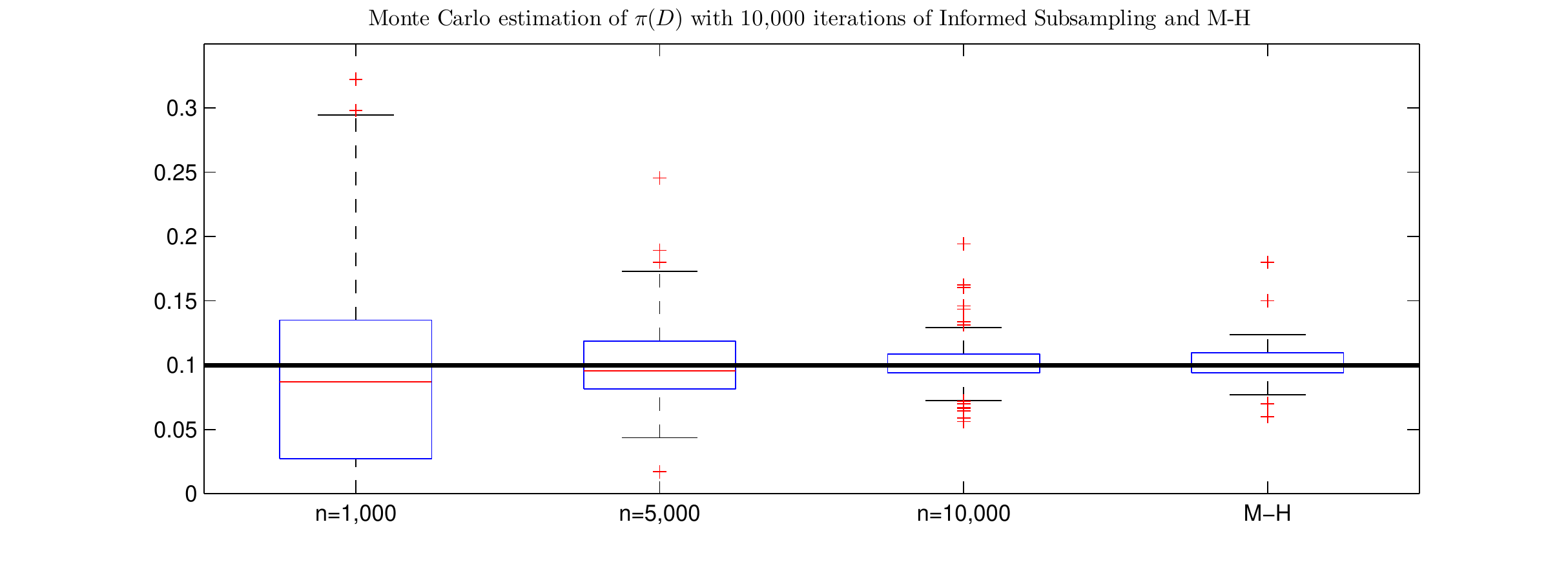}

\includegraphics[scale=.6]{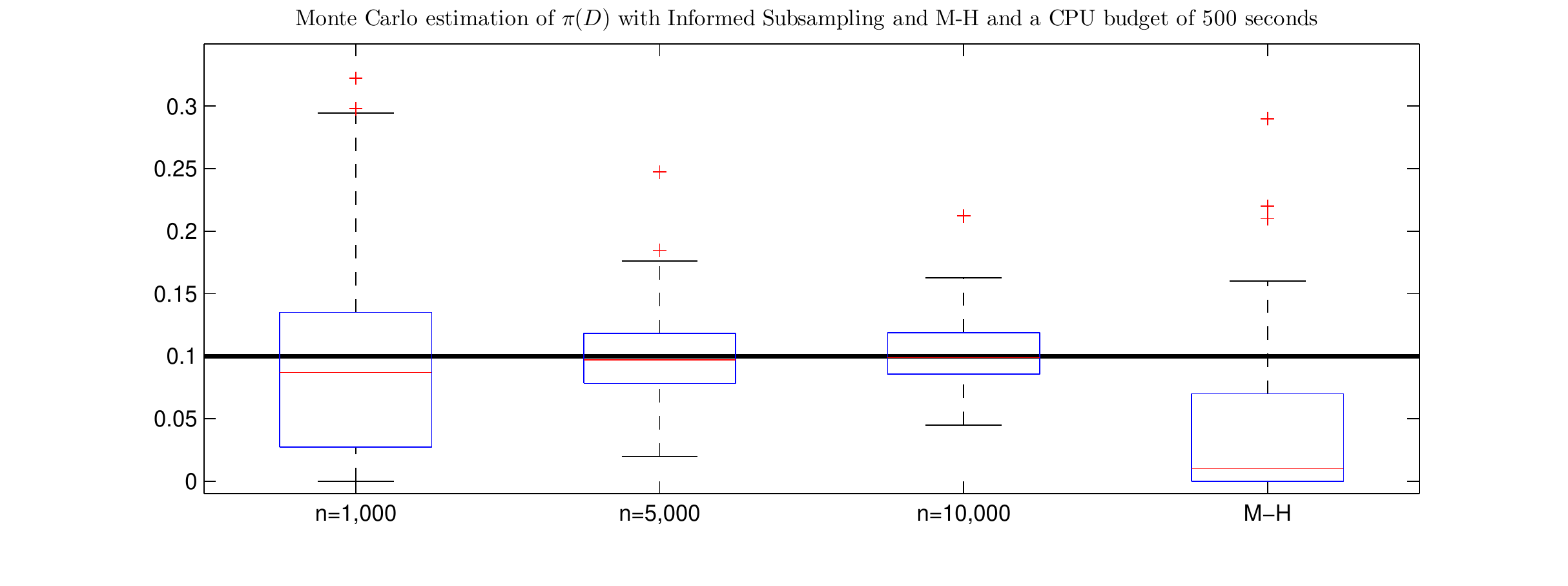}
\caption{(Example \ref{ex:logreg}: Logistic regression) Estimation of $\pi(D)$ based on ISS-MCMC implemented with $n\in\{1,000\,;\,5,000\,;\, 10,000\}$ and Metropolis-Hastings. Top: the experiment is iteration-normalized, \ie the chains run for $10,000$ iterations. Bottom: the experiment is time-normalized, \ie the chains run for $500$ seconds. Each chain was replicated 100 times and started from $\theta^\ast$. \label{fig:var}}
\end{figure}

\begin{table}[h]
\centering
\begin{tabular}{ccccc}
algorithm & time/iter.(s) & iter. completed & RMSE & $\var\{\widehat{\pi(D)}\}$\\
\hline
\hline
M-H & $10$ & $50$ & $0.1417$ & $0.004$\\
\hline
ISS-MCMC, $n=1,000$ & $0.05$ & $10,000$ & $0.1016$ & $0.0104$\\
ISS-MCMC, $n=5,000$ & $0.08$ & $6,250$ & $0.0351$ & $0.0012$\\
ISS-MCMC, $n=10,000$ & $0.13$ & $3,840$ & $0.0267$ & $0.0007$\\
\hline
SGLD, $n=1,000$ & $0.08$ & $6,000$ & $0.1370$ & $0.0157$\\
SGLD, $n=5,000$ & $0.11$ & $5,250$ & $0.0996$ & $0.0100$\\
SGLD, $n=10,000$ & $0.12$ & $4,500$ & $0.0326$ & $0.0011$\\
\hline
SubLhd1, $n=1,000$ & $1.45$ & $350$ & $0.0762$ & $0.0042$\\
SubLhd1, $n=5,000$ & $1.56$ & $323$ & $0.0680$ & $0.0046$\\
SubLhd1, $n=10,000$ & $2.24$ & $223$ & $0.0656$ & $0.0044$\\
\hline
SubLhd2, $n=1,000$ & $0.10$ & $5,046$ & $0.0304$ & $0.0002$\\
SubLhd2, $n=5,000$ & $0.14$ & $3,581$ & $0.0260$ & $0.0006$\\
SubLhd2, $n=10,000$ & $0.19$ & $2,631$ & $0.0195$ & $0.0002$
\end{tabular}
\caption{(Example \ref{ex:logreg}: Logistic regression) Tradeoff Bias-Variance of the Monte Carlo estimator from Metropolis-Hastings, ISS-MCMC, Stochastic Gradient Langevin Dynamics (SGLD) \cite{welling2011bayesian}, the Subsampled likelihoods (SubLhd1) \cite{bardenet2014towards} and the improved Confidence Sampler (SubLhd2) \cite{bardenet2015markov} for a fixed computational budget of 500 seconds. Those results were replicated using 100 replications of each algorithm. Note that for SubLhd1 and SubLhd2, $n$ corresponds to the initial subset size and not to the actual subsample size that was actually used in each iteration, a parameter which is chosen by the algorithms. \label{tab:var}}
\end{table}

\begin{figure}
\centering
\includegraphics[scale=0.6]{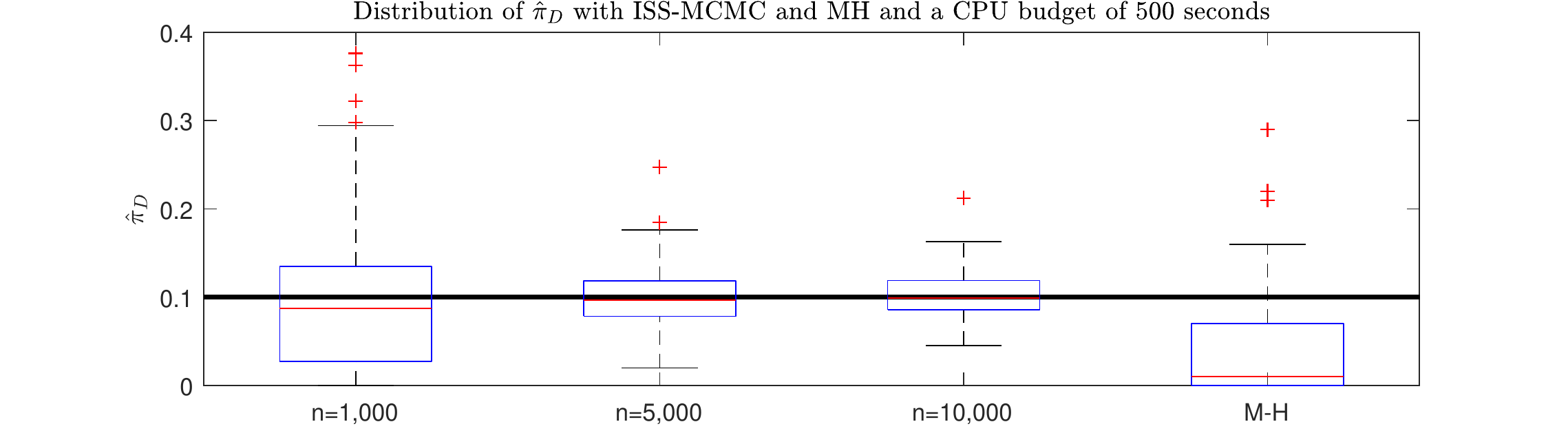}

\includegraphics[scale=0.6]{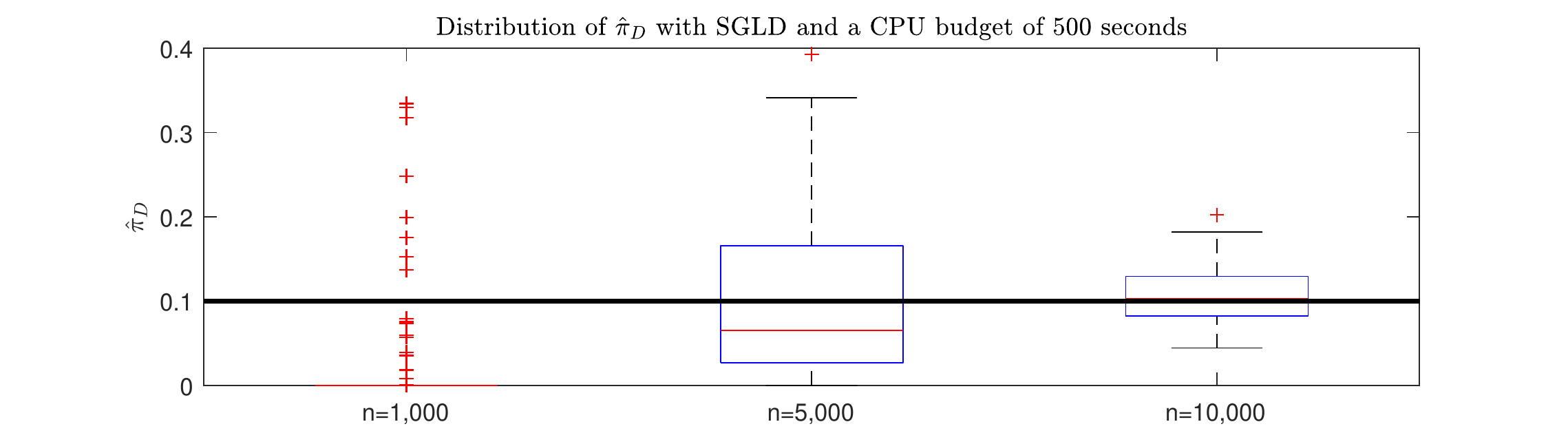}

\includegraphics[scale=0.6]{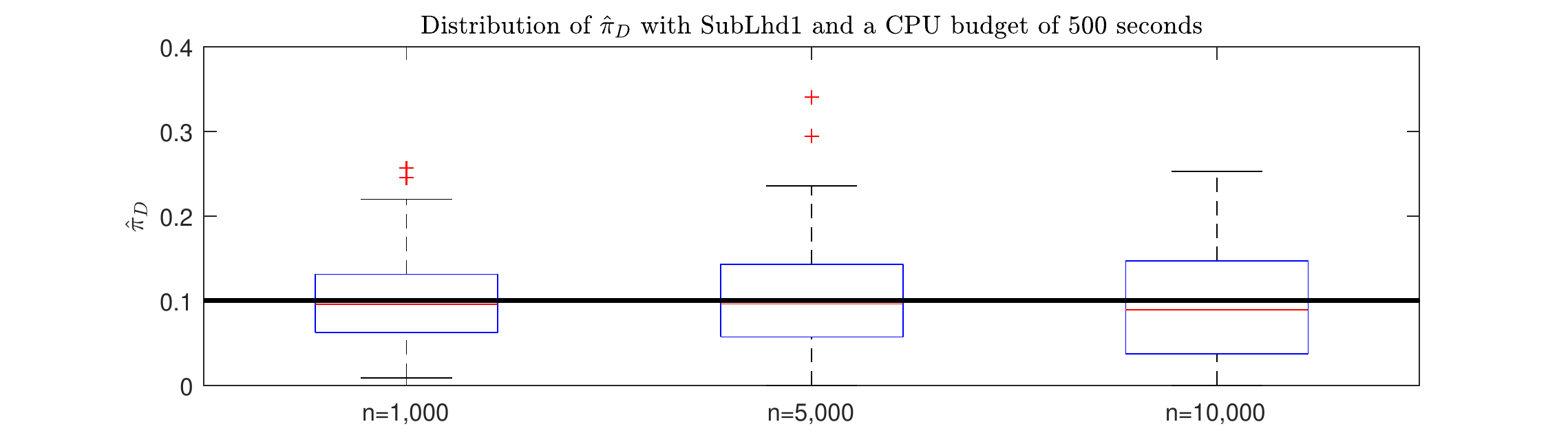}

\includegraphics[scale=0.6]{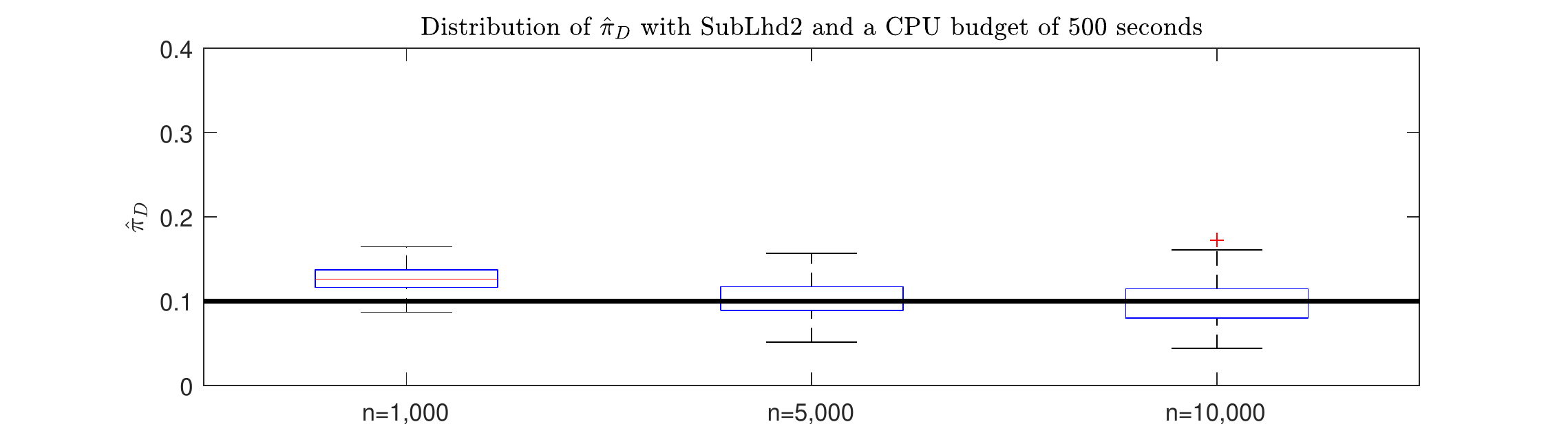}

\caption{(Example \ref{ex:logreg}: Logistic regression) Estimation of $\pi(D)$ based on ISS-MCMC, MH, SGLD, SubLhd1 and SubLhd2 implemented with $n\in\{1,000\,;\,5,000\,;\, 10,000\}$. The Monte Carlo estimation of $\pi(D)$ was carried out using those algorithms for $500$ seconds. Each estimation was replicated 100 times and started from $\theta^\ast$ for each algorithm. Note that for SubLhd1 and SubLhd2, $n$ corresponds to the initial subset size and not to the actual subsample size that was actually used in each iteration, a parameter which is adaptively tuned by the algorithms.  \label{fig:var_app}}
\end{figure}

\subsection{Binary Classification}
\begin{exa}
\label{ex:bincla}
A training dataset consisting of $N=10^7$ labeled observations $Y=\{Y_k,\,k\leq N\}$ from a 2 dimensional Gaussian mixture model is simulated with
$$
Y_k\,|\,I_k=i\sim \norm(\mu_i,\Gamma_i)\,,\qquad I_k\sim\ber(1/2)\,,
$$
where $\mu_1=[\theta_1\,,\,0]$, $\mu_2=[\theta_2\,,\,0]$, $\Gamma_1=\diag([\theta_3/2\,,\,\theta_3])$ and $\Gamma_2=\diag([\theta_4/2\,,\,\theta_4])$. We define $\theta=(\theta_1,\theta_2)$ with $\theta_i\in\rset\times \rset_*^{+}$ for each model $i\in\{1,2\}$. A prior distribution $(\theta_1,\theta_2)\sim_{\iid} p:=\norm(0,1/2)\otimes\Gamma(1,2)$ ($\Gamma(a,b)$ is the Gamma distribution with shape $a$ and rate $b$) is assigned to $\theta$. Consider an algorithm $\A$ that simulates a Markov chain
$$
\theta_\A:=\left\{\left(\theta_{1,k}^{(\A)},\theta_{2,k}^{(\A)}\right),\,k\in\nset\right\}
$$
targeting the posterior distribution of $\theta$ given $Y$, perhaps approximately. We consider the real-time supervised classifier $I_{\A}^\ast(t)$, driven by $\theta_\A$, for the test dataset $Y^{\ast}=\{Y_k^{\ast},\,k\leq \Ntest\}$ ($\Ntest=10^4$) and defined as:
\begin{equation}
\label{eq:class}
I_{\A}^\ast(t)=(I_{\A,1}^\ast(t),\ldots,I_{\A,\Ntest}^\ast(t))\,,\qquad I_{\A,k}^\ast(t)=\arg\max_{i\in\{1,2\}}f(Y_k^\ast\,|\,\bar{\theta}_{i,\kappa_\A(t)}^{(\A)})\,,
\end{equation}
where $\kappa_\A(t)=\sup_{k\in\nset}\{\tau_k^\A\leq t\}$ and $\bar{\theta}_{i,k}^{(\A)}=(1/k)\sum_{\ell=1}^k\theta_{i,\ell}^{(\A)}$. We have defined $\tau_k^\A$ as the wall clock time to generate $k$ iterations of algorithm $\A$. We define the live classification error rate as $\eps_\A(t)=\|I_\A^\ast(t)-I^\ast\|_1$ where $I_\A^{\ast}=(I_{\A,1}^{\ast},\ldots,I_{\A,N}^{\ast})$ and $I_k^{\ast}$ is the true class of $Y_k^{\ast}$. We compare $\epsilon_\A$ for three different algorithms $\A$: ISS-MCMC, Metropolis-Hastings and Subsampled Likelihoods \citep{bardenet2014towards}.
\end{exa}

In this simulation example, we have used the true value $\theta^\ast=(-1, 1/2 , 1 , 1/2)$ and simulated $Y$ such that it contains the same number of observations from model 1 and model 2 \ie $N/2$. The three algorithms were implemented with the same proposal kernel, namely a single site random walk with adaptive variance that guarantees an acceptance rate between 0.40 and 0.50, see \cite{roberts2001optimal,haario2001adaptive}. ISS-MCMC was implemented with parameters $n=1,000$ and $\eps=10^{7}$. The summary statistics were taken as $S(Y_U)=0$ if $\sum_{k\in U}\1_{\{I_k=1\}}\neq \sum_{k\in U}\1_{\{I_k=2\}}$ and
\begin{multline}
\label{eq:ex4_ss}
\bS(Y_U)=\bigg[(2/n)\sum_{k=1}^{n/2}Y_k\1_{\{I_k=1\}},\;\tra(\cov(Y_k,\,k\in U,\, I_k=1))\,,\\
(2/n)\sum_{k=1}^{n/2}Y_k\1_{\{I_k=2\}},\;\tra(\cov(Y_k,\,k\in U,\, I_k=2))\bigg]
\end{multline}
otherwise. This choice allows to keep the right proportion of data from the two models in any subsample used for the inference. The statistics in \eqref{eq:ex4_ss} are sufficient for each model, taken separately. Subsampled Likelihoods was implemented with the default parameters prescribed in the introduction of Section 4 in \cite{bardenet2014towards}.

Figure \ref{fig:binary} compares the live classification error rate achieved by the three algorithms. We also report the optimal Bayes classifier which achieves $\epsilon_B(t)=0.0812$ classifying $Y_k^{\ast}$ in class 1 if $Y_{k,1}^{\ast}<0$ and in class 2 alternatively.
Unsurprisingly, Metropolis-Hastings is penalized because it evaluates the norm of a $N=10^7$ dimensional vector at each iteration. Subsampled Likelihoods does slightly better than M-H but suffers from the fact that close to stationary regime, the algorithm ends up drawing the quasi-entire dataset with high probability, a fact which was explained in \cite{bardenet2014towards}.

\begin{figure}
\centering

\includegraphics[scale=0.55]{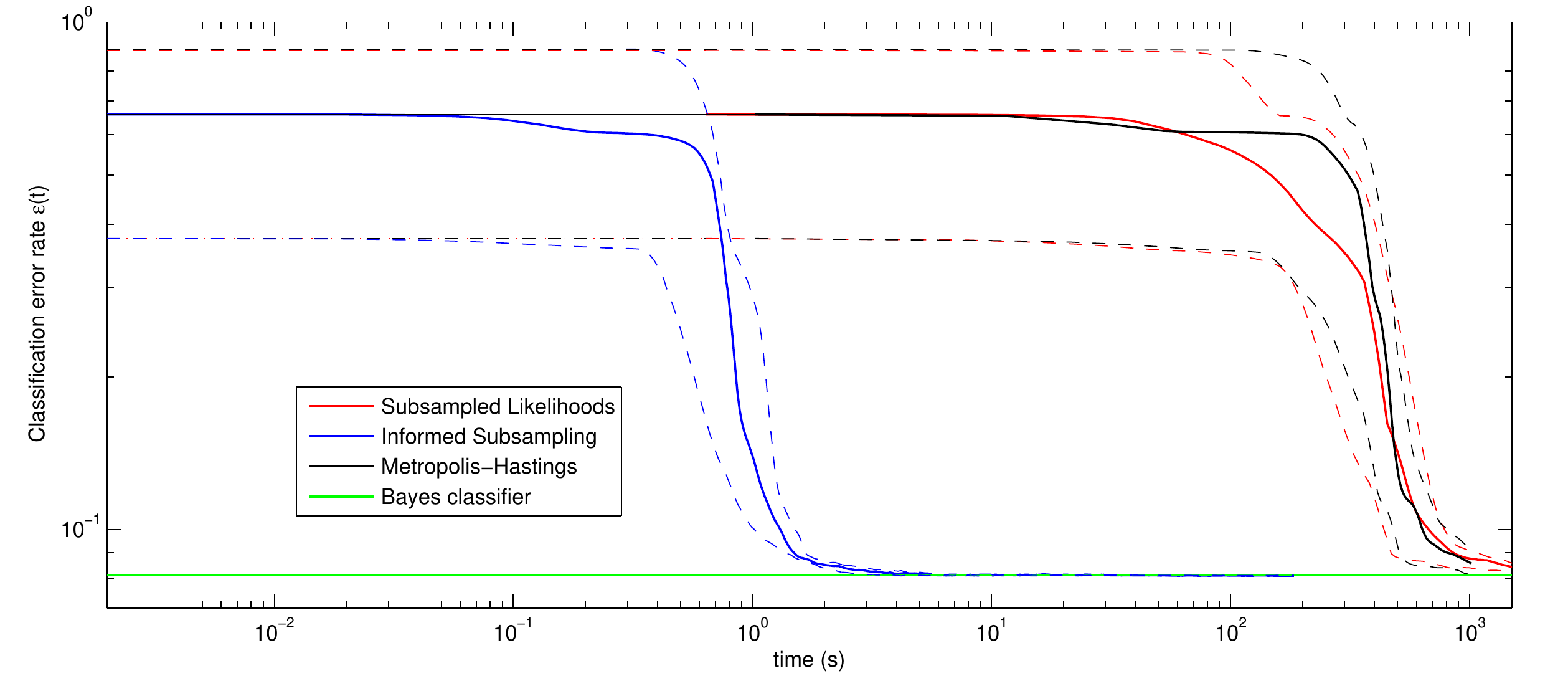}
\caption{(Example \ref{ex:bincla}: Binary classification) Live classification error rate for three algorithms. This plot was generated by classifying the same test dataset $Y^{\ast}$ using the same training dataset $Y$ for the three algorithms. The variability arises from the initial state of the Markov chains. We have used 30 different initial states for the three algorithms and report the median (plain line) and the two quartiles (dashed lines).\label{fig:binary}}
\end{figure}

\subsection{Additional details for the handwritten digit inference (Example \ref{ex1})}

In the handwritten digit example (Example \ref{ex1}), we have used batches of $n=100$ data. Since the initial dataset comprises $2,000$ observations per digit, the summary statistics were defined in a way that any subsample contains $20$ observations from each class. More precisely, we have set for any subsample $Y_U$, $S(Y_U)=0$ if for at least one class $i\in\{1,\ldots,5\}$, $(1/n)\sum_{k\in U}\1_{\{J(k)=i\}}\neq 20$ and
$$
\bS(Y_U)=\left\{\sum_{k\in U} \phi(\theta_{J(k)})\1_{\{J(k)=i\}}\bigg\slash \sum_{k\in U}\1_{\{J(k)=i\}}\right\}_{i=1}^5
$$
otherwise. We set the bandwidth to $\eps=10^{5}$. The proposal kernels of the Informed Sub-Sampling chain and M-H were defined as the same Random Walk kernel. In particular, at each iteration only a bloc of the template parameter of one of the five classes is updated. The variance parameter of the Random Walk is adapted according to the past trajectory of the chain, so as to maintain an acceptance rate of $.25$.

Figure \ref{fig:template2} reports the empirical marginal distribution of one component for each vector $\theta_1,\ldots,\theta_5$ obtained from ISS-MCMC and from M-H. Those distributions are estimated from $50,000$ iterations of both algorithms, in stationary regime. This shows that the distribution of those parameters are in line with each other.

\begin{figure}
\centering
\begin{tabular}{c}
\includegraphics[scale=.7]{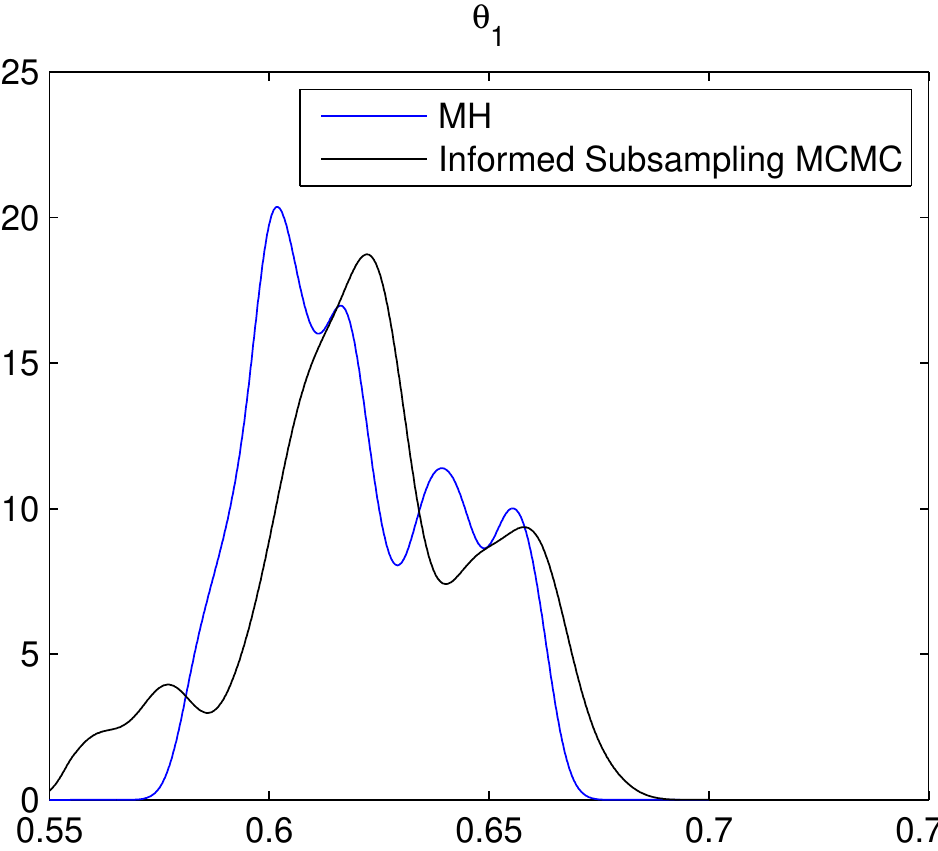}
\includegraphics[scale=.7]{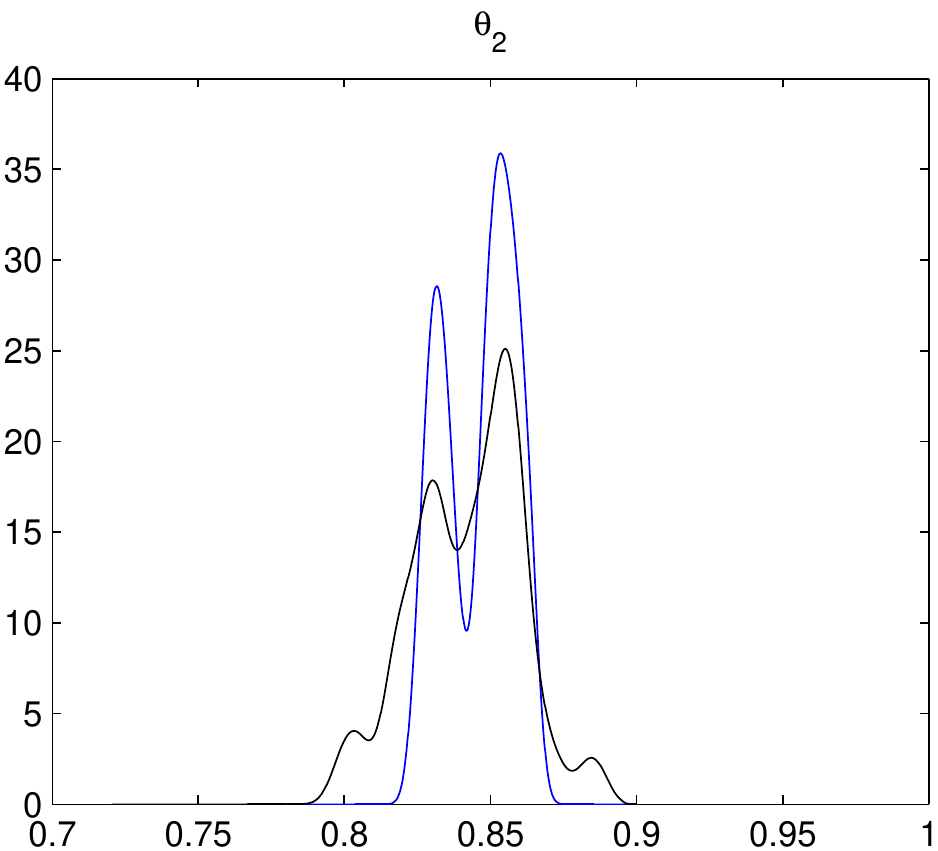}
\end{tabular}

\begin{tabular}{c}
\includegraphics[scale=.7]{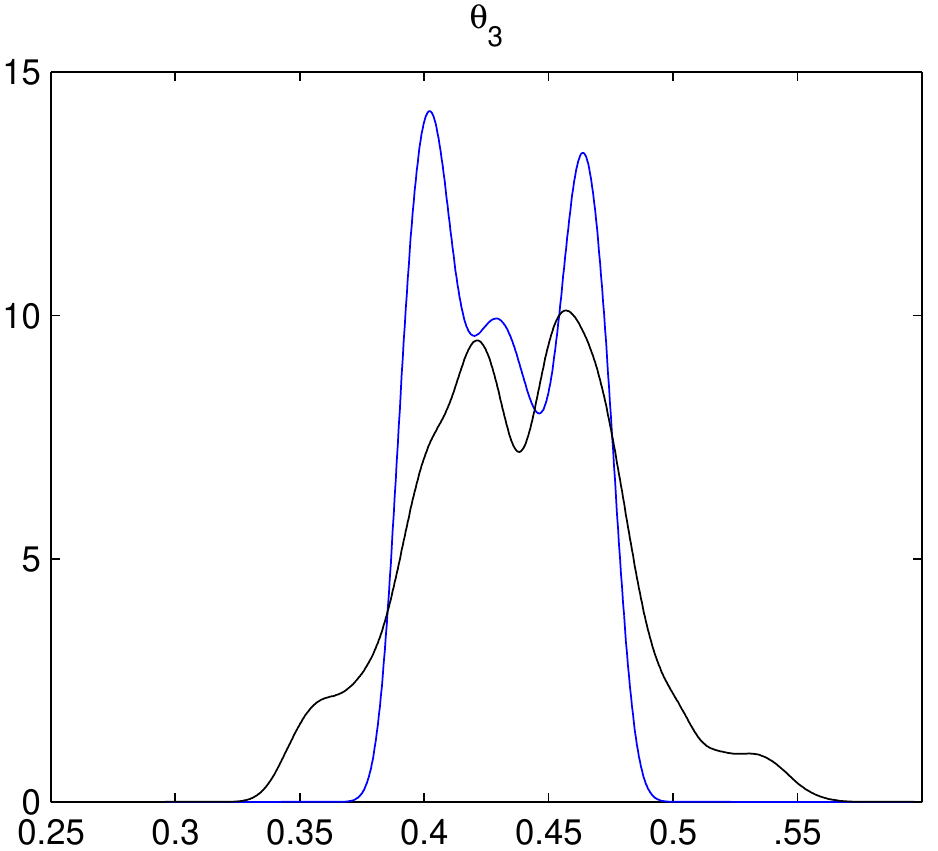}
\includegraphics[scale=.7]{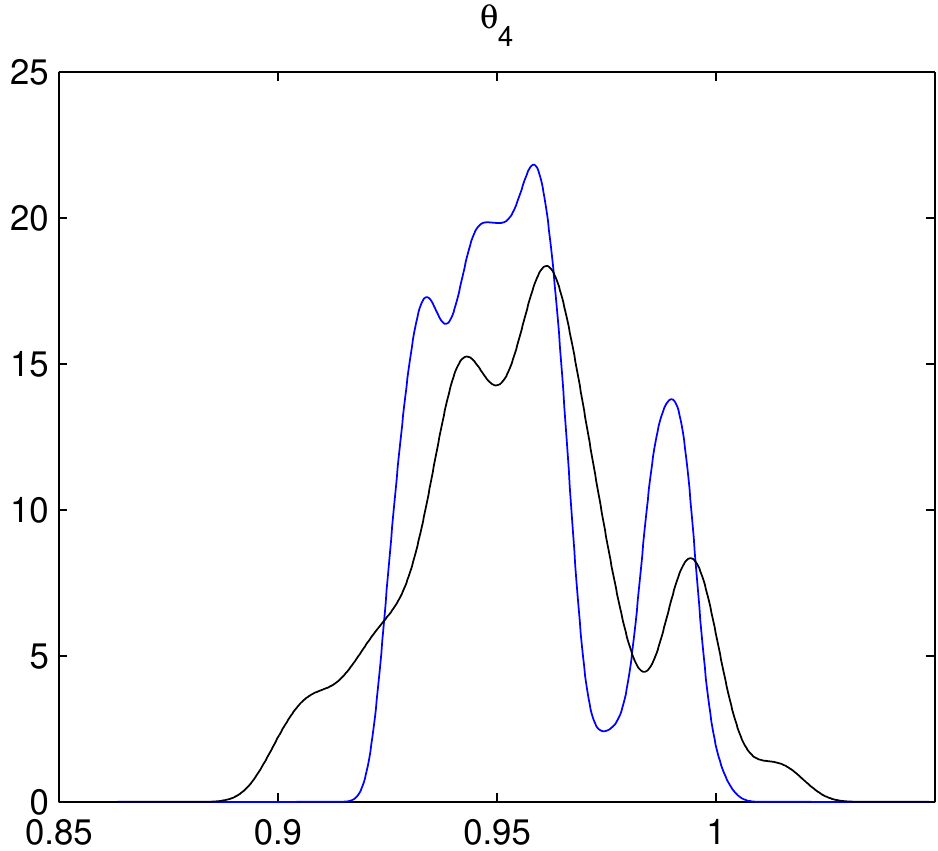}
\end{tabular}

\begin{tabular}{c}
\includegraphics[scale=.7]{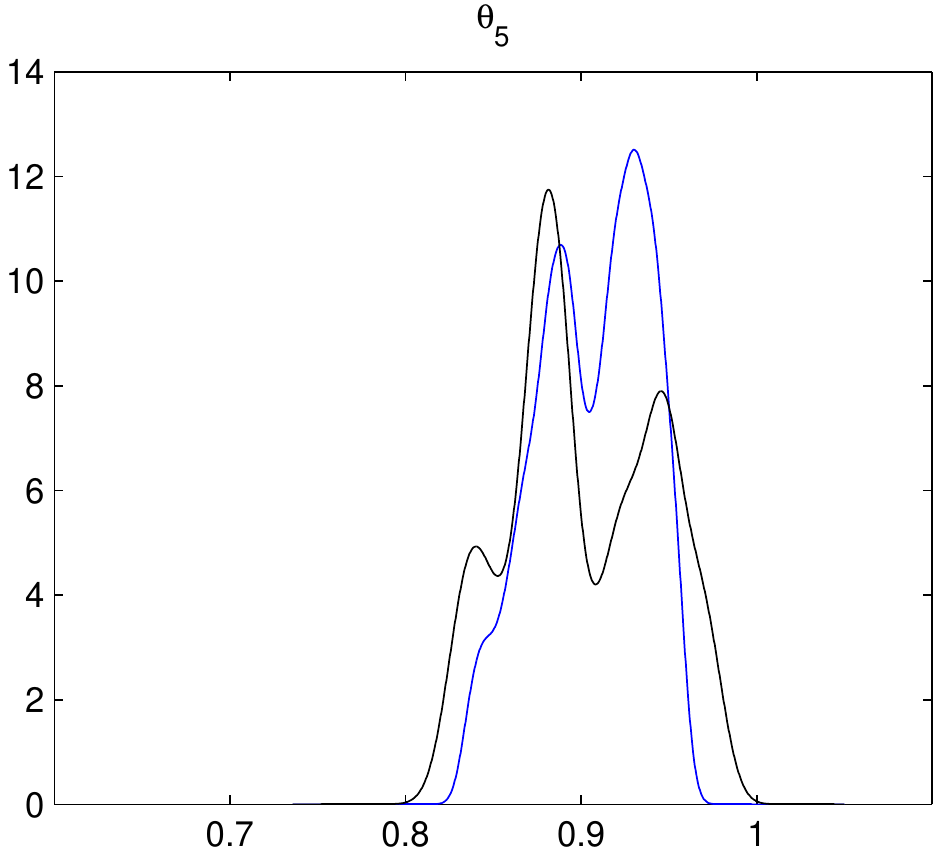}
\end{tabular}

\caption{(Example \ref{ex1}: Handwritten digits) Empirical marginal distribution of one component of the vectors $\theta_1,\ldots,\theta_5$ using the Metropolis-Hastings chain (blue) and the Informed Sub-Sampling chain (black), estimated from $10,000$ transitions at stationary regime.
\label{fig:template2}}

\end{figure} 

\section{Conclusion}
When the available computational budget is limited, inferring a statistical model based on a tall dataset in the Bayesian paradigm using the Metropolis-Hastings (M-H) algorithm is not computationally efficient. Several variants of the M-H algorithm have been proposed to address this computational issue \citep{bardenet2014towards,banterle2015accelerating,korattikara2013austerity,maclaurin2014firefly}. However, (i) they often lose the original simplicity of M-H, (ii) they are only applicable in situations where the data are independent and (iii) the computational cost of one iteration is stochastic which can potentially compromise any computational saving. Informed Sub-Sampling MCMC pushes the approximation one step forward: the computational cost of one iteration is deterministic and is controlled through a user specified parameter, the size of the subsamples.

The aforementioned methods rely on subsampling the whole dataset uniformly at random, at each iteration. Using such subsamples in lieu of the whole dataset in the original M-H chain leads to an algorithm whose 
approximation of $\pi$ comes with no guarantee when the subsamples size is fixed. Our main contribution is to show that assigning a distribution to the subsamples that reflects their fidelity to the whole dataset 
allows one to control the L1 approximation error, even when the subsample size is fixed, which is of computational interest.

The algorithm we propose to achieve this task, Informed Sub-Sampling MCMC, offers an alternative to situations where other scalable Metropolis-Hastings variants cannot be implemented (because the model does not satisfy the assumptions \eg independence of the data, existence of a concentration inequality for the model or a cheap lower bound of the likelihood) or are inefficient (because the method ends up using nearly the whole dataset at each iteration). However, in scaling up Metropolis-Hastings there is no free lunch neither. In particular, our method replaces the uniform subsampling approach by a more sophisticated subsampling mechanism involving summary statistics. In this regard, even though our method is in principle widely applicable, it will only be useful in situations where a cheap summary statistics function satisfying Assumption \ref{assu:4} is available. In particular, we have shown that our method will give meaningful results when the maximum likelihood estimator is cheap to compute, which somewhat correlates with the optimality of summary statistics in ABC established in \cite{fearnhead2012constructing}. We note that it is possible to construct a counterpart to ISS-MCMC in the Sequential Monte Carlo framework, as a sequence of scaled subposteriors naturally arises in our algorithm. We leave the design and evaluation of such a method for future research.

Similarly to the other noisy subsampling methods that approximate MH (\eg \cite{welling2011bayesian}, \cite{korattikara2013austerity}, \cite{bardenet2015markov}, etc.), our theoretical results address the convergence in distribution of the ISS-MCMC chain to the posterior distribution. Questions of interest to practitioners encompass establishing Law of Large Numbers and Central limit theorems for those types of algorithms. Addressing those questions would allow a better understanding of approximate MCMC applied to large dataset contexts, from a practical perspective. In the noisy Monte Carlo literature, the closest result is perhaps Theorem 2.5 of \cite{johndrow2015approximations} which comes in the form of a non-asymptotic error bound in the L2 norm between $\pi f$ and its MCMC estimate. The assumptions are however quite strong which prevent applying it to ISS-MCMC and we will study those questions in a future work.  Finally, recent developments in the understanding of approximate Markov chains have been carried out in \cite{rudolf2015perturbation}. These authors provide explicit convergence bounds in the Wasserstein metric under the Wasserstein ergodicity assumption, a notion which is closely related to geometric ergodicity in V-norm, hence milder than uniform ergodicity. We regard this contribution as a very promising research avenue that unveils new ways of deriving quantitative error bounds for the practical problem of approximate Metropolis-Hastings algorithms that make use of subsets of data, under more realistic assumptions. 

\appendix
\section{Proofs}
\subsection{Proof of Proposition \ref{prop:KL}}
\label{app1}
\begin{proof}
For notational simplicity and without loss of generality, we take here $g$ as the identity on $\Theta$. Let $n<N$ and $U$ be a subset of $\{1,\ldots,N\}$ with cardinal $n$. Consider the power likelihood:
\begin{equation*}
\tf_n(Y_U\,|\,\theta)=f(Y_U\,|\,\theta)^{N/n}=\left\{\prod_{k\in U}f(Y_k\,|\,\theta)\right\}^{N/n}=\frac{\exp\left\{(N/n)\sum_{k\in U}S(Y_k)\right\}\T \theta}{L(\theta)^N}\,,
\end{equation*}
and the corresponding power posterior:
\begin{equation*}
\ttarg_n(\theta\,|\,Y_U)=\frac{\exp\left\{(N/n)\sum_{k\in U}S(Y_k)\right\}\T \theta}{L(\theta)^N}p(\theta)\bigg\slash
\tZ_n(Y_U)\,,
\end{equation*}
where
$$
\tZ_n(Y_U)=\int p(\rmd\theta)\frac{\exp\left\{(N/n)\sum_{k\in U}S(Y_k)\right\}\T \theta}{L(\theta)^N}\,.
$$
For any $\theta$ such that $p(\theta)\neq 0$, write:
\begin{equation}
\log\frac{\targ(\theta\,|\,Y_{1:N})}{\ttarg_n(\theta\,|\,Y_{U})}=\left\{\sum_{k=1}^N S(Y_k)-(N/n)\sum_{k\in U}S(Y_k)\right\}\T\theta+\log \frac{\tZ_n(Y_U)}{Z(Y_{1:N})}\,.
\end{equation}
and the KL divergence between $\pi(\,\cdot\,|\,Y_{1:N})$ and $\ttarg(\,\cdot\,|\,Y_U)$, denoted $\text{KL}_n(U)$, simply writes
\begin{equation}
\label{eq3a}
\text{KL}_n(U)=\Delta_n(U)\T\esp_\pi(\theta)+\log \frac{\tZ_n(Y_U)}{Z(Y_{1:N})}\,,
\end{equation}
where $\Delta_n(U)=\sum_{k=1}^N S(Y_k)-(N/n)\sum_{k\in U}S(Y_k)$. Now, note that
\begin{multline}
\label{eq4a}
\tZ_n(Y_U)=\int p(\rmd\theta)\frac{\exp\left\{(N/n)\sum_{k\in U}S(Y_k)\right\}\T \theta}{L(\theta)^N}
=\int p(\rmd\theta)\frac{\exp\left\{\sum_{k=1}^N S(Y_k)-\Delta_n(U)\right\}\T \theta}{L(\theta)^N}\\
=\int p(\rmd\theta)f(Y_{1:N}\,|\,\theta){\exp\left\{-\Delta_n(U)\T \theta\right\}}=Z(Y_{1:N})\esp_\pi\left\{\exp\left(-\Delta_n(U)\T \theta\right)\right\}\,.
\end{multline}
Plugging \eqref{eq4a} into \eqref{eq3a} yields:
\begin{multline}
\text{KL}_n(U)=\Delta_n(U)\T\esp_\pi(\theta)+\log \esp_\pi\left\{\exp\left(-\Delta_n(U)\T \theta\right)\right\}\,,
\\
=\log \frac{\esp_\pi\left\{\exp\left(-\Delta_n(U)\T \theta\right)\right\}}{\exp(-\Delta_n(U)\T \esp_\pi(\theta))}
=\log\esp_\pi\exp\left[\left\{\esp_\pi(\theta)-\theta\right\}\T \Delta_n(U)\right]\,.
\end{multline}
Finally, Cauchy-Schwartz inequality provides the following upper bound for $\text{KL}_n(U)$:
\begin{equation}
\label{eq5}
\text{KL}_n(U)\leq \log\esp_\pi\exp\left\{\left\|\esp_\pi(\theta)-\theta\right\|\|\Delta_n(U)\|\right\}\,.
\end{equation}
\end{proof}

\subsection{Proof of Proposition \ref{prop:BvM}}
\label{app2}
\begin{proof}
Under some weak assumptions, Bernstein-von Mises theorem states that $\pi(\,\cdot\,|\,Y_{1:N})$ is asymptotically (in $N$) a Gaussian distribution with the maximum likelihood $\thetaMLE$ as mean and $\Gamma_N=I^{-1}(\thetaMLE)/N$ as covariance matrix, where $I(\theta)$ is the Fisher information matrix at $\theta$. Let us denote by $\Phi$ the pdf of $\norm(\thetaMLE,\Gamma_N)$. Under this approximation, $\esp_\pi(\theta)=\thetaMLE$ and from \eqref{eq4a}, we write:
\begin{multline}
\label{eq6a}
\exp\text{KL}_n(U)\approx \int \Phi(\rmd\param)\exp\left[\{\thetaMLE-\theta\}\T\Delta_n(U)\right]=\int \Phi(\thetaMLE-\theta)\exp\{\theta\T \Delta_n(U)\}\rmd\theta\\
=\int \frac{1}{{(2\pi)}^{(d/2)}|\Gamma_N|^{(1/2)}}\exp\left\{-(1/2)\theta\T \Gamma_N^{-1}\theta +\theta\T \Delta_n(U)\right\}\rmd\theta\,,\\
=\frac{1}{{(2\pi)}^{(d/2)}|\Gamma_N|^{(1/2)}}\int \exp\left[-(1/2)\left\{\theta\T \Gamma_N^{-1}\theta -2\theta\T\Gamma_N^{-1}\Gamma_N \Delta_n(U)\right\}\right]\rmd\theta\,,\\
=\exp\{(1/2)\Delta_n(U)\T\Gamma_N \Delta_n(U)\}\,,
\end{multline}
by integration of a multivariate Gaussian density function. Eventually, \eqref{eq6a} yields the following approximation:
\begin{equation}
\label{eq7}
\text{KL}_n(U)\approx\wKL_n(U)=(1/2)\Delta_n(U)\T\Gamma_N \Delta_n(U)\,.
\end{equation}
\end{proof}

\subsection{Proof of Proposition \ref{prop:esp}}
\label{sec:proof:prop3}
\begin{proof}
Let $\Uset_n\supset A_n(\theta):=\left\{U\in\Uset_n,\;g(\theta)\T\Delta_n(U)\leq 0\right\}$ and remark that using Cauchy-Schwartz inequality, we have:
$$
\esp\left\{\frac{f(Y\,|\,\theta)}{f(Y_U\,|\,\theta)^{N/n}}\right\}\leq \nu_{n,\eps}\left\{A_n(\theta)\right\}+\sum_{U\in \Uset_n\backslash A_n(\theta)}\nu_{n,\eps}(U)\exp\{
\|g(\theta)\|\| \Delta_n(U)\|\}\,.
$$

Now, define $\bDelta_n(U):=\bS(Y)-\bS(Y_U)$ where $\bS$ is the normalized summary statistics vector, \ie if $U\in\Uset_n$, $\bS(Y_U)=S(Y_U)/n$. Clearly, when $N\to\infty$, some terms
$$
\exp\{\|g(\theta)\|\| \Delta_n(U)\|\}=\exp\{N\|g(\theta)\|\|\bDelta_n(U)\|\}
$$
will have a large contribution to the sum. More precisely, any mismatch between summary statistics of some subsamples $\{Y_U,\,U\in\Uset_n\backslash A_n(\theta)\}$ with respect to the full dataset will be amplified by the factor $N$, whereby exponentially inflating the upper bound. However, assigning the distribution $\nu_{n,\eps}$ \eqref{eq:nu} to the subsamples $\{Y_U,\,U\in\Uset_n\}$, allows to balance out this effect. Indeed, note that
$$
\esp\left\{\frac{f(Y\,|\,\theta)}{f(Y_U\,|\,\theta)^{N/n}}\right\}
\leq
\nu_{n,\eps}\{A_n(\theta)\}+\sum_{U\in \Uset_n\backslash A_n(\theta)}\exp\{-\eps\|\Delta_n(U)\|^2+\|g(\theta)\|\|\Delta_n(U)\|\}\slash Z(\eps)\,,
$$
where $Z(\eps)=\sum_{U\in\Uset_n}\exp\{-\eps\|\Delta_n(U)\|^2\}$ and we have, for a fixed $n$ and when $N\to\infty$, that
$$
\nu_{n,\eps}(U)\frac{f(Y\,|\,\theta)}{f(Y_U\,|\,\theta)}\to_{\|\Delta_n(U)\|\to\infty} 0\,.
$$
Since $g$ is bounded, then $\esp\left\{{f(Y\,|\,\theta)}\slash{f(Y_U\,|\,\theta)^{N/n}}\right\}$ is bounded too.
\end{proof}

\subsection{Proof of Proposition \ref{prop:geo}}
\label{proof:geo}
We preface the proof Proposition \ref{prop:geo} with five Lemmas, some of which are inspired from \cite{medina2016stability}. For notational simplicity, the dependence on $(n,\eps)$ of any ISS-MCMC related quantities is implicit. For all $(\theta,U)\in\Theta\times\Uset_n$, we denote by $\phi_U(\theta)=f(y_U\,|\,\theta)^{N/n}/f(y\,|\,\theta)$ and recall that $a(\theta,\theta')$ is the (exact) MH acceptance ratio so that $\alpha(\theta,\theta')=1\wedge a(\theta,\theta')$. Unless stated otherwise, $\esp$ is the expectation taken under $\nu_{n,\eps}$. For simplicity, $\tK_{n,\eps}$ is written as $\tK_n$.

\begin{lemma}
\label{lem:reply:1}
For any $(\theta,\theta')\in\Theta^2$, we have
$$
\talpha(\theta,\theta')\leq\alpha(\theta,\theta')\left\{1\vee\esp\frac{\phi_U(\theta')}{\phi_U(\theta)}\right\}\,.
$$
\end{lemma}
\begin{proof}
This follows from a slight adaptation of Lemma 3.3 in \cite{medina2016stability}:
\begin{multline*}
\talpha(\theta,\theta')=\esp\left\{1\wedge\frac{f(Y_U\,|\,\theta')^{N/n}p(\theta')Q(\theta',\theta)}{f(Y_U\,|\,\theta)^{N/n}p(\theta)Q(\theta,\theta')}\frac{f(Y\,|\,\theta)f(Y\,|\,\theta')}{f(Y\,|\,\theta)f(Y\,|\,\theta')}\right\}\\
1\wedge \left\{a(\theta,\theta')\esp\frac{\phi_U(\theta')}{\phi_U(\theta)}\right\}
\leq 1\wedge \left[a(\theta,\theta')\left\{\esp\frac{\phi_U(\theta')}{\phi_U(\theta)}\vee 1\right\}\right]
\leq \alpha(\theta,\theta')\left\{\esp\frac{\phi_U(\theta')}{\phi_U(\theta)}\vee 1\right\}\,,
\end{multline*}
where we have used Jensen's inequality and the fact that the inequality $1\wedge ab\leq (1\wedge a)b$ holds for $a>0$ and $b\geq 1$.
\end{proof}
\begin{lemma}
\label{lem:reply:2}
For any $\theta\in\Theta$ and all $\delta>0$, we have
$$
\trho(\theta)-\rho(\theta)\leq\delta+2\sup_{\theta\in\Theta}\proba\left\{\left|\phi_U(\theta)-1\right|\geq \frac{\delta}{2}\right\}\,.
$$
\end{lemma}
\begin{proof}
The proof is identical to proof of Lemma 3.2 in \cite{medina2016stability} by noting that Lemma 3.1 in the same reference holds for two random variables $\phi_U(\theta)$ and $\phi_{U}(\theta')$ that are not independent, \ie for all $(\theta,\theta')\in\Theta^2$ any $U\in\Uset_n$ and all $\delta\in(0,1)$
$$
\proba\left\{\frac{\phi_U(\theta)}{\phi_U(\theta')}\leq 1-\delta\right\}\leq 2\sup_{\theta\in\Theta}\proba\left\{|\phi_U(\theta)-1|\geq \delta/2\right\}\,.
$$
\end{proof}

\begin{lemma}
\label{lem:reply:3}
Assume that \textbf{A.}\ref{assu:4} holds. Then we have
$$
\sup_{(\theta,\theta')\in\Theta^2}1\vee \esp\bigg\{\frac{\phi_U(\theta)}{\phi_U(\theta')}\bigg\}\leq \esp\left\{e^{2\gamma\|\Delta_n(u)\|}\right\}\,.
$$
\end{lemma}
\begin{proof}
Using Cauchy-Schwartz inequality, we write that for all $(\theta,\theta')\in\Theta^2$,
\begin{multline}
\label{eq:2}
\esp\left\{\frac{\phi_U(\theta')}{\phi_U(\theta)}\right\}=\esp\left\{\frac{f(Y_U\,|\,\theta')^{N/n}}{f(Y\,|\,\theta')}\frac{f(Y\,|\,\theta)}{f(Y_U\,|\,\theta)^{N/n}}\right\}\\
\leq \left[\esp\left\{\frac{f(Y_U\,|\,\theta')^{N/n}}{f(Y\,|\,\theta')}\right\}^2\right]^{1/2}
\left[\esp\left\{\frac{f(Y\,|\,\theta)}{f(Y_U\,|\,\theta)^{N/n}}\right\}^2\right]^{1/2}\,.
\end{multline}
Now for all $\theta\in\Theta$, we define the event $\event_\theta:=\{U\in\Uset_n\,,\;f(Y\,|\,\theta)\leq f(Y_U\,|\,\theta)^{N/n}\}$ so that
$$
\esp\left\{\frac{f(Y_U\,|\,\theta)^{N/n}}{f(Y\,|\,\theta)}\right\}^2=\esp\left\{\frac{f(Y_U\,|\,\theta)^{N/n}}{f(Y\,|\,\theta)}\1_{\event_\theta}(U)\right\}^2+
\esp\left\{\frac{f(Y_U\,|\,\theta)^{N/n}}{f(Y\,|\,\theta)}\1_{\overline{\event_\theta}}(U)\right\}^2
$$
and we note that for all $(\theta,U)\in\Theta\times\Uset_n$, Eq. \eqref{eq:assu4} writes
$$
\left\{\frac{f(Y_U\,|\,\theta)^{N/n}}{f(Y\,|\,\theta)}\right\}^2\1_{\event_\theta}(U)\leq e^{2\gamma\|\Delta_n(U)\|}\1_{\event_\theta}(U)\,,
$$
but also
$$
\left\{\frac{f(Y_U\,|\,\theta)^{N/n}}{f(Y\,|\,\theta)}\right\}^2\1_{\overline{\event_\theta}}(U)\leq e^{2\gamma\|\Delta_n(U)\|}\1_{\overline{\event_\theta}}(U)\,,
$$
so that
\begin{multline*}
\esp\left\{\frac{f(Y_U\,|\,\theta)^{N/n}}{f(Y\,|\,\theta)}\1_{\event_\theta}(U)\right\}^2+
\esp\left\{\frac{f(Y_U\,|\,\theta)^{N/n}}{f(Y\,|\,\theta)}\1_{\overline{\event_\theta}}(U)\right\}^2\\
\leq
\esp\left\{e^{2\gamma\|\Delta_n(U)\|}\1_{\event_\theta}(U)\right\}+
\esp\left\{e^{2\gamma\|\Delta_n(U)\|}\1_{\overline{\event_\theta}}(U)\right\}
=\esp\left\{ e^{2\gamma\|\Delta_n(U)\|}\right\}\,.
\end{multline*}
A similar argument gives the same upper bound for $\esp\left\{{f(Y\,|\,\theta)}\slash{f(Y_U\,|\,\theta)^{N/n}}\right\}^2$ so that Eq. \eqref{eq:2} yields
$$
\esp\left\{\frac{\phi_U(\theta')}{\phi_U(\theta)}\right\}\leq \esp\left\{ e^{2\gamma\|\Delta_n(U)\|}\right\}\,.
$$
The proof is completed by noting that for three numbers $a,b$ and $c$,  $c>b\Rightarrow a\vee b\leq a\vee c$ and $\gamma\|\Delta_n(U)\|>0$.
\end{proof}

\begin{lemma}
Assume that \textbf{A.}\ref{assu:4} holds. Then we have for all $\theta\in\Theta$ and $\delta>0$
$$
\proba\left\{\left|\phi_U(\theta)-1\right|\geq \delta/2\right\}\leq \frac{2\gamma}{\log(1+\delta/2)}\esp\{\|\Delta_n(U)\|\}\,.
$$
\end{lemma}
\begin{proof}
With the same notations as in proof of Lemma \ref{lem:reply:3} and roughly with the same reasoning we have for all $\theta\in\Theta$ and all $\delta>0$
\begin{multline*}
\proba\left\{\left|\phi_U(\theta)-1\right|\geq \delta/2\right\}=\proba\left\{\left|\phi_U(\theta)-1\right|\geq \delta/2\cap \event_\theta\right\}+
\proba\left\{\left|\phi_U(\theta)-1\right|\geq \delta/2\cap \overline{\event_\theta}\right\}\\
=
\proba\left\{\frac{f(Y_U\,|\,\theta)^{N/n}}{f(Y\,|\,\theta)}\geq 1+\delta/2\cap \event_\theta\right\}+
\proba\left\{\frac{f(Y_U\,|\,\theta)^{N/n}}{f(Y\,|\,\theta)}\leq 1-\delta/2\cap \overline{\event_\theta}\right\}\\
\leq
\proba\left\{e^{\gamma\|\Delta_n(U)\|}\geq 1+\delta/2\cap \event_\theta\right\}+
\proba\left\{e^{-\gamma\|\Delta_n(U)\|}\leq 1-\delta/2\cap \overline{\event_\theta}\right\}\\
\leq
\proba\left\{\gamma\|\Delta_n(U)\|\geq \log(1+\delta/2)\right\}+
\proba\left\{\gamma\|\Delta_n(U)\|\geq -\log(1-\delta/2)\right\}\,,
\end{multline*}
where the first inequality follows by inclusion (on $\event_\theta$) of
$$
\left\{\frac{f(Y_U\,|\,\theta)^{N/n}}{f(Y\,|\,\theta)}\geq 1+\delta/2\right\}\subset \left\{e^{\gamma\|\Delta_n(U)\|}\geq 1+\delta/2\right\}
$$
and similarly for the second term. Now, note that for all $x>0$, $\log(1+x)<-\log(1-x)$ so that
$$
\proba\left\{\left|\phi_U(\theta)-1\right|\geq \delta/2\right\}\leq 2\proba\left\{\gamma\|\Delta_n(U)\|\geq \log(1+\delta/2)\right\}\leq
\frac{2\gamma}{\log(1+\delta/2)}\esp\left\{\|\Delta_n(U)\|\right\}\,,
$$
where the last inequality follows from Markov inequality.
\end{proof}
We study the limiting case where $N$ is fixed and $n\to N$.
\begin{lemma}
\label{lem:reply:4}
Assume $N$ is fixed and let $n\to N$. Then,
$$
\esp\{\|\Delta_n(U)\|\}\to 0\quad\text{and}\quad
\esp\left\{\exp{2\gamma\|\Delta_n(U)\|}\right\}\to 1\,.
$$
\end{lemma}
\begin{proof}
It follows from the fact that when $n\to N$, $\nu_{n,\eps}$ converges to the dirac on $U^\dag=\{1,\ldots,N\}$ and therefore,
$$
\esp\{\|\Delta_n(U)\|\}\to \|\Delta\bS(U^\dag)\|=0\quad\text{and}\quad
\esp\left\{\exp{2\gamma\|\Delta_n(U)\|}\right\}\to \exp{2\gamma\|\Delta\bar{S}(U^\dag)\|}=1\,.
$$
\end{proof}

We can now prove Proposition \ref{prop:geo}:
\paragraph{\textbf{Proposition}}
\textit{Assume that \textbf{A.}\ref{assu:1} and \textbf{A.}\ref{assu:4} hold. If the marginal MH chain $K$ is geometrically ergodic, \ie \textbf{A.}\ref{assu:2} holds, then there exists an $n_0\leq N$ such that for all $n>n_0$, $\tK_n$ is also geometrically ergodic.}

\begin{proof}
By \cite[Theorems 14.0.1 \& 15.0.1]{meyn2009markov}, there exists a function $V:\Xset\to[1,\infty[$, two constants $\lambda\in(0,1)$ and $b<\infty$ and a small set $S\subset \Xset$ such that $K$ satisfies a drift condition:
\begin{equation}
\label{eq:7}
KV\leq \lambda V+b\1_S\,.
\end{equation}

We now show how to use the previous Lemmas to establish the geometric ergodicity of $\tK_n$ for some $n$ sufficiently large. This reasoning is very similar to that presented in \cite[Theorem 3.2]{medina2016stability}.
\begin{multline}
\label{eq:4}
(\tK_n-K)V(\theta)=\int Q(\theta,\rmd\theta')\left(\talpha(\theta,\theta')-\alpha(\theta,\theta')\right)V(\theta')+\left(\trho(\theta)-\rho(\theta)\right)V(\theta)\\
\leq \left(\esp\left\{e^{2\gamma\|\Delta_n(u)\|}\right\}-1\right)\int Q(\theta,\rmd\theta')\alpha(\theta,\theta')V(\theta')+\\
\left(\delta+\frac{2\gamma}{\log(1+\delta/2)}\esp\left\{\|\Delta_n(U)\|\right\}\right)V(\theta)\\
\leq \left(\esp\left\{e^{2\gamma\|\Delta_n(u)\|}\right\}-1\right)\left(\lambda V(\theta)+b\1_S(\theta)-\rho(\theta)V(\theta)\right)+\\
\left(\delta+\frac{2\gamma}{\log(1+\delta/2)}\esp\left\{\|\Delta_n(U)\|\right\}\right)V(\theta)\\
\leq \esp\left\{e^{2\gamma\|\Delta_n(u)\|}\right\}b\1_S(\theta)+\\
\left(\lambda
\left(\esp\left\{e^{2\gamma\|\Delta_n(u)\|}\right\}-1\right)+\delta+\frac{2\gamma}{\log(1+\delta/2)}\esp\left\{\|\Delta_n(U)\|\right\}\right)V(\theta)
\end{multline}
Combining Eq. \eqref{eq:7} with Eq. \eqref{eq:4}, we have that
\begin{multline}
\tK_n V(\theta)\leq
\left\{1+\esp e^{2\gamma\|\Delta_n(u)\|}\right\}b\1_S(\theta)+\\
\left(\lambda
\esp\left\{e^{2\gamma\|\Delta_n(u)\|}\right\}+\delta+\frac{2\gamma}{\log(1+\delta/2)}\esp\left\{\|\Delta_n(U)\|\right\}\right)V(\theta)
\end{multline}
Fix $\eps>0$. From Lemma \ref{lem:reply:4}, there exists $(n_1,n_2)\in\nset^2$ such that
\begin{multline}
\label{eq:5}
n\geq n_1\Rightarrow \esp\exp\{2\gamma\|\Delta_n(U)\|\}-1\leq\eps\,, \\
n\geq n_2\Rightarrow \esp\|\Delta_n(U)\|\leq\eps\log(1+\eps/4)/4\gamma\,.
\end{multline}
Combining Eqs. \eqref{eq:4} and \eqref{eq:5} yields that for all $n\geq n_0:=\max(n_1,n_2)$, we have
\begin{equation}
\label{eq:6}
\tK_n V(\theta)\leq (\eps+1)b\1_S(\theta)+V(\theta)\left({\lambda(\eps+1)}+\delta+\frac{\eps\log(1+\eps/4)}{2\log(1+\delta/2)}\right)\,.
\end{equation}
Taking $\delta=\eps/2$ in Eq. \eqref{eq:6} gives
$$
\tK_n V(\theta)\leq (\eps+1)b\1_S(\theta)+V(\theta)\left\{\eps\left(\lambda+1\right)+\lambda\right\}\,.
$$
To show that $\tK_n$ (for $n>n_0$) satisfies a geometric drift condition, it is sufficient to take $\eps<(1-\lambda)/(1+\lambda)$ and to check that $S$ is also small for $\tK_n$. This is demonstrated exactly as in the proof of \citet[Theorem 3.2]{medina2016stability}.
\end{proof}

\subsection{Proof of Proposition \ref{prop:bound}}
\label{app4}
This proof borrows ideas from the perturbation analysis of uniformly ergodic Markov chains. First, note that by straightforward algebra we have that
\begin{multline}
\label{eq:app4_1}
\|K(\theta,\,\cdot\,)-\tK(\theta,\,\cdot\,)\|\leq \int Q(\theta,\rmd\theta')\esp\left|\alpha(\theta,\theta')-\tilde{\alpha}(\theta,\theta'\,|\,U)\right|\,,\\
\leq \int Q(\theta,\rmd\theta')\esp\left|a(\theta,\theta')-\ta(\theta,\theta'\,|\,U)\right|\,,\\
=\int Q(\theta,\rmd\theta')a(\theta,\theta')\esp\left|1-\frac{\phi_U(\theta')}{\phi_U(\theta)}\right|\,,\\
=\esp\left\{ \int Q(\theta,\rmd\theta')a(\theta,\theta')\left|\phi_U(\theta)-{\phi_U(\theta')}\right|\frac{f(Y\,|\,\theta)}{f(Y_U\,|\,\theta)^{N/n}}\right\}\,,\\
\leq \esp\left\{ \sup_{\theta\in\Theta}\frac{f(Y\,|\,\theta)}{f(Y_U\,|\,\theta)^{N/n}}\int Q(\theta,\rmd\theta')a(\theta,\theta')\left|\phi_U(\theta)-{\phi_U(\theta')}\right|\right\}\,,\\
\leq \esp\left\{ \sup_{\theta\in\Theta}\frac{f(Y\,|\,\theta)}{f(Y_U\,|\,\theta)^{N/n}}\right\}
\sup_{U\in\Uset_n}\int Q(\theta,\rmd\theta')a(\theta,\theta')\left|\phi_U(\theta)-{\phi_U(\theta')}\right|\,.
\end{multline}
Now, under \textbf{A}.\ref{assu:3} and using \citet[Corollary 3.1]{mitrophanov2005sensitivity} we have that for any starting point $\theta_0\in\Theta$,
\begin{equation}
\label{eq:app4_2}
\|K^i(\theta_0,\,\cdot\,)-\tK^i(\theta_0,\,\cdot\,)\|\leq \left(\lambda+\frac{C\rho^\lambda}{1-\rho}\right)\sup_{\theta\in\Theta}\|K(\theta,\,\cdot\,)-\tK(\theta,\,\cdot\,)\|\,,
\end{equation}
where $\lambda=\lceil\log(1/C)/\log \rho\rceil$. Combining Eqs \eqref{eq:app4_1} and \eqref{eq:app4_2} leads to Eq. \eqref{eq:bound2} with $\kappa=\lambda+{C\rho^\lambda}/{1-\rho}$. Moreover, note that using Eq. \eqref{eq:bound2} we have
\begin{multline*}
\sup_{\theta\in\Theta}\|\pi-\tK^i(\theta,\,\cdot\,)\|\leq \sup_{\theta\in\Theta}\|\pi-K^i(\theta,\,\cdot\,)\|+\sup_{\theta\in\Theta}\|K^i(\theta,\,\cdot\,)-\tK^i(\theta,\,\cdot\,)\|\,,\\
\leq
C\rho^i+\kappa A_n\sup_{(\theta,U)\in\Theta\times \Uset_n}B_n(\theta,U)
\end{multline*}
and taking the limit when $i\to\infty$ leads to Eq. \eqref{eq:bound2_bis}. Finally, for a large enough $n$, we know from Proposition \ref{prop:geo} that the marginal Markov chain $\{\ttheta_i\,,i\in\nset\}$ produced by ISS-MCMC is geometrically ergodic and we denote by $\tpi_n$ its stationary distibution. For such a $n$, we have for any $\theta_0\in\Theta$
\begin{multline*}
\|\pi-\tpi_n\|\leq \|K^i(\theta_0,\,\cdot\,)-\pi\|+\|\tK^i(\theta_0,\,\cdot\,)-\tpi_n\|+\|K^i(\theta_0,\,\cdot\,)-\tK^i(\theta_0,\,\cdot\,)\|\\
\leq \|K^i(\theta_0,\,\cdot\,)-\pi\|+\|\tK^i(\theta_0,\,\cdot\,)-\tpi_n\|+\kappa A_n\sup_{(\theta,U)\in\Theta\times \Uset_n}B_n(\theta,U)
\end{multline*}
and taking the limit as $i\to\infty$ yields Eq. \eqref{eq:bound3}.

\subsection{Extension of Proposition \ref{prop:bound} beyond the time homogeneous case}
\label{app4_bis}

We start with the two following remarks relative to the Informed Sub-Sampling Markov chain.

\begin{rem}
Assume $U_0\sim\nu_{n,\eps}$ and $\ttheta_0\sim\mu$ for some initial distribution $\mu$ on $(\Theta,\vartheta)$. The distribution of $U_i$ given $\ttheta_i$ is for some $u\in\Uset_n$,
\begin{equation*}
\proba(U_i=u\,|\,\ttheta_i)
\propto
\sum_{U_0\in\Uset_n}\int_{\ttheta_0\in\Theta} \nu_{n,\eps}(U_0)\mu(\rmd\ttheta_0)\bar{K}^{i}(\ttheta_0,U_0;\ttheta_i,u)\,,
\end{equation*}
where $\barK(\theta,U;\rmd\theta',U'):=K(\theta,\rmd\theta'\,|\,U)H(U,U')$ and $H$ is the transition kernel of the Markov chain $\{U_i,\,i\in\nset\}$. As a consequence $\proba(U_i\in\,\cdot\,|\,\ttheta)$ depends on $\ttheta$ and $i$.
\end{rem}

\begin{rem}
\label{rem1}
The marginal Markov chain $\{\ttheta_i,\,i\in\nset\}$ produced by ISS-MCMC algorithm is  time inhomogeneous since for all $A\in\Xalg$,
\begin{equation}
\tK(\theta_{i-1},A):=\proba(\ttheta_i\in A\,|\,\ttheta_{i-1})=\sum_{u\in\Uset_n}{K}(\ttheta_{i-1},\rmd\ttheta_i\,|\,U_i)\proba(U_i=u\,|\,\ttheta_i)\,,
\end{equation}
and $\proba(U_i=u\,|\,\ttheta_i)$ depends on $i$ (Remark 1). We thus denote by $\tilde{K}_i$ the marginal transition kernel $\ttheta_{i-1}\to\ttheta_i$. However, we observe that if the random variables  $\{U_i,\,i\in\nset\}$ are \iid with distribution $\nu_{n,\eps}$, $K_i$ becomes time homogeneous as $\proba(U_i=u\,|\,\theta_i)=\nu_{n,\eps}(u)$ for all $i$.
\end{rem}

A consequence of Remark 2 is that \citet[Theorem 3.1]{mitrophanov2005sensitivity} does not hold when Assumption \textbf{A}.\ref{assu:1} is not satisfied. Indeed, $\{\ttheta_i,\,i\in\nset\}$ is not a time homogeneous Markov chain in this case and we first need to generalize the result from Mitrophanov in order to apply it to our context. This is presented in Lemma \ref{lem1_tih}.

\begin{lemma}
\label{lem1_tih}
Let $K$ be the transition kernel of an uniformly ergodic Markov chain that admits $\pi$ as stationary distribution. Let $\tilde{K}_i$ be the $i$-th transition kernel of the ISS-MCMC Markov chain. In particular, let $p_i(\,\cdot\,|\,\theta):=\proba(U_i\in\,\cdot\,|\,\theta)$ be the distribution of the random variable $U_i$, used at iteration $i$ of the noisy Markov chain given $\theta$. We have:
\begin{equation}
\label{eq:noisyThm}
\lim_{i\to\infty}\|\pi-\tpi_i\|\leq\kappa\sup_{\theta\in\Theta}\sup_{i\in\nset}\int\delta_i(\theta,\theta')Q(\theta,\rmd\theta')\,,
\end{equation}
where $\delta_i:\Theta\times\Theta\to\rset^+$ is a function that satisfies
$$
\esp_i\left\{\left|a(\theta,\theta')-\ta(\theta,\theta'\,|\,U)\right|\right\}\leq \delta_i(\theta,\theta')
$$
and the expectation is under $p_i(\,\cdot\,|\,\theta)$.
\end{lemma}

\begin{proof}
In addition of the notations of Section \ref{sec:alg}, we define the following quantities for a Markov transition kernel regarded as an operator on $\meas$, the space of signed measures on $(\Theta,\borel(\Theta))$: $\tau(K):=\sup_{\pi\in\meas_{0,1}}\|\pi K\|$ is the ergodicity coefficient of $K$, $\|K\|:=\sup_{\pi\in\meas_{1}}\|\pi K\|$ is the operator norm of $K$ and $\meas_1:=\{\pi\in\meas,\,\|\pi\|=1\}$ and $\meas_{0,1}:=\{\pi\in\meas_1,\,\pi(\Theta)=0\}$.

Remarks 1 and 2 explain why, in general, $\{\ttheta_i,\,i\in\nset\}$ is a time-inhomogeneous Markov chain with transition kernel $\{\tK_i,\,i\in\nset\}$. For each $i\in\nset$, define $\pi_i$ as the distribution of $\theta_i$ produced by the Metropolis-Hastings algorithm (Alg. \ref{alg:mh}) with transition kernel $K$, referred to as the exact kernel hereafter. Our proof is based on the following identity:
\begin{multline}
\label{eq:app3_1}
K^i-\tK_1\tK_2\cdots\tK_i=(K-\tK_1)K^{i-1}+\tK_1(K-\tK_2)K^{i-2}+\tK_1\tK_2(K-\tK_3)K^{i-3}+\cdots\\
+\tK_1\cdots\tK_{i-1}(K-\tK_i)\,,
\end{multline}
for each $i\in\nset$. Equation \eqref{eq:app3_1} will help translating the proof of Theorem 3.1 in \cite{mitrophanov2005sensitivity} to the time-inhomogeneous setting and in particular, we have for each $i\in\nset$:
\begin{equation}
\label{eq:app3_2}
\pi_i-\tpi_i=(\pi_0-\tpi_0)K^i+\sum_{j=0}^{i-1}\tpi_j(K-\tK_{j+1})K^{i-j-1}\,.
\end{equation}
Following the proof of Theorem 3.1 in \cite{mitrophanov2005sensitivity}, we obtain
\begin{multline}
\label{eq:app3_3}
\|\pi_i-\tpi_i\|\leq \|\pi_0-\tpi_0\|\tau(K^i)+\sum_{j=0}^{i-1}\|K-\tK_{i-j}\|\tau(K^{j})\,,\\
\leq
\left\{
\begin{array}{lc}
\|\pi_0-\tpi_0\|+i\sup_{j\leq i}\|K-\tK_j\|&\text{if }i\leq \lambda\\
\|\pi_0-\tpi_0\|C\rho^i+\sup_{j\leq i}\|K-\tK_j\|\left\{\lambda+C\frac{\rho^\lambda-\rho^{i}}{1-\rho}\right\}&\text{else}\\
\end{array}
\right.
\end{multline}
where $\lambda=\left\lceil {\log_\rho(1/C)}\right\rceil$. Without loss of generality, we take $\pi_0=\tpi_0$ and since $\|\pi-\tpi_i\|\leq \|\pi-\pi_i\|+\|\pi_i-\tpi_i\|$ we have for all $i>\lambda$ that
\begin{equation}
\|\pi-\tpi_i\|\leq\left\{\lambda+C\frac{\rho^\lambda-\rho^{i}}{1-\rho}\right\} \sup_{j\leq i}\|K-\tK_j\|\,.
\end{equation}
Taking the limit as $i\to\infty$ leads to
\begin{equation}
\label{eq:app3_5}
\lim_{i\to\infty}\|\pi-\tpi_i\|\leq \left\{\lambda+C\frac{\rho^\lambda}{1-\rho}\right\} \sup_{i\in\nset}\|K-\tK_i\|\,.
\end{equation}

Using a similar derivation than in the proof of Corollary 2.3 in \cite{alquier2014noisy}, we obtain
$$
\|K-\tK_i\|\leq \sup_{\theta\in\Theta}\int Q(\theta,\rmd\theta')\esp_i\left|a(\theta,\theta')-\ta(\theta,\theta'\,|\,U_i)\right|\,,
$$
where the expectation is under $p_i(\,U\,|\,\theta)$ and which combined with \eqref{eq:app3_5} leads to
$$
\lim_{i\to\infty}\|\pi_i-\tpi_i\|\leq \left(\lambda+C\frac{\rho^\lambda}{1-\rho}\right)\sup_{\theta\in\Theta}\sup_{i\in\nset}\esp_i\left|a(\theta,\theta')-\ta(\theta,\theta'\,|\,U_i)\right|
$$
where the expectation is under $Q(\theta,\cdot)\otimes p_i(\,\cdot\,|\,\theta)$. Any upper bound $\delta_i(\theta,\theta')$ of the expectation on the right hand side yields \eqref{eq:noisyThm}.
\end{proof}

By straightforward algebra, we have:
\begin{equation}
\esp_i\left|a(\theta,\theta')-\ta(\theta,\theta'\,|\,U_i)\right|=a(\theta,\theta')\esp_i\left\{\frac{f(Y\,|\,\theta)}{f(Y_U\,|\,\theta)^{N/n}}\left|\phi_U(\theta)-\phi_U(\theta')\right|\right\}
\end{equation}
where we have defined $\phi_U(\theta)=f(Y_U\,|\,\theta)^{N/n}\slash f(Y\,|\,\theta)$. Using Lemma \ref{lem1_tih}, we have that
\begin{multline}
\label{eq:app4}
\lim_{i\to\infty}\|\pi-\tpi_i\|\leq \kappa\sup_{\theta\in\Theta}\sup_{i\in\nset}\esp_i\left\{\sup_{\theta\in\Theta}\frac{f(Y\,|\,\theta)}{f(Y_U\,|\,\theta)^{N/n}}\int Q(\theta,\rmd\theta')a(\theta,\theta')\left|\phi_U(\theta)-\phi_U(\theta')\right|\right\}\,,\\
\leq \kappa\sup_{\theta\in\Theta}\sup_{i\in\nset}\esp_i\left\{\sup_{\theta\in\Theta}\frac{f(Y\,|\,\theta)}{f(Y_U\,|\,\theta)^{N/n}}\right\}\sup_{(\theta,U)\in\Theta\times\Uset_n}\int Q(\theta,\rmd\theta')a(\theta,\theta')\left|\phi_U(\theta)-\phi_U(\theta')\right|\,.
\end{multline}
which is the counterpart of \eqref{eq:bound2_bis} when Assumption \textbf{A}.\ref{assu:1} does not hold. We note that the second supremum in Eq. \eqref{eq:app4} is in fact $B_n$ defined at Eq. \eqref{eq:B_n} and, as such, can be controlled as described in Section \ref{sec:alg:prop}.  However, this is not clearly the case for the first supremum in Eq. \eqref{eq:app4} which differs from $A_n$ defined at Eq. \eqref{eq:d2}:
\begin{equation}
    \label{eq:reply4}
    \tA_n:=\sup_i\sup_\theta\esp_i\{\sup_\theta 1/\phi_U(\theta)\}\neq\esp\{\sup_\theta 1/\phi_U(\theta)\}=A_n\,.
\end{equation}
We now show that, under two additional assumptions (\textbf{A}.\ref{assum_reply} and \textbf{A}.\ref{assum2_reply}), the control based on the summary statistics also applies to the time inhomogeneous case when Assumption \textbf{A}.\ref{assu:1} does not hold.

\begin{assumption}[One-step minorization]
\label{assum_reply}
For all $i\in\nset$ and all $A\in\vartheta$, there exists some $\eta>0$ such that $p_i(A)>\eta\lambda(A)$ where $\lambda$ is the Lebesgue measure.
\end{assumption}

This assumption typically holds if $\Theta$ is compact or if the chain $\{\ttheta_i,U_i\}_i$ admits a minorization condition. Since we assume, in this discussion, that the exact MH Markov chain is uniformly ergodic and as such satisfy a minorization condition, see \eg \citet[Thm 16.2.3]{meyn2009markov} and \cite{hobert2004mixture}. We may study conditions on which $\{\ttheta_i\}_i$ inherits this property and leave this for future work but already note that Assumption \textbf{A}.\ref{assum_reply} is not totally unrealistic.

\begin{assumption}
\label{assum2_reply}
The marginal Markov chain $\{U_i\}_i$ has initial distribution $U_0\sim \nu_{n,\eps}$.
\end{assumption}

Even though this assumption is difficult to meet in practice as $|\Uset_n|$ may be very large, the discussion at the beginning of Section \ref{sec:6_1} indicates an approach to set the distribution of $U_0$ close from $\nu_{n,\eps}$.

Again, while the Assumptions \ref{assum_reply} and \ref{assum2_reply} are perhaps challenging to guarantee, Proposition \ref{prop_reply} aims at giving some level of confidence to the user that the ISS-MCMC method is useful, even when Assumption \textbf{A.}\ref{assu:1} does not hold. In addition, it reinforces the importance of choosing summary statistics that satisfy Assumption \textbf{A.}\ref{assu:4}.

\begin{prop}
\label{prop_reply}
Assume that Assumptions  \textbf{A.}\ref{assu:2}, \textbf{A.}\ref{assu:4}, \textbf{A.}\ref{assum_reply} and \textbf{A.}\ref{assum2_reply} hold. Then there exists a positive number $M>0$ such that
\begin{equation}
\tA_n\leq M A_n\,,
\end{equation}
where $A_n$ and $\tA_n$ have been defined at Eq. \eqref{eq:reply4}.
\end{prop}

\begin{coro}
Under the same Assumptions as Proposition \ref{prop_reply}, the control explained in Section \ref{sec:alg:sumstat} is also valid in the time inhomogeneous case.
\end{coro}

\begin{proof}[Proof of Proposition \ref{prop_reply}]
From Assumption A4, there exists some $\gamma>0$ such that
\begin{equation}
\label{eq:reply}
\tA_n=\sup_i\sup_\theta\esp_{i}\{f(Y\,|\,\theta)/f(Y_U\,|\,\theta)^{N/n}\}\leq \sup_i\sup_\theta\int \rmd p_i(U\,|\,\theta)e^{\gamma \|\Delta_n(U)\|}\,,
\end{equation}
where $\rmd p_i(U\,|\,\theta)=p_i(U\,|\,\theta)\rmd U$ and $\rmd U$ is the counting measure. Now, the conditional probability writes:
\begin{equation*}
p_i(U\in\cdot\,|\,\theta):=\proba(U_i\in\cdot\,,\,\ttheta_i\in\rmd\theta)/\proba(\ttheta_i\in\rmd\theta)\,.
\end{equation*}
On the one hand, Lemma \ref{lemma_reply} shows that there exists a bounded function $f_i$ such that $\proba(U_i\in\cdot\,,\,\ttheta_i\in\rmd\theta)\leq f_i(\ttheta)\rmd \theta\nu_{n,\eps}(\,\cdot\,)$. On the other hand, Assumption \ref{assum_reply} guarantees that there exists some $\eta>0$ such that for all $\ttheta\in\Theta$, $\proba(\ttheta_i\in\rmd\theta)>\eta\rmd\theta$. Combining those two facts allows to write that
\begin{equation}
\label{eq:reply2}
p_i(U\in\cdot\,|\,\theta)\leq \frac{f_i(\theta)\rmd \theta\nu_{n,\eps}(\,\cdot\,)}{\eta\rmd\theta}=\frac{f_i(\theta)}{\eta}\nu_{n,\eps}(\,\cdot\,)\,.
\end{equation}
Plugging Eq. \eqref{eq:reply2} into Eq. \eqref{eq:reply}, yields to
\begin{equation*}
\tA_n\leq\sup_i\sup_\theta\int \rmd p_i(U\,|\,\theta)e^{\gamma \|\Delta_n(U)\|}\leq \sup_\theta\sup_i \frac{f_i(\ttheta)}{\eta} A_n\,,
\end{equation*}
which completes the proof, setting $M:=\sup_\theta\sup_i f_i(\theta)/\eta$.
\end{proof}

\begin{lemm}
\label{lemma_reply}
Assume that Assumptions  \textbf{A.}\ref{assu:2}, \textbf{A.}\ref{assu:4}, \textbf{A.}\ref{assum_reply} and \textbf{A.}\ref{assum2_reply} hold. In addition, let us assume that $U_0\sim\nu_{n,\eps}$. Then $p_i(\theta,U)$ is dominated by $\rmd \theta\rmd U$ where $\rmd \theta$ and $\rmd U$ implicitly refer  to the Lebesgue and the counting measure, respectively. In other words there is a sequence of bounded functions $\{f_i:\Theta\to\rset^+\}$ such that
\begin{equation}
\label{eq:reply3}
  \rmd p_i(U,\ttheta)\leq f_i(\theta)\rmd\ttheta\rmd\nu_{n,\eps}(U)\,.
\end{equation}
\end{lemm}

\begin{proof}
We proceed by induction. Defining $\varrho(\ttheta\,|\,U)$ as the probability to reject a MH move for the parameter $\ttheta$ when the subset variable is $U$, we recall that $\varrho(\ttheta\,|\,U)<1$ and $\talpha(\ttheta,\ttheta'\,|\,U)<1$. By assumption on the proposal kernel, it satisfies $Q(\ttheta, \,\rmd\ttheta'\,)=Q(\ttheta,\ttheta')\rmd\ttheta'$ and define the function $\bQ:\theta\mapsto\sup_{\ttheta'\in\Theta}Q(\ttheta',\theta)$. Similarly, we define the function $\bvarrho:\theta\mapsto \sup_{U\in\Uset_n}\varrho(\theta\,|\,U)$. Deriving the calculation separately for the continuous and the diagonal parts of the Metropolis-Hastings kernel $K(\theta,\cdot\,|\,U)$ (see Eq. \eqref{eq:8}), we have:
\begin{multline*}
\rmd p_1(U,\ttheta)=\int_{\ttheta_0\in\Theta}\sum_{U_0\in\Uset_n} \mu(\rmd\ttheta_0)\nu(U_0)H(U_0,U)K(\ttheta_0,\rmd\ttheta\,|\,U)\,,\\
\leq\int\sum\mu(\rmd\ttheta_0)\nu(U_0)H(U_0,U)Q(\ttheta_0,\rmd\ttheta)\talpha(\ttheta_0,\ttheta\,|\,U)\\
+\int\sum\mu(\rmd\ttheta_0)\nu(U_0)H(U_0,U)\delta_{\ttheta_0}(\rmd\ttheta)\varrho(\ttheta_0\,|\,U)\,,\\
\leq \int\sum\mu(\rmd\ttheta_0)\nu(U_0)H(U_0,U)\bQ(\ttheta)\rmd\ttheta+
\int\sum\mu(\rmd\ttheta)\nu(U_0)H(U_0,U)\bvarrho(\ttheta)\,,\\
\leq \sum\nu(U_0)H(U_0,U)\bQ(\ttheta)\rmd\ttheta+
\sum\nu(U_0)H(U_0,U)\mu(\ttheta)\bvarrho(\ttheta) \rmd\ttheta\,,\\
=\underbrace{\left\{\bQ(\ttheta)+\mu(\ttheta)\bvarrho(\ttheta)\right\}}_{:=f_1(\ttheta)}\rmd\ttheta\rmd\nu(U)\,,
\end{multline*}
where the last equality follows from the $\nu_{n,\eps}$-stationarity of $H$. In this derivation, we have defined $\mu$ as the initial distribution of the Markov chain $\{\ttheta_i\}_i$  and $\nu$ as a shorthand notation for $\nu_{n,\eps}$. Now, let us assume that there is a bounded function $f_{i-1}$ such that $\rmd p_1(U,\ttheta)\leq f_{i-1}(\ttheta)\rmd\ttheta\rmd\nu(U)$. Using the notation $\mu K:=\int \mu(\rmd x)K(x,\cdot)$ for any Markov kernel $K$ and a measure $\mu$ on some measurable space $(\Xset,\Xalg)$ and recalling that $\barK$ is the transition kernel of ISS-MCMC on the extended space $\Theta\times\Uset_n$, we have:
\begin{multline*}
\rmd p_i(U,\ttheta)=\sum_{U_{i-1}\in\Uset_n}\int_{\ttheta_{i-1}\in\Theta}\bar{\mu} \barK^{i-1}(U_{i-1},\rmd \ttheta_{i-1})H(U_{i-1},U)K(\ttheta_{i-1},\rmd\ttheta\,|\,U)\,,\\
\leq  \sum_{U_{i-1}\in\Uset_n}\int_{\ttheta_{i-1}\in\Theta}\bar{\mu} \barK^{i-1}(U_{i-1},\rmd \ttheta_{i-1})H(U_{i-1},U)\bQ(\ttheta)\rmd\ttheta
\\
+\sum_{U_{i-1}\in\Uset_n}\bar{\mu} \barK^{i-1}(U_{i-1},\rmd \ttheta)H(U_{i-1},U)\varrho(\ttheta\,|\,U)\,,\\
\leq\sum_{U_{i-1}\in\Uset_n}\bar{\mu} \bar{K}^{i-1}(U_{i-1})H(U_{i-1},U)\bQ(\ttheta)\rmd\ttheta+
\sum_{U_{i-1}\in\Uset_n} \rmd p_{i-1}(U_{i-1},\ttheta)H(U_{i-1},U)\varrho(\ttheta\,|\,U)\\
\leq \nu(U)\bQ(\ttheta)\rmd\ttheta+f_{i-1}(\ttheta)\sum_{U_{i-1}}H(U_{i-1},U)\varrho(\ttheta\,|\,U)\leq \underbrace{\left\{\bQ(\ttheta)+f_{i-1}(\ttheta)\bvarrho(\ttheta)\right\}}_{:=f_i(\ttheta)}\rmd\ttheta\rmd\nu(U)
\end{multline*}
and $f_i$ is bounded. The first term in the third inequality follows from noting that
\begin{multline*}
\sum \bar{\mu}\barK^{i-1}(U_{i-1})H(U_{i-1},U_i)\\
=\sum\int\bar{\mu}\barK^{i-2}(U_{i-2},\rmd\ttheta_{i-2})\sum\int H(U_{i-2},U_{i-1})K(\ttheta_{i-2},\rmd\ttheta_{i-1}\,|\,U_{i-1})H(U_{i-1},U)\\
=\sum\int \bar{\mu}\barK^{i-2}(U_{i-2},\rmd\ttheta_{i-2}) H^2(U_{i-2},U)=\cdots
=\sum\int \mu(\rmd\theta_0)\nu(U_0)H^{i}(U_0,U)=\nu(U)\,.
\end{multline*}
\end{proof}

\subsection{Proof of Proposition \ref{prop:bound2}}
\label{app5}

\begin{proof}
Note that for all $(\theta,\zeta)\in\Theta\times \rset^d$, a Taylor expansion of $\pi(\theta)$ and $\phi_U(\theta)$ at $\theta+\Sigma\zeta$ in \eqref{eq:boundB} combined to the triangle inequality leads to:
\begin{multline*}
B(U,\theta)\leq\frac{1}{\sqrt{N}}\esp\left\{\left|(M\zeta)\T\grad_\theta\phi_U(\theta)\right|\left(1+\frac{1}{\sqrt{N}}(M\zeta)\T\grad_\theta\log\pi(\theta)\right)\right\}\\
+\frac{1}{2N}\esp\left\{|(M\zeta)\T\grad_\theta^2\phi_U(\theta)M\zeta|\right\}+\esp\{R(\|M\zeta\|/\sqrt{N})\}\,,
\end{multline*}
where the expectation is under $\Phi_d$ and $R(x)=o(x)$ at $0$. Applying Cauchy-Schwartz gives:
\begin{multline*}
B(U,\theta)\leq \frac{1}{\sqrt{N}}\esp\{\|M\zeta\|\}\|\grad_\theta\phi_U(\theta)\|+\frac{1}{N}\esp\{\|M\zeta\|^2\}\|\grad_\theta\phi_U(\theta)\|\|\grad_\theta\log\pi(\theta)\|\\
+\frac{1}{2N}\esp\{|\zeta\T M\T\grad_\theta^2\phi_U(\theta)M\zeta|\}+\esp\{R(\|M\zeta\|/\sqrt{N})\}\,.
\end{multline*}
Now, we observe that:
\begin{itemize}
\item $\esp\{\|M\zeta\|\}=\esp\{\sum_{i=1}^d(\sum_{j=1}^d M_{i,j}\zeta_{j})^2\}^{1/2}\leq \esp\{\sum_{i=1}^d|\sum_{j=1}^d M_{i,j}\zeta_{j}|\}\leq
\esp\{\sum_{i=1}^d\sum_{j=1}^d|M_{i,j}||\zeta_{j}|\}=\sum_{i=1}^d\sum_{j=1}^d|M_{i,j}|\esp\{|\zeta_i|\}=\sqrt{\frac{2}{\pi}}\|M\|_1
$
\item $\esp\{\|M\zeta\|^2\}=\esp\{\sum_{i=1}^d(\sum_{j=1}^d M_{i,j} \zeta_j)^2\}=\sum_{i=1}^d\esp\{(\sum_{j=1}^dM_{i,j}\zeta_j)^2\}
=\sum_{i=1}^{d}\var(\sum_{j=1}^d M_{i,j} \zeta_j)=\sum_{i=1}^{d}\sum_{j=1}^d M_{i,j}^2 \var(\zeta_j)=\|M\|_2^2
$
\item considering the quadratic form associate to the operator $T(U,\theta)=M\T\grad_{\theta}^2\phi_U(\theta) M$, noting that $T(U,\theta)$ is symmetric its eigenvalues $\lambda_1\geq \lambda_2\geq \cdots\geq \lambda_d$ are real and we have
$$
\zeta\T T(U,\theta) \zeta\leq \lambda_1\|\zeta\|^2
$$
so that:
$$
\esp\left\{|(M\zeta)\T\grad_\theta^2\phi_U(\theta)M\zeta|\right\}\leq d \sup_{i}|\lambda_i| \leq d \tn M\T\grad_{\theta}^2\phi_U(\theta) M\tn
$$
where for any square matrix $A$, we have defined $\tn A\tn=\sup_{x\in\rset^d,\|x\|=1}\|Ax\|$ as the operator norm.
\end{itemize}
\end{proof}

\section{Proof of Proposition \ref{propMLE}}
\label{proof:propMLE}

In this section, we are assuming that there is an infinite stream of observations $(Y_1,Y_2,\ldots)$ and a parameter $\theta_0\in\Theta$ such that
$Y_i\sim f(\,\cdot\,|\,\theta_0)$. Let $\rho>1$ be a constant defined as the ratio $N/n$ \ie the size of the full dataset over the size of the subsamples of interest. The full dataset is thus $Y_{1:\rho n}$. We define the set
$$
\Uset_n^\rho=\left\{U\subset\{1,\ldots,\rho n\},\;|U|=n\right\}
$$
such that $Y_U$ ($U\in\Uset_n^\rho$) is the set of subsamples of interest. We study the asymptotics when $n\to\infty$ \ie we let the whole dataset and the size of subsamples of interest grow at the same rate.

\setcounter{prop}{6}
\begin{prop}
Let $\thetaMLE_{\rho n}$ be the MLE of $Y_1,\ldots,Y_{\rho n}$ and $\thetaMLE_U$ be the MLE of the subsample $Y_U$ ($U\in\Uset_n^\rho$). Assume that there exists a compact set $\kappa_n\subset\Theta$ such that $(\thetaMLE_{\rho n},\theta_0)\in\kappa_n^2$ and for all $U$, there exists a compact set $\kappa_U\subset\Theta$ such that $(\thetaMLE_U,\theta_0)\in\kappa_U^2$. Then, there exists a constants $\beta$, a metric $\|\cdot\|_{\theta_0}$ on $\Theta$ and a non-decreasing subsequence $\{\sigma_n\}_{n\in\nset}$, ($\sigma_n\in\nset$) such that for all $U\in\Uset_{\sigma_n}^\rho$, we have for $p$-almost all $\theta\in\kappa_n\cap\kappa_U$
\begin{equation}
\label{eq0_bis}
\log f(Y_{1:\rho \sigma_n}\,|\,\theta)-\rho\log f(Y_U\,|\,\theta)\leq
H_n(Y,\theta)+\beta+\frac{\rho \sigma_n}{2}\|\thetaMLE_U-\thetaMLE\|_{\theta_0}\,,
\end{equation}
where
$$
\underset{n\to \infty}{\plim}\quad H_n(Y,\theta)\overset{\proba_{\theta_{0}}}{=} 0\,.
$$
\end{prop}

\begin{proof}
Fix $n\in\nset$. Consider the case where the prior distribution $p$ is uniform on $\kappa_n$. In this case, the posterior is
$$
\pi_{n}(\theta\,|\,Y_{1:\rho n})=f(Y_{1:\rho n}\,|\,\theta) \1_{\kappa_n}(\theta)\big\slash Z_{\rho n}\,,\qquad Z_{\rho n}=\int_{\kappa_n} f(Y_{1:\rho n}\,|\,\theta)\rmd \theta
$$
and from corollary \ref{coro1}, we know that there exists a subsequence $\tau_n\subset\nset$ such that for $p$-almost all $\theta\in\kappa_n$
\begin{equation}
\label{eq4}
\left|\log \frac{f(Y_{1:\rho \tau_n}\,|\,\theta)}{Z_{\rho \tau_n}}-\log\Phi_{\rho\tau_n}(\theta)\right|\overset{\proba_{\theta_{0}}}{\to}0\,,
\end{equation}
where $\theta\mapsto\Phi_{\rho\tau_n}(\theta)$ is the pdf of $\norm(\thetaMLE_{\rho \tau_n},I(\theta_0)^{-1}/\rho\tau_n)$. Similarly, there exists another subsequence $\gamma_n\subset \nset$ such that for all $U\in\Uset_{\gamma_n}^\rho$ and for $p$-almost all $\theta\in\kappa_U$
\begin{equation}
\label{eq5}
\left|\rho\log \frac{f(Y_U\,|\,\theta)}{Z_{\gamma_n}(U)}-\rho\log\Phi_{U}(\theta)\right|\overset{\proba_{\theta_{0}}}{\to}0\,,\qquad Z_{\gamma_n}(U)=\int_{\kappa_{U}} f(Y_U\,|\,\theta)\rmd \theta
\end{equation}
where $\theta\mapsto\Phi_{U}(\theta)$ is the pdf of $\norm(\thetaMLE_{U},I(\theta_0)^{-1}/|U|)$. Let $\{\sigma_n\}_{n\in\nset}$ be the sequence defined as $\sigma_n=\max\{\tau_n,\gamma_n\}$. We know from \eqref{eq4} and \eqref{eq5} that for all $\varepsilon>0$ and all $\eta>0$, there exists $n_1\in\nset$ such that for all $U\in\Uset_{\sigma_n}^{\rho}$ and for all $n\geq n_1$
\begin{equation}
\label{eq6}
\proba_{\theta_0}\left\{\left|\log \frac{f(Y_{1:\rho \sigma_n}\,|\,\theta)}{Z_{\rho\sigma_n}}-\log\Phi_{\rho\sigma_n}(\theta)\right|
+\left|\rho\log \frac{f(Y_U\,|\,\theta)}{Z_{\sigma_n}(U)}-\rho\log\Phi_{U}(\theta)\right|
\geq \varepsilon\right\}\leq \eta\,.
\end{equation}
Now, by straightforward algebra, we have for any $U\in\Uset_{\sigma_n}^{\rho}$
\begin{multline}
\label{eq7}
\log f(Y_{1:\rho\sigma_n}\,|\,\theta)-\rho \log f(Y_U\,|\,\theta)=\log \frac{f(Y_{1:\rho\sigma_n}\,|\,\theta)}{Z_{\rho\sigma_n}}-\log \Phi_{\rho\sigma_n}(\theta)
-\rho \log \frac{f(Y_U\,|\,\theta)}{Z_{\sigma_n}(U)}\\+\rho\log \Phi_{U}(\theta)
+\log\frac{Z_{\rho\sigma_n}}{Z_{\sigma_n}(U)^\rho}+\log\Phi_{\rho\sigma_n}(\theta)-\rho\log\Phi_U(\theta)\\
\leq
\left|\log \frac{f(Y_{1:\rho\sigma_n}\,|\,\theta)}{Z_{\rho\sigma_n}}-\log \Phi_{\rho\sigma_n}(\theta)-
\rho \log \frac{f(Y_U\,|\,\theta)}{Z_{\sigma_n}(U)}+\rho\log \Phi_{U}(\theta)\right|\\
+\log\frac{Z_{\rho\sigma_n}}{Z_{\sigma_n}(U)^\rho}
+ (\rho-1)\log(2\pi)^{d/2}+\frac{\rho \sigma_n}{2}\bigg|\|\theta-\theta^\ast_U\|_{\theta_0}-
\|\theta-\theta^\ast\|_{\theta_0}\bigg|\\
\leq
\left|\log \frac{f(Y_{1:\rho\sigma_n}\,|\,\theta)}{Z_{\rho\sigma_n}}-\log \Phi_{\rho\sigma_n}(\theta)\right|+
\left|\rho \log \frac{f(Y_U\,|\,\theta)}{Z_{\sigma_n}(U)}-\rho\log \Phi_{U}(\theta)\right|\\
+\log\frac{Z_{\rho\sigma_n}}{Z_{\sigma_n}(U)^\rho}
+ (\rho-1)\log(2\pi)^{d/2}+\frac{\rho \sigma_n}{2}\|\theta^\ast_U-\theta^\ast\|_{\theta_0}\,,
\end{multline}
where we have used Lemma \ref{lem2} for the first inequality and the triangle inequalities for the second. Combining \eqref{eq7} with \eqref{eq6} yields \eqref{eq0}.
\end{proof}

\begin{lemm}
\label{lem1}
Consider a posterior distribution $\pi_n$ given $n$ data $Y_{1:n}$ where $p$ is the prior distribution and its Bernstein-von Mises approximation is $\Phi_n=\norm(\thetaMLE(Y_{1:n}),I(\theta_0)^{-1}/n)$. There exists a subsequence $\{\tau_n\}_n\subset\nset$ such that
\begin{equation}
\label{eq1}
\underset{n\to\infty}{\plim}\left|\pi_{\tau_n}(\theta)-\Phi_{\tau_n}(\theta)\right|\overset{\proba_{\theta_0}}{=}0\,,\quad \text{for}\;p\text{-almost all}\; \theta\,.
\end{equation}
\end{lemm}

\begin{proof}
This follows for the fact that convergence in $L_1$ implies pointwise convergence almost everywhere of a subsequence, \ie there exists a subsequence $\{\tau_n\}_{n\in\nset}\subset\nset$ such that
\begin{equation}
\label{eq2}
\|\pi_{n}-\Phi_n\|_1\to 0 \Rightarrow |\pi_{\tau_n}(\theta)-\Phi_{\tau_n}(\theta)|\to 0\quad p\text{-a.e.}
\end{equation}
Eq. \ref{eq1} follows from combining the Bernstein-von Mises theorem and Eq. \eqref{eq2}:
$$
\plim_{n\to\infty} \|\pi_n-\Phi_n(\thetaMLE,I(\theta_0)^{-1}/n)\|_1\overset{\proba_{\theta_0}}{=}0\Rightarrow \plim_{n\to\infty}|\pi_{\tau_n}(\theta)-\Phi_{\tau_n}(\theta)|\overset{\proba_{\theta_0}}{=}0\quad p\text{-a.e.}
$$
\end{proof}

\begin{coro}
\label{coro1}
There exists a subsequence $\{\tau_n\}_{n\in\nset}\in\nset$ such that
\begin{equation}
\label{eq3}
\underset{n\to\infty}{\plim}\left|\log\pi_{\tau_n}(\theta)-\log\Phi_{\tau_n}(\theta)\right|\overset{\proba_{\theta_0}}{=}0\,,\quad \text{for}\;p \text{-almost all}\;\theta\,.
\end{equation}
\end{coro}
\begin{proof}
Follows from Lemma \ref{lem1}, by continuity of the logarithm.
\end{proof}

\begin{lemm}
\label{lem2}
For any $U\in\Uset_n$, let $\theta\mapsto\Phi_U(\theta)$ be the pdf of $\norm(\theta_U^\ast,I(\theta_0)^{-1}/n)$ and $\Phi_{\rho n}$ be the pdf of $\norm(\theta_{\rho n}^\ast,I(\theta_0)^{-1}/\rho n)$ be the Bernstein-von Mises approximations of respectively $\pi(\,\cdot\,|\,Y_U)$ and $\pi(\cdot\,|\,Y_{1:\rho n})$ where $U\subset \Uset_n(Y_{1:\rho n})$. Then we have for all $\theta\in\Theta$
$$
\log\Phi_{\rho n}(\theta)-\rho\log\Phi_U(\theta)\leq (\rho-1)\log(2\pi)^{d/2}+\frac{\rho n}{2}\left\{\|\theta-\theta^\ast_U\|_{\theta_0}-
\|\theta-\theta^\ast\|_{\theta_0}\right\}\,,
$$
where for any $d$-squared symmetric matrix $M$, we have defined by $\|\cdot\|_M$ the norm associated to the scalar product $\pscal{u}{v}_M=u\T M v$.
\end{lemm}

\begin{proof}
This follows from straightforward algebra and noting that
$$
\log \rho n |I(\theta_0)|-\rho\log n|I(\theta_0)|\leq 0\,.
$$
\end{proof}

\section*{Acknowledgements}
The Insight Centre for Data Analytics is supported by Science Foundation Ireland under Grant Number SFI/12/RC/2289. Nial Friel's research was also supported by an Science Foundation Ireland grant: 12/IP/1424. Pierre Alquier's research was funded by Labex ECODEC (ANR - 11-LABEX-0047) and by the research programme New Challenges for New Data from LCL and GENES, hosted by the Fondation du Risque.

We thank the Associate Editor and two anonymous Referees for their contribution to this work.

\section*{References}
\bibliographystyle{elsarticle-harv}
\bibliography{biblio}

\begin{thebibliography}{40}
\expandafter\ifx\csname natexlab\endcsname\relax\def\natexlab#1{#1}\fi
\expandafter\ifx\csname url\endcsname\relax
  \def\url#1{\texttt{#1}}\fi
\expandafter\ifx\csname urlprefix\endcsname\relax\def\urlprefix{URL }\fi

\bibitem[{Allassonni{\`e}re et~al.(2007)Allassonni{\`e}re, Amit, and
  Trouv{\'e}}]{allassonniere2007towards}
Allassonni{\`e}re, S., Amit, Y., Trouv{\'e}, A., 2007. Towards a coherent
  statistical framework for dense deformable template estimation. Journal of
  the Royal Statistical Society: Series B (Statistical Methodology) 69~(1),
  3--29.

\bibitem[{Alquier et~al.(2016)Alquier, Friel, Everitt, and
  Boland}]{alquier2014noisy}
Alquier, P., Friel, N., Everitt, R., Boland, A., 2016. Noisy {M}onte {C}arlo:
  Convergence of {M}arkov chains with approximate transition kernels.
  Statistics and Computing 26~(1-2), 29--47.

\bibitem[{Andrieu and Roberts(2009)}]{andrieu2009pseudo}
Andrieu, C., Roberts, G.~O., 2009. The pseudo-marginal approach for efficient
  {M}onte {C}arlo computations. The Annals of Statistics, 697--725.

\bibitem[{Andrieu and Vihola(2015)}]{andrieu2012convergence}
Andrieu, C., Vihola, M., 2015. Convergence properties of pseudo-marginal markov
  chain monte carlo algorithms. The Annals of Applied Probability 25~(2),
  1030--1077.

\bibitem[{Banterle et~al.(2015)Banterle, Grazian, Lee, and
  Robert}]{banterle2015accelerating}
Banterle, M., Grazian, C., Lee, A., Robert, C.~P., 2015. Accelerating
  {Metropolis-Hastings} algorithms by delayed acceptance. arXiv preprint
  arXiv:1503.00996.

\bibitem[{Bardenet et~al.(2014)Bardenet, Doucet, and
  Holmes}]{bardenet2014towards}
Bardenet, R., Doucet, A., Holmes, C., 2014. Towards scaling up {Markov chain
  Monte Carlo}: an adaptive subsampling approach. In: ICML. pp. 405--413.

\bibitem[{Bardenet et~al.(2017)Bardenet, Doucet, and
  Holmes}]{bardenet2015markov}
Bardenet, R., Doucet, A., Holmes, C., 2017. On {M}arkov chain {M}onte {C}arlo
  methods for tall data. Journal of Machine Learning Research 18, 1--43.

\bibitem[{Bierkens et~al.(2018)Bierkens, Fearnhead, and
  Roberts}]{bierkens2016zig}
Bierkens, J., Fearnhead, P., Roberts, G., 2018. The {Zig-Zag} process and
  super-efficient sampling for {B}ayesian analysis of big data. The Annals of
  Statistics (to appear).

\bibitem[{Chib and Greenberg(1995)}]{chib1995understanding}
Chib, S., Greenberg, E., 1995. Understanding the {M}etropolis-{H}astings
  algorithm. The american statistician 49~(4), 327--335.

\bibitem[{Csill{\'e}ry et~al.(2010)Csill{\'e}ry, Blum, Gaggiotti, and
  Fran{\c{c}}ois}]{csillery2010approximate}
Csill{\'e}ry, K., Blum, M.~G., Gaggiotti, O.~E., Fran{\c{c}}ois, O., 2010.
  Approximate {B}ayesian computation ({ABC}) in practice. Trends in ecology \&
  evolution 25~(7), 410--418.

\bibitem[{Dalalyan(2017)}]{dalalyan2017further}
Dalalyan, A.~S., 2017. Further and stronger analogy between sampling and
  optimization: {L}angevin {M}onte {C}arlo and gradient descent. arXiv preprint
  arXiv:1704.04752.

\bibitem[{Douc et~al.(2004)Douc, Moulines, and
  Rosenthal}]{douc2004quantitative}
Douc, R., Moulines, E., Rosenthal, J.~S., 2004. Quantitative bounds on
  convergence of time-inhomogeneous {Markov} chains. The Annals of Applied
  Probability, 1643--1665.

\bibitem[{Fearnhead et~al.(2016)Fearnhead, Bierkens, Pollock, and
  Roberts}]{fearnhead2016piecewise}
Fearnhead, P., Bierkens, J., Pollock, M., Roberts, G.~O., 2016. Piecewise
  deterministic {M}arkov processes for continuous-time {M}onte {C}arlo. arXiv
  preprint arXiv:1611.07873.

\bibitem[{Fearnhead and Prangle(2012)}]{fearnhead2012constructing}
Fearnhead, P., Prangle, D., 2012. Constructing summary statistics for
  approximate {B}ayesian computation: semi-automatic approximate {B}ayesian
  computation. Journal of the Royal Statistical Society: Series B (Statistical
  Methodology) 74~(3), 419--474.

\bibitem[{Geyer and Thompson(1995)}]{geyer1995annealing}
Geyer, C.~J., Thompson, E.~A., 1995. Annealing {M}arkov chain {M}onte {C}arlo
  with applications to ancestral inference. Journal of the American Statistical
  Association 90~(431), 909--920.

\bibitem[{Haario et~al.(2001)Haario, Saksman, and
  Tamminen}]{haario2001adaptive}
Haario, H., Saksman, E., Tamminen, J., 2001. An adaptive {M}etropolis
  algorithm. Bernoulli, 223--242.

\bibitem[{Hobert and Robert(2004)}]{hobert2004mixture}
Hobert, J.~P., Robert, C.~P., 2004. A mixture representation of $\pi$ with
  applications in {Markov chain Monte Carlo} and perfect sampling. The Annals
  of Applied Probability, 1295--1305.

\bibitem[{Huggins and Zou(2016)}]{huggins2016}
Huggins, J., Zou, J., 2016. Quantifying the accuracy of approximate diffusions
  and {M}arkov chains. In: Proceedings of the 20th International Conference on
  Artifical Intelligence and Statistics, PLMR. Vol.~54. pp. 382--391.

\bibitem[{Jacob et~al.(2015)Jacob, Thiery, et~al.}]{jacob2015nonnegative}
Jacob, P.~E., Thiery, A.~H., et~al., 2015. On nonnegative unbiased estimators.
  The Annals of Statistics 43~(2), 769--784.

\bibitem[{Johndrow and Mattingly(2017)}]{johndrow2017error}
Johndrow, J.~E., Mattingly, J.~C., 2017. Error bounds for approximations of
  {Markov} chains. arXiv preprint arXiv:1711.05382.

\bibitem[{Johndrow et~al.(2015)Johndrow, Mattingly, Mukherjee, and
  Dunson}]{johndrow2015approximations}
Johndrow, J.~E., Mattingly, J.~C., Mukherjee, S., Dunson, D., 2015.
  Approximations of {Markov Chains and B}ayesian inference. arXiv preprint
  arXiv:1508.03387.

\bibitem[{Korattikara et~al.(2014)Korattikara, Chen, and
  Welling}]{korattikara2013austerity}
Korattikara, A., Chen, Y., Welling, M., 2014. Austerity in {MCMC} land: Cutting
  the {M}etropolis-{H}astings budget. In: Proceedings of the 31st International
  Conference on Machine Learning.

\bibitem[{Le~Cam(1953)}]{le1953some}
Le~Cam, L., 1953. On some asymptotic properties of maximum likelihood estimates
  and related {B}ayes' estimates. Univ. Calif. Publ. in Statist. 1, 277--330.

\bibitem[{Le~Cam(1986)}]{le2012asymptotic}
Le~Cam, L., 1986. Asymptotic methods in statistical decision theory. Springer
  Science \& Business Media.

\bibitem[{Maclaurin and Adams(2015)}]{maclaurin2014firefly}
Maclaurin, D., Adams, R.~P., 2015. Firefly {Monte Carlo: Exact MCMC} with
  subsets of data. In: Twenty-Fourth International Joint Conference on
  Artificial Intelligence.

\bibitem[{Marin et~al.(2012)Marin, Pudlo, Robert, and
  Ryder}]{marin2012approximate}
Marin, J.-M., Pudlo, P., Robert, C.~P., Ryder, R.~J., 2012. Approximate
  {B}ayesian computational methods. Statistics and Computing 22~(6),
  1167--1180.

\bibitem[{Medina-Aguayo et~al.(2016)Medina-Aguayo, Lee, and
  Roberts}]{medina2016stability}
Medina-Aguayo, F.~J., Lee, A., Roberts, G.~O., 2016. Stability of noisy
  {Metropolis--Hastings}. Statistics and Computing 26~(6), 1187--1211.

\bibitem[{{M}etropolis et~al.(1953){M}etropolis, Rosenbluth, Rosenbluth,
  Teller, and Teller}]{metropolis1953equation}
{M}etropolis, N., Rosenbluth, A.~W., Rosenbluth, M.~N., Teller, A.~H., Teller,
  E., 1953. Equation of state calculations by fast computing machines. The
  journal of chemical physics 21~(6), 1087--1092.

\bibitem[{Meyn and Tweedie(2009)}]{meyn2009markov}
Meyn, S.~P., Tweedie, R.~L., 2009. Markov chains and stochastic stability.
  Cambridge University Press.

\bibitem[{Mitrophanov(2005)}]{mitrophanov2005sensitivity}
Mitrophanov, A.~Y., 2005. Sensitivity and convergence of uniformly ergodic
  {M}arkov chains. Journal of Applied Probability, 1003--1014.

\bibitem[{Nunes and Balding(2010)}]{nunes2010optimal}
Nunes, M.~A., Balding, D.~J., 2010. On optimal selection of summary statistics
  for approximate {B}ayesian computation. Statistical applications in genetics
  and molecular biology 9~(1).

\bibitem[{Pollock et~al.(2016)Pollock, Fearnhead, Johansen, and
  Roberts}]{pollock2016scalable}
Pollock, M., Fearnhead, P., Johansen, A.~M., Roberts, G.~O., 2016. The scalable
  {L}angevin exact algorithm: {B}ayesian inference for big data. arXiv preprint
  arXiv:1609.03436.

\bibitem[{Pritchard et~al.(1999)Pritchard, Seielstad, Perez-Lezaun, and
  Feldman}]{pritchard1999population}
Pritchard, J.~K., Seielstad, M.~T., Perez-Lezaun, A., Feldman, M.~W., 1999.
  Population growth of human {Y} chromosomes: a study of {Y} chromosome
  microsatellites. Molecular biology and evolution 16~(12), 1791--1798.

\bibitem[{Quiroz et~al.(2015)Quiroz, Villani, and Kohn}]{quiroz2015speeding}
Quiroz, M., Villani, M., Kohn, R., 2015. Speeding up {MCMC} by efficient data
  subsampling. Riksbank Research Paper Series~(121).

\bibitem[{Quiroz et~al.(2016)Quiroz, Villani, and Kohn}]{quiroz2016exact}
Quiroz, M., Villani, M., Kohn, R., 2016. Exact subsampling {MCMC}. arXiv
  preprint arXiv:1603.08232.

\bibitem[{Roberts et~al.(2001)Roberts, Rosenthal, et~al.}]{roberts2001optimal}
Roberts, G.~O., Rosenthal, J.~S., et~al., 2001. Optimal scaling for various
  {Metropolis-Hastings} algorithms. Statistical science 16~(4), 351--367.

\bibitem[{Rudolf and Schweizer(2018)}]{rudolf2015perturbation}
Rudolf, D., Schweizer, N., 2018. Perturbation theory for {M}arkov chains via
  {W}asserstein distance. Bernoulli 24~(4A), 2610--2639.

\bibitem[{Van~der Vaart(2000)}]{van2000asymptotic}
Van~der Vaart, A.~W., 2000. Asymptotic statistics. Vol.~3. Cambridge university
  press.

\bibitem[{Welling and Teh(2011)}]{welling2011bayesian}
Welling, M., Teh, Y.~W., 2011. Bayesian learning via stochastic gradient
  {L}angevin dynamics. In: Proceedings of the 28th International Conference on
  Machine Learning (ICML-11). pp. 681--688.

\bibitem[{Wilkinson(2013)}]{wilkinson2013approximate}
Wilkinson, R.~D., 2013. Approximate {B}ayesian computation {(ABC)} gives exact
  results under the assumption of model error. Statistical applications in
  genetics and molecular biology 12~(2), 129--141.

\end{thebibliography}

\end{document}